\documentclass[10pt,letterpaper]{article}
\pdfoutput=1

\usepackage[margin=1in]{geometry}
\usepackage[utf8]{inputenc}
\usepackage[T1]{fontenc}
\usepackage[nopatch=footnote]{microtype}
\usepackage{amsfonts}
\usepackage{amssymb}
\usepackage{amsmath}
\usepackage{amsthm}
\usepackage{booktabs}
\usepackage[inline]{enumitem}
\usepackage{setspace}
\usepackage[hidelinks,bookmarks,bookmarksopen,bookmarksnumbered]{hyperref}
\usepackage[capitalize]{cleveref}
\usepackage{mathtools}
\usepackage{tikz}
\usepackage[font=small,labelfont=bf]{caption}
\usepackage{titlesec}
\usepackage{thmtools}
\usepackage{tablefootnote}
\usepackage{soul}
\usepackage[ruled,noend,linesnumbered]{algorithm2e}
\usepackage{framed}

\usetikzlibrary{decorations.pathreplacing}
\usetikzlibrary{patterns}

\newtheorem{theorem}{Theorem}[section]

\newtheorem{lemma}[theorem]{Lemma}

\newtheorem{observation}[theorem]{Observation}
\newtheorem{proposition}[theorem]{Proposition}

\theoremstyle{definition}
\newtheorem{definition}[theorem]{Definition}
\newtheorem{conjecture}[theorem]{Conjecture}
\theoremstyle{remark}
\newtheorem{remark}[theorem]{Remark}
\definecolor{mypink}{RGB}{199, 21 133}
\hypersetup{
  colorlinks,
  urlcolor = blue,
  linkcolor = mypink,
  citecolor = blue
}

\setul{2pt}{0.6pt}

\setlist[enumerate]{nosep, topsep=1ex}
\setlist[itemize]{nosep, topsep=1ex}
\setlist[description]{nosep,topsep=1ex}

\allowdisplaybreaks

\apptocmd{\sloppy}{\hbadness 10000\relax}{}{}

\let\oldabstract\abstract
\let\oldendabstract\endabstract
\makeatletter
\renewenvironment{abstract}
{%
  {\list{}{\addtolength{\leftmargin}{0.0em}%
    \listparindent 1.5em%
     \itemindent    \listparindent%
     \rightmargin   \leftmargin%
     \parsep        \z@ \@plus\p@}%
     \item\relax}%
  {\endlist}%
\oldabstract}
{\oldendabstract}
\makeatother

\SetCommentSty{textit}
\SetKwInOut{Input}{Input}
\SetKwInOut{Output}{Output}

\makeatletter
\patchcmd\algocf@Vline{\vrule}{\vrule \kern-0.4pt}{}{}
\patchcmd\algocf@Vsline{\vrule}{\vrule \kern-0.4pt}{}{}
\makeatother

\newcommand{\bigO}{\mathcal{O}}

\DeclareMathOperator{\polylog}{polylog}
\newcommand{\ceil}[1]{\left\lceil #1 \right\rceil}

\newcommand{\dd}{\mathinner{.\,.}}
\newcommand{\probname}[1]{\text{\sc #1}}

\newcommand{\Z}{\mathbb{Z}}
\newcommand{\Zn}{\Z_{\ge 0}}
\newcommand{\Zp}{\Z_{>0}}

\newcommand{\BinaryAlphabet}{\{{\tt 0},{\tt 1}\}}
\newcommand{\emptystring}{\varepsilon}
\newcommand{\Text}{T}
\newcommand{\Occ}[2]{\mathrm{Occ}(#1, #2)}
\newcommand{\zero}{{\tt 0}}
\newcommand{\one}{{\tt 1}}

\newcommand{\Rows}[1]{\mathrm{rows}(#1)}
\newcommand{\Cols}[1]{\mathrm{cols}(#1)}
\newcommand{\vconcat}{\;\rotatebox[origin=c]{-90}{$\circleddash$}\;}
\newcommand{\hconcat}{\circleddash}

\newcommand{\Rhs}[2]{\mathrm{rhs}_{#1}(#2)}
\newcommand{\Exp}[2]{\mathrm{exp}_{#1}(#2)}
\newcommand{\Lang}[1]{L(#1)}

\newcommand{\Hook}[4]{\mathrm{hook}_{#1}(#2,#3,#4)}
\newcommand{\Offset}[4]{\mathrm{offset}_{#1}(#2,#3,#4)}
\newcommand{\Access}[4]{\mathrm{access}_{#1}(#2,#3,#4)}
\newcommand{\LeftMap}[4]{\mathrm{left}\mbox{-}\mathrm{map}_{#1}(#2,#3,#4)}
\newcommand{\RightMap}[4]{\mathrm{right}\mbox{-}\mathrm{map}_{#1}(#2,#3,#4)}
\newcommand{\DirLeft}{{\tt L}}
\newcommand{\DirRight}{{\tt R}}

\newcommand{\HookTwoDim}[6]{\mathrm{hook}_{#1}(#2,#3,#4,#5,#6)}
\newcommand{\OffsetTwoDim}[6]{\mathrm{offset}_{#1}(#2,#3,#4,#5,#6)}
\newcommand{\AccessTwoDim}[6]{\mathrm{access}_{#1}(#2,#3,#4,#5,#6)}
\newcommand{\TopLeftMap}[6]{\mathrm{top}\mbox{-}\mathrm{left}\mbox{-}\mathrm{map}_{#1}(#2,#3,#4,#5,#6)}
\newcommand{\TopRightMap}[6]{\mathrm{top}\mbox{-}\mathrm{right}\mbox{-}\mathrm{map}_{#1}(#2,#3,#4,#5,#6)}
\newcommand{\BottomLeftMap}[6]{\mathrm{bottom}\mbox{-}\mathrm{left}\mbox{-}\mathrm{map}_{#1}(#2,#3,#4,#5,#6)}
\newcommand{\BottomRightMap}[6]{\mathrm{bottom}\mbox{-}\mathrm{right}\mbox{-}\mathrm{map}_{#1}(#2,#3,#4,#5,#6)}
\newcommand{\DirTop}{{\tt T}}
\newcommand{\DirBottom}{{\tt B}}

\newcommand{\Occurs}[4]{\mathrm{occurs}_{#1}(#2,#3,#4)}
\newcommand{\Rank}[3]{\mathrm{rank}_{#1}(#2,#3)}

\newcommand{\SumQuery}[5]{\mathrm{sum}_{#1}(#2,#3,#4,#5)}
\newcommand{\LineSumQuery}[4]{\mathrm{line\mbox{-}sum}_{#1}(#2,#3,#4)}
\newcommand{\EqualQuery}[7]{\mathrm{equal}_{#1}(#2,#3,#4,#5,#6,#7)}
\newcommand{\SquareLCE}[5]{\mathrm{sq\mbox{-}lce}_{#1}(#2,#3,#4,#5)}
\newcommand{\LineLCE}[6]{\mathrm{line\mbox{-}lce}_{#1}(#2,#3,#4,#5,#6)}
\newcommand{\AllZero}[5]{\mathrm{all\mbox{-}zero}_{#1}(#2,#3,#4,#4)}
\newcommand{\SquareAllZero}[4]{\mathrm{sq\mbox{-}all\mbox{-}zero}_{#1}(#2,#3,#4)}

\newcommand{\MarkChar}[2]{\mathrm{MarkChar}_{#1,#2}}
\newcommand{\MarkAllChars}[2]{\mathrm{MarkAllChars}_{#1,#2}}
\newcommand{\ExtMarkAllChars}[2]{\mathrm{ExtMarkAllChars}_{#1,#2}}

\newcommand{\LevelId}{p}

\begin{document}

\title{Optimal Random Access and Conditional Lower Bounds\\ for 2D Compressed Strings}

\author{
  \large Rajat De\\[-0.3ex]
  \normalsize Stony Brook University,\\[-0.3ex]
  \normalsize Stony Brook, NY, USA\\[-0.3ex]
  \normalsize \texttt{rde@cs.stonybrook.edu}
  \and
  \large Dominik Kempa\thanks{Partially funded by the
  NSF CAREER Award 2337891 and the Simons Foundation
  Junior Faculty Fellowship.}\\[-0.3ex]
  \normalsize Stony Brook University,\\[-0.3ex]
  \normalsize Stony Brook, NY, USA\\[-0.3ex]
  \normalsize \texttt{kempa@cs.stonybrook.edu}
}

\date{\vspace{-0.5cm}}
\maketitle

\begin{abstract}
  Compressed indexing is a powerful technique that enables efficient
  querying over data stored in compressed form, significantly reducing
  memory usage and often accelerating computation. While extensive
  progress has been made for one-dimensional strings, many real-world
  datasets (such as images, maps, and adjacency matrices) are inherently
  two-dimensional and highly compressible. Unfortunately, naively
  applying 1D techniques to 2D data leads to suboptimal results, as
  fundamental structural repetition is lost during linearization. This
  motivates the development of native 2D compressed indexing schemes
  that preserve both compression and query efficiency.

  We present three main contributions that advance the theory of
  compressed indexing for 2D strings:
  \begin{itemize}
  \item We design the first data structure that supports
    \emph{optimal-time random access} to a 2D string compressed
    by a 2D grammar. Specifically, for a 2D string $T\in\Sigma^{r\times c}$
    compressed by a 2D grammar $G$ and any constant $\epsilon>0$,
    we achieve $\bigO\left(\frac{\log n}{\log\log n}\right)$ query time and
    $\bigO(|G| \cdot \log^{2+\epsilon}n)$ space, where $n=\max(r,c)$.
  \item We prove conditional lower bounds for \emph{pattern matching}
    over 2D-grammar compressed strings. Assuming the Orthogonal
    Vectors Conjecture, no algorithm can solve this problem in time
    $\bigO(|G|^{2-\epsilon}\cdot |P|^{\bigO(1)})$ for any $\epsilon >
    0$, demonstrating a separation from the 1D case, where optimal
    solutions exist.
  \item We show that \emph{several fundamental 2D queries}, such as the 2D
    longest common extension, rectangle sum, and equality, cannot be
    supported efficiently under hardness assumptions for rank
    and symbol occurrence queries on 1D grammar-compressed strings. This is
    the first evidence connecting the complexity of 2D compressed
    indexing to long-standing open problems in the 1D setting.
  \end{itemize}

  In summary, our results provide both algorithmic advances and
  conditional hardness results for 2D compressed indexing, narrowing
  the gap between one- and two-dimensional settings and identifying
  critical barriers that must be overcome for further progress.
\end{abstract}

\section{Introduction}\label{sec:intro}

We address a fundamental question: how can we store a 2D array/2D
string/matrix $M \in \{\zero, \one\}^{r \times c}$ (or, more generally,
any $M \in \Sigma^{r \times c}$, where $\Sigma = [0 \dd \sigma)$ for
some $\sigma \geq 2$) in small space while still supporting efficient
queries on $M$? The most basic query of interest is
\emph{random access}: given any $(i, j) \in [1 \dd r] \times [1 \dd c]$,
return $M[i, j]$. The problem of representing 2D strings efficiently
dates back 40 years to the seminal work of Lempel and Ziv, who studied
2D compression~\cite{lz2d}. However, their approach renders the data
unusable in compressed form: once compressed, the original input
cannot be accessed without full decompression.

This property of classical compression algorithms has, in recent years,
led to data compression being used in a new way. The field of
\emph{computation over compressed data}\footnote{Throughout this
  paper, by \emph{compression} we always mean
  \emph{lossless compression}.}
combines elements of combinatorial pattern
matching, information theory, and data compression to design
algorithms and data structures that can operate directly on data
in compressed form~\cite{navarrobook,NavarroMeasures,NavarroIndexes}. This can
reduce the time and space for computation by orders of magnitude and
has found applications in bioinformatics, information retrieval, and
databases~\cite{rindex,pangenome,rossi2022finding,phoni,FuentesSepulvedaGNRS25,Navarro23,ring24}.

The vast majority of the above research focuses on strings, i.e.,
one-dimensional data. Recently, however, it has been observed that
many forms of modern high-dimensional data are also highly compressible,
including images (formats such as $\text{png}$, $\text{bmp}$,
$\text{gif}$, $\text{tiff}$, and $\text{webp}$ all use lossless
compression)~\cite{blocktree2d,moffat97}; graphs (e.g., adjacency lists
of web and social networks)~\cite{BrisaboaLN14,Arroyuelo24,ArroyueloGH0R24};
spatial datasets such as maps~\cite{CarfagnaMRSU24};
trajectories~\cite{BrisaboaGNP16}; and
databases~\cite{Navarro23,FuentesSepulvedaGNRS25,ring24}.
Unfortunately, however, applying existing 1D compression techniques 
directly to
high-dimensional data is not effective:
\begin{itemize}
\item When \emph{linearizing}~\cite{lz2d} a 2D array
  row by row or column by column, the two-dimensional repetition is lost.
  Carfagna et al.~\cite{CarfagnaMRSU24} describe explicit families of
  inputs where 1D compression of a linearized string is
  \emph{exponentially} worse than native 2D compression.
\item Even more advanced linearizations, such as plane-filling curves,
  including the Peano--Hilbert curve (see, e.g.,~\cite{lz2d}),
  can be less efficient than native
  2D compression~\cite[Proposition~14]{CarfagnaMRSU24}.
\end{itemize}

At this point, it is natural to ask: can existing 2D compression
methods be \emph{augmented with support for queries} over compressed
data? The state of the art in computation over 2D-compressed data
can be summarized as follows:
\begin{itemize}
\item On the one hand, most compression frameworks and
  repetitiveness measures for 1D strings have been generalized to 2D
  strings. Specifically, the classical Lempel--Ziv (LZ77) algorithm for
  strings~\cite{LZ77} was generalized to 2D strings as early as
  1986~\cite{lz2d}. Other compression methods and repetitiveness
  measures in 1D, including (run-length) grammar
  compression~\cite{Charikar05,NishimotoMFCS}, bidirectional
  parsing~\cite{macro}, string attractors~\cite{attractors}, and
  substring complexity~\cite{delta}, have been generalized to
  2D strings in~\cite{CarfagnaM24,CarfagnaMRSU24}.
\item However, when it comes to \emph{computation over compressed
  data}, the situation is more nuanced. For 1D strings, random
  access can be supported in $\bigO(\log n)$ time on virtually every
  known compression method, with only a $\bigO(\log n)$-factor
  overhead in space
  usage~\cite{BLRSRW15,blocktree,attractors,balancing,KempaS22,delta}.
  Subsequently, random access in 1D has been improved to
  $\bigO(\tfrac{\log n}{\log \log n})$ time at the cost of a
  $\bigO(\log^{\epsilon} n)$-factor increase in
  space~\cite{BelazzouguiCPT15,blocktree,balancing,attractors}. This time bound
  was proved optimal in~\cite{VerbinY13} (and further shown to hold even
  for many grammar compressors with provably weaker compression ratios~\cite{boosting}).
  Furthermore, many 1D
  compressed representations can be augmented at a similar cost to
  support more advanced queries, including pattern matching
  queries~\cite{ClaudeN11,kreft2013compressing,ClaudeNP21,GagieGKNP12,%
    GagieGKNP14,BilleEGV18,navarro201941,resolution,ChristiansenEKN21,KociumakaNO24},
  longest common extension (LCE) queries~\cite{tomohiro-lce,KempaK23},
  and even the powerful suffix array functionality~\cite{rindex,KempaK23}. The complexity of the majority of those
  more general queries has also been settled: \cite{dichotomy}
  shows that in compressed space, central string queries have optimal complexity
  of either $\Theta(\tfrac{\log n}{\log \log n})$ (this holds for, e.g.,
  suffix array, inverse suffix array, LCP array, and LCE queries), or $\Theta(\log \log n)$
  (this includes BWT, lexicographical predecessor/successor, or permuted LCP array queries).
  A notable exception in 1D compressed indexing is support for rank/select
  queries,\footnote{Consider a string $S \in \Sigma^{m}$.
    Given any index $i \in [0 \dd m]$ and a symbol $c \in \Sigma$, a
    \emph{rank query} returns $\Rank{S}{i}{c} =
    |\{j \in [1 \dd i] : S[j] = c\}|$. Given any $c \in \Sigma$ and
    $i \in [1 \dd \Rank{S}{m}{c}]$, a \emph{select query} returns the
    $i$th smallest element of the set $\{j \in [1 \dd m] : S[j] = c\}$.}
  where the best solutions still require space quadratic in the
  compressed size of the input (in the worst case of a large
  alphabet)~\cite{BelazzouguiCPT15,Prezza19}. Overall, however, most
  queries on 1D compressed data can be supported efficiently and in
  small space, enabling complex algorithms to run directly
  on compressed text.
  
  \vspace{1ex}
  In contrast, much less is known about supporting queries on
  2D-compressed strings. One of the recent results making progress in
  this area is~\cite{blocktree2d}. Given any
  $M \in \Sigma^{n \times n}$, the data structure
  in~\cite{blocktree2d} implements random access queries to $M$ in
  $\bigO(\log n)$ time and uses
  $\bigO((\delta_{\rm 2D}(M) + \sqrt{n \cdot \delta_{\rm 2D}(M)}) \log
  (n / \delta_{\rm 2D}(M)))$ space, where $\delta_{\rm 2D}(M)$ is a
  repetitiveness measure for 2D matrices~\cite{CarfagnaM24}. Observe
  that the size of this structure is always $\Omega(\sqrt{n})$,
  potentially incurring a significant overhead relative to the
  compressibility of $M$. The question of whether this space can be
  reduced further was recently asked by Carfagna et
  al.~\cite{CarfagnaMRSU24}, who showed how to support
  random access queries on a 2D string $M \in \Sigma^{r \times c}$
  compressed with a 2D grammar $G$ of size $g$ in $\bigO(g)$ space and
  $\bigO(\log n)$ time (where $n = \max(r,c)$). To our knowledge,
  this is the first compressed data structure in 2D that achieves
  efficient queries in space that is near-linear in the size of the
  compressed 2D text (in fact, exactly linear, i.e., optimal).
  Another recent work by Ganguly et al.~\cite{Ganguly2025} showed how
  to support 2D longest common extension (2D LCE) queries (which, given
  two positions in the input matrix $M \in \Sigma^{n \times n}$,
  compute the largest matching square submatrix with top-left corners
  at the specified positions) in
  $\bigO((n^{5/3} + n^{8/5} \cdot \delta_{\rm 2D}(M)) \log n)$ space and
  $\bigO(\log n)$ time.
\end{itemize}

To sum up, with the exception of rank/select queries, in the 1D case,
even the most powerful queries (e.g.,~\cite{rindex,KempaK23}) add only a
$\polylog(n)$ factor in space on top of the compressed size for
virtually all commonly used compression algorithms (including
Lempel--Ziv and grammars) and known repetitiveness measures. Moreover,
in the 1D setting, both the optimal space~\cite{delta} and the
optimal query times are known for most central string queries~\cite{hierarchy,VerbinY13}.
In contrast, the field of 2D compressed data structures is currently more
fragmented: queries like random access admit small-space data
structures supporting efficient queries~\cite{CarfagnaMRSU24}, whereas
queries like LCE currently require $\Omega(n^{\epsilon})$ space (for
some constant $\epsilon > 0$), even for highly compressible matrices
(e.g., when $\delta_{\rm 2D}(M) = \bigO(1)$). We thus ask:

\begin{center}
  \emph{What is the complexity of fundamental queries on 2D strings in compressed space?}
\end{center}

\paragraph{Our Results}

We present three distinct results that shed new light on the field of
compressed indexing for 2D strings. We begin with a positive
result, showing that for every
2D grammar $G$ of size $|G| = g$ representing a 2D string
$M \in \Sigma^{r \times c}$, and every parameter $\tau \geq 2$, there
exists a data structure of size
$\bigO(g \cdot \tau^2 \cdot \log_{\tau}^2 n)$, where $n = \max(r,c)$,
that returns $M[i,j]$ in $\bigO(\log_{\tau} n)$ time for any
$(i,j) \in [1 \dd r] \times [1 \dd c]$. By setting
$\tau = \log^{\epsilon/2} n$ (for any constant $\epsilon > 0$), we
obtain a data structure of size
$\bigO(g \cdot \log^{2 + \epsilon} n)$ that supports random access to $M$ in
$\bigO(\tfrac{\log n}{\log \log n})$ time.
By a simple generalization of the lower bound of Verbin and
Yu~\cite{VerbinY13}, this is the best possible query time for any data
structure that uses $\bigO(g \polylog n)$ space.\footnote{This follows
  because any 1D SLP is a special case of a 2D SLP; thus, random
  access on a 2D SLP in time $o(\tfrac{\log n}{\log \log n})$ would
  immediately imply the same bound for 1D SLPs,
  contradicting~\cite{VerbinY13}.} Hence, we obtain the first data
structure supporting \emph{optimal-time random access} on
2D-compressed data:

\begin{restatable}{theorem}{thtwodopt}\label{th:2d-opt}
  For every 2D SLP $G$ of size
  $|G| = g$ representing a 2D string $T \in \Sigma^{r \times c}$
  (see \cref{sec:prelim-2d-grammars}),
  and any constant $\epsilon > 0$,
  there exists a data structure of size
  $\bigO(g \cdot \log^{2+\epsilon} n)$, where $n = \max(r, c)$,
  that, given any
  $(i,j) \in [1 \dd r] \times [1 \dd c]$,
  returns
  $T[i,j]$
  in
  $\bigO(\tfrac{\log n}{\log \log n})$ time.
\end{restatable}

\paragraph{Hardness Results}

Our next two results concern the hardness of computation over
2D-compressed data.

First, we show that, conditioned on the hardness of the widely
used Orthogonal Vectors (OV) Conjecture, there are no near-linear-time
algorithms for pattern matching over 2D-compressed data.
Specifically, in \cref{sec:hardness-pattern-matching}, we prove:

\begin{restatable}{theorem}{twodpatternmatchinghardness}\label{th:2d-pattern-matching-hardness}
  Assuming the Orthogonal Vectors Conjecture (\cref{con:ov}), there is
  no algorithm that, given a pattern
  $P \in \BinaryAlphabet^{1 \times p}$ and a 2D SLP $G$ encoding a 2D
  text $T$ over the alphabet $\BinaryAlphabet$, determines whether $P$
  occurs in $T$ (see \cref{sec:2d-pattern-matching-problem-def}) in
  $\bigO(|G|^{2-\epsilon} \cdot p^{\bigO(1)})$ time, for any constant
  $\epsilon > 0$.
\end{restatable}

This result reveals a key distinction between pattern matching over
texts compressed by a 1D versus a 2D grammar, since in the 1D case the
problem can be solved in optimal time $\bigO(|G| + |P|)$~\cite{GanardiG22}.

Our second set of results concerns data structures---specifically, the
hardness of compressed indexing for queries over texts compressed with
2D grammars. We establish a series of reductions showing that,
unless certain simpler queries on 1D grammars can be solved efficiently,
a range of fundamental problems on 2D grammars (including 2D LCE
queries recently studied by Ganguly et al.~\cite{Ganguly2025}) cannot
be solved efficiently either. By \emph{efficiently}, we mean achieving
$\bigO(\log^{\bigO(1)} n)$ query time while incurring only a
$\bigO(\log^{\bigO(1)} n)$ space overhead beyond the size of the input
1D or 2D grammar encoding the text $\Text \in \Sigma^{n}$. To our
knowledge, this is the first result tying the hardness of 2D
compressed queries to that of 1D compressed queries.

The two problems on 1D grammar-compressed text that form the basis of
our hardness arguments are the related problems of
\emph{rank queries} and \emph{character occurrence queries}. Let
$T \in \Sigma^{n}$.
\begin{itemize}
\item Given any $j \in [0 \dd n]$ and $c \in \Sigma$,
  a \emph{rank query},
  denoted $\Rank{T}{j}{c}$,
  asks for the number of occurrences of $c$ in $T[1 \dd j]$;
  see \cref{def:rank}.
\item Given any $b, e \in [0 \dd n]$ and $c \in \Sigma$,
  a \emph{symbol occurrence query},
  denoted $\Occurs{T}{b}{e}{c}$,
  asks to return a value in $\{0,1\}$ that equals one if and only if
  $c$ occurs in $T(b \dd e]$;
  see \cref{def:symbol-occ}.
\end{itemize}

Currently, no data structure is known that, given any 1D SLP $G$
representing a string $T \in \Sigma^{n}$ (where $\Sigma = [0 \dd \sigma)$
and $\sigma = n^{\bigO(1)}$), uses $\bigO(|G| \cdot \log^{\bigO(1)} n)$
space and answers either of the above query types in
$\bigO(\log^{\bigO(1)} n)$ time. The best-known
solutions~\cite{BelazzouguiCPT15} either answer queries in
$\bigO(\log n)$ time using $\bigO(|G| \cdot |\Sigma|)$ space (which, in
the worst case of $|\Sigma| = |G|$, is $\bigO(|G|^2)$),\footnote{None
  of the 1D compressed indexes for other queries (such as random access,
  LCE, or suffix array
  queries)~\cite{rindex,BLRSRW15,blocktree,attractors,balancing,KempaS22,KempaK23,delta}
  have such a dependency on the alphabet size; that is, they all work
  even for large alphabets.} or in $\bigO(|G|)$ time using $\bigO(|G|)$
space. Furthermore, in~\cite{BelazzouguiCPT15}, the authors provide
evidence that any significant improvement on these trade-offs would
yield breakthrough results in graph algorithms.

Using the above problems as the basis of our hardness assumptions, in this paper
we prove that the following problems must also be hard on a 2D text
$\Text$ compressed with 2D grammars (see
\cref{def:2d-integer-problems,def:2d-general-problems}):
\begin{itemize}
\item \emph{Sum query:} Compute the sum of symbols in a given subrectangle.
\item \emph{Line sum query:} Compute the sum of symbols in a given subrectangle of height one.
\item \emph{All-zero query:} Check if all symbols in a given subrectangle are zero.
\item \emph{Square all-zero query:} Check if all symbols in a given subsquare are zero.
\item \emph{Equality query:} Check if two given subrectangles are equal.
\item \emph{Line LCE query:} Find the maximal $t$ such that
  $T[r \dd r + \ell)[c \dd c + t) = T[r' \dd r' + \ell)[c' \dd c' + t)$.
\item \emph{Square LCE query:} Find the maximal $t$ such that
  $T[r \dd r + t)[c \dd c + t) = T[r' \dd r' + t)[c' \dd c' + t)$.
\end{itemize}

\noindent
In \cref{sec:hardness-data-structures}, we prove the following results:

\begin{restatable}{theorem}{reducefromrank}\label{th:reduce-from-rank}
  Assume that, for every 2D SLP $G = (V_{l}, V_{h}, V_{v}, \Sigma, R, S)$,
  where $\Sigma = \BinaryAlphabet$, representing a 2D string
  $T \in \BinaryAlphabet^{r \times c}$, there exists a data structure of
  size $\bigO(|G| \cdot \log^{\bigO(1)} n)$ that answers one of the following
  query types:
  \begin{itemize}
  \item line sum queries (\cref{def:2d-integer-problems}),
  \item sum queries (\cref{def:2d-integer-problems}),
  \end{itemize}
  on $T$ in $\bigO(\log^{\bigO(1)} n)$ time, where $n = \max(r, c)$.
  Then, for every 1D SLP $\hat{G} = (\hat{V}, \hat{\Sigma}, \hat{R},
  \hat{S})$ representing $\hat{T} \in \hat{\Sigma}^{\hat{n}}$, where
  $\hat{\Sigma} = [0 \dd \hat{\sigma})$ and
  $\hat{\sigma} = \hat{n}^{\bigO(1)}$, there exists a data structure
  of size $\bigO(|\hat{G}| \cdot \log^{\bigO(1)} \hat{n})$ that
  answers rank queries (\cref{def:rank}) on $\hat{T}$ in
  $\bigO(\log^{\bigO(1)} \hat{n})$ time.
\end{restatable}

\begin{restatable}{theorem}{reducefromsymbolocc}\label{th:reduce-from-symbol-occ}
  Assume that, for every 2D SLP $G = (V_{l}, V_{h}, V_{v}, \Sigma, R, S)$,
  where $\Sigma = \BinaryAlphabet$, representing a 2D string $T \in \Sigma^{r \times c}$,
  there exists a data structure of size $\bigO(|G| \cdot \log^{\bigO(1)} n)$
  that answers one of the following query types:
  \begin{itemize}
  \item sum queries (\cref{def:2d-integer-problems}),
  \item all-zero queries (\cref{def:2d-integer-problems}),
  \item square all-zero queries (\cref{def:2d-integer-problems}),
  \item equality queries (\cref{def:2d-general-problems}),
  \item square LCE queries (\cref{def:2d-general-problems}),
  \item line LCE queries (\cref{def:2d-general-problems}),
  \end{itemize}
  on $T$ in $\bigO(\log^{\bigO(1)} n)$ time, where $n = \max(r, c)$.
  Then, for every 1D SLP $\hat{G} = (\hat{V}, \hat{\Sigma}, \hat{R}, \hat{S})$
  representing $\hat{T} \in \hat{\Sigma}^{\hat{n}}$, where
  $\hat{\Sigma} = [0 \dd \hat{\sigma})$ and $\hat{\sigma} = \hat{n}^{\bigO(1)}$,
  there exists a data structure of size $\bigO(|\hat{G}| \cdot \log^{\bigO(1)} \hat{n})$ 
  that answers symbol occurrence queries (\cref{def:symbol-occ}) on $\hat{T}$
  in $\bigO(\log^{\bigO(1)} \hat{n})$ time.
\end{restatable}

\paragraph{Related Work}

Research on computation over compressed data has developed along
several parallel directions. A large body of work has focused on
compressed indexes capable of strong compression of the input text by
exploiting redundancy arising from repeated substrings. Such
redundancy is captured by measures including Lempel--Ziv
size~\cite{LZ77}, run-length BWT size~\cite{bwt,rindex}, grammar
compression size~\cite{Rytter03,Charikar05}, string attractor
size~\cite{attractors}, and substring complexity~\cite{delta}. These
indexes have been studied extensively, and optimal space
bounds~\cite{delta} as well as optimal query
times~\cite{VerbinY13,dichotomy} are known when it comes to supporting
them in compressed space characterized by the above frameworks.

A complementary line of research investigates the \emph{compact} or
\emph{entropy-bounded} setting, in which a text $T \in [0 \dd \sigma)^{n}$
is indexed using $\bigO(n \log \sigma)$ bits of space (or, in some
cases, space close to the empirical entropy $H_k(T)$). This
matches the space required to represent the text $T$ and improves upon
classical indexes such as the suffix array~\cite{sa} and the suffix
tree~\cite{Weiner73}, which require $\Omega(n \log n)$ bits. The two
classes of space-efficient indexes (those based on repetitiveness and
those based on compact/entropy bounds) are not comparable, as entropy-based
measures are oblivious to repetition~\cite{attractors}. Hence, these
two approaches have been studied largely independently. Many of the
core queries mentioned above have been studied in this compact
setting. Among the most widely used results is the ability to support
suffix array~\cite{FerraginaM05,GrossiV05} and suffix
tree~\cite{Sadakane02} or longest common extension (LCE)
queries~\cite{sss} in compact space. While LCE queries admit an
optimal solution---that is, they can be supported in $\bigO(1)$ time
and $\bigO(n \log \sigma)$ bits of space (and moreover, given the
$\bigO(n \log \sigma)$-bit representation of the text, the data
structure can be constructed in the optimal $\bigO(n / \log_{\sigma}
n)$ time)~\cite{sss}---other central queries (like the suffix array
or inverse suffix array) are not known to admit such bounds. Recent
work sheds some light on this phenomenon. Specifically,
\cite{PrefixEquiv} recently proved that many central tasks, such as
suffix array (SA) or inverse suffix array (ISA) queries, are nearly
perfectly equivalent in all major aspects (space usage, query time,
construction time, and construction working space) to a new class of
\emph{prefix queries}, which are much simpler to study and develop new
trade-offs for. This provides a unified abstraction for compact text
indexes, reducing their design and analysis to the study of simple
operations on short bitstrings.

Parallel to research on compact data structures is the study of
space-efficient algorithms that operate in compact space.
Given a text $\Text \in [0 \dd \sigma)^{n}$
represented in $\bigO(n \log \sigma)$ bits, many important problems
can indeed be solved in the optimal $\bigO(n / \log_{\sigma} n)$ time
in this setting,
e.g.,~\cite{KikiBBGGW14,Ellert23,BannaiE23,Charalampopoulos22,RadoszewskiZ24,Charalampopoulos25}.
Such optimal construction is, however, not known for some of the
central problems, such as the computation of the Burrows--Wheeler
transform~\cite{sss}, the Lempel--Ziv
factorization~\cite{sublinearlz}, the longest common
factor~\cite{Charalampopoulos21}, or compressed suffix array
construction~\cite{breaking}. The above algorithms all achieve
$\bigO(n \sqrt{\log n} / \log_{\sigma} n)$ time for a string
$T \in [0 \dd \sigma)^{n}$. Recent work~\cite{hierarchy} casts some
light on why so many problems in this category share the same
complexity and develops a comprehensive hierarchy of reductions that
identify two equivalent core problems: \emph{dictionary matching} and
\emph{string nesting}, which must first be solved faster to hope for
any improvement for any of the above central string problems.

\paragraph{Organization of the Paper}

In \cref{sec:prelim}, we introduce the basic definitions used in this
paper. In \cref{sec:overview}, we provide an overview of our new upper
bounds and hardness results. In \cref{sec:1d-log,sec:1d-opt}, we
describe our key new concepts applied to 1D grammars. The purpose of
these sections is to ease the exposition of the subsequent ones. In
\cref{sec:2d-log,sec:2d-opt}, we present our new data structure
implementing random access on 2D grammars. Finally, in
\cref{sec:hardness}, we present our hardness results, beginning with
the hardness of pattern matching over 2D
grammar-compressed text (\cref{sec:hardness-pattern-matching}) and
continuing with our data structure hardness results for various 2D
queries (\cref{sec:hardness-data-structures}).

\section{Preliminaries}\label{sec:prelim}

\subsection{Strings}\label{sec:prelim-strings}

By a \emph{string} over an alphabet $\Sigma$, we mean any sequence $w$
whose elements are characters in $\Sigma$. The length of $w$ is
denoted $|w|$, and the $i$th leftmost symbol of $w$ is denoted $w[i]$,
where $i \in [1 \dd |w|]$. We denote the set of strings of length $n$
over the alphabet $\Sigma$ by $\Sigma^{n}$.

Substrings of $w$ are denoted $w[i \dd j]$, where $1 \leq i \leq j
\leq |w|$. We use an open bracket to indicate that the range is open;
e.g., $w(i \dd j] = w[i+1 \dd j]$ and $w[i \dd j) = w[i \dd j-1]$.

A \emph{prefix} (resp.\ \emph{suffix}) of $w$ is any string of the
form $w[1 \dd i]$ (resp.\ $w(i \dd |w|]$), where $0 \leq i \leq |w|$.

The concatenation of strings $u$ and $v$ is a new string $w$ of length
$|u| + |v|$ such that $u$ is a prefix of $w$ and $v$ is a suffix of
$w$. We denote the concatenation by $uv$ or $u \cdot v$.

For every pattern $P \in \Sigma^{*}$ and text $T \in \Sigma^{*}$, we
denote the set of occurrences of $P$ in $T$ by $\Occ{P}{T} = \{j \in
[1 \dd |T|] : j + |P| \leq |T| + 1\text{ and }T[j \dd j + |P|) =
P\}$. We denote the empty string by $\emptystring$.

\subsection{2D Strings}\label{sec:prelim-2d-strings}

By a \emph{2D string} over an alphabet $\Sigma$, we mean any (not
necessarily square) matrix $A$ whose elements are symbols from
$\Sigma$. The numbers of rows and columns of matrix $A$ are denoted
$\Rows{A}$ and $\Cols{A}$, respectively. The symbol of $A$
in row $i \in [1 \dd \Rows{A}]$ and column $j \in [1 \dd \Cols{A}]$
is denoted $A[i,j]$. We denote the set of matrices over $\Sigma$
with $r$ rows and $c$ columns by $\Sigma^{r \times c}$.

Submatrices of $A \in \Sigma^{r \times c}$ are denoted
$A[b_r \dd e_r][b_c \dd e_c]$.
Similarly to strings, we allow excluding the
first/last row and/or column; e.g., $A(b_r \dd e_r)(b_c \dd e_c] =
A[b_r+1 \dd e_r-1][b_c+1 \dd e_c]$.

For any 2D strings $A \in \Sigma^{r \times c_1}$ and $B \in \Sigma^{r \times c_2}$,
by $A \vconcat B$ we define a 2D string $C \in \Sigma^{r \times (c_1+c_2)}$
satisfying $C[1 \dd r](0 \dd c_1] = A$ and $C[1 \dd r](c_1 \dd c_1+c_2] = B$.
Symmetrically, for any $A \in \Sigma^{r_1 \times c}$ and
$B \in \Sigma^{r_2 \times c}$, by $A \hconcat B$ we define a 2D string
$C \in \Sigma^{(r_1+r_2) \times c}$ satisfying
$C(0 \dd r_1][1 \dd c] = A$ and $C(r_1 \dd r_1+r_2][1 \dd c] = B$;
see~\cite{giammarresi1997two}.

\subsection{Grammars}\label{sec:prelim-grammars}

A \emph{context-free grammar (CFG)} is a tuple $G = (V, \Sigma,
R, S)$ such that $V \cap \Sigma = \emptyset$ and
\begin{itemize}
  \item $V$ is a finite nonempty set of \emph{nonterminals} or \emph{variables},
  \item $\Sigma$ is a finite nonempty set of \emph{terminal} symbols,
  \item $R \subseteq V \times (V \cup \Sigma)^*$ is a set of
    \emph{productions} or \emph{rules}, and
  \item $S \in V$ is the special \emph{starting nonterminal}.
\end{itemize}

We say that $u \in (V \cup \Sigma)^{*}$ \emph{derives} $v$, and write
$u \Rightarrow^* v$, if $v$ can be obtained from $u$ by repeatedly
replacing nonterminals according to the rule set $R$. We then denote
$\Lang{G} := \{w \in \Sigma^* \mid S \Rightarrow^* w\}$.

By a \emph{straight-line grammar (SLG)} we mean a CFG $G =
(V,\Sigma,R,S)$ such that:
\begin{enumerate}
\item there exists an ordering $(N_1, \dots, N_{|V|})$ of all elements
  in $V$ such that, for every $(N,\gamma) \in R$, letting $i \in [1
    \dd |V|]$ be such that $N = N_i$, it holds $\gamma \in (\{N_{i+1},
  \dots, N_{|V|}\} \cup \Sigma)^{*}$, and
\item for every nonterminal $N \in V$, there exists exactly one
  $\gamma \in (V \cup \Sigma)^{*}$ such that $(N,\gamma) \in R$.
\end{enumerate}
The unique string $\gamma \in (V \cup \Sigma)^{*}$ such that
$(N,\gamma) \in R$ is called the \emph{definition} or \emph{right-hand side}
of nonterminal $N$ and is denoted $\Rhs{G}{N}$. Note that in an SLG, for every $\alpha
\in (V \cup \Sigma)^{*}$, there exists exactly one string $\gamma \in
\Sigma^{*}$ satisfying $\alpha \Rightarrow^{*} \gamma$. Such $\gamma$
is called the \emph{expansion} of $\alpha$ and is denoted
$\Exp{G}{\alpha}$. In particular, in an SLG, we have $|\Lang{G}| = 1$.
We define the size of an SLG as $|G| = \sum_{N \in V}\max(|\Rhs{G}{N}|, 1)$.

An SLG in which, for every $N \in V$, it holds $\Rhs{G}{N} = AB$,
where $A,B \in V$, or $\Rhs{G}{N} = a$, where $a \in \Sigma$, is
called a \emph{straight-line program (SLP)}.

\begin{observation}\label{ob:slg-to-slp}
  Every SLG $G$ satisfying $\Lang{G} \neq \{\emptystring\}$
  can be transformed in $\bigO(|G|)$ time into an
  SLP $G'$ satisfying $|G'| = \Theta(|G|)$ and $\Lang{G'} = \Lang{G}$.
\end{observation}

\begin{lemma}\label{lm:slp-min-size}
  If $G$ is an SLP such that $\Lang{G} = \{T\}$, where $|T| = n$, then
  $|G| = \Omega(\log n)$.
\end{lemma}
\begin{proof}
  The length of the longest text represented by an SLP
  of size $g$ is $2^g$. Thus, letting $T$ be such that $\Lang{G} = \{T\}$
  and $|T| = n$, we have $n \leq 2^{|G|}$, which implies $|G| = \Omega(\log n)$.
\end{proof}

\subsection{2D Grammars}\label{sec:prelim-2d-grammars}

A \emph{2D straight-line grammar (2D SLG)} is a tuple
$G = (V_{l}, V_{h}, V_{v}, \Sigma, R, S)$, where,
letting $V = V_{l} \cup V_{h} \cup V_{v}$,
\begin{itemize}
\item $V_{l}$ is a finite nonempty set of \emph{literal nonterminals},
\item $V_{h}$ is a finite set of \emph{horizontal nonterminals},
\item $V_{v}$ is a finite set of \emph{vertical nonterminals},
\item $\Sigma$ is a finite nonempty set of \emph{terminal} symbols,
\item $R \subseteq V \times (V \cup \Sigma)^*$ is a set of
  \emph{productions} or \emph{rules}, and
\item $S \in V$ is the special \emph{starting nonterminal}.
\end{itemize}
We assume that the sets $V_{l}$, $V_{h}$, $V_{v}$, and $\Sigma$ are
pairwise disjoint. We assume that $|R| = |V|$ and,
for every $N \in V$, there exists exactly one $(X,\gamma) \in R$ such
that $X = N$. The unique string $\gamma$ in the rule $(N,\gamma)$ is
denoted $\Rhs{G}{N}$. We assume that for every $N \in V_{l}$, it
holds $\Rhs{G}{N} \in \Sigma$. We assume that there exists an
ordering $(N_1, N_2, \dots, N_{g})$ of all elements in $V$ such that,
for every $i \in [1 \dd g]$ satisfying $N_i \in V_{h} \cup V_{v}$, it
holds $\Rhs{G}{N_i} \in \{N_{i+1}, \dots, N_{g}\}^{*}$. Given the
above ordering $(N_1, \dots, N_{g})$, we inductively define the
\emph{expansion} of a nonterminal $N_i$, denoted $\Exp{G}{N_i}$, as
follows. For every $i = g, \dots, 1$, by $\Exp{G}{N_i}$ we denote a 2D
string defined as follows:
\begin{itemize}
\item If $N_i \in V_{l}$, then $\Exp{G}{N_i} \in \Sigma^{1 \times 1}$
  is such that $\Exp{G}{N_i}[1,1] = \Rhs{G}{N_i}$.
\item If $N_i \in V_{h}$, then we require that there exists $w \in
  \Zp$ such that, letting $\gamma = \Rhs{G}{N_i}$, for every
  $j \in [1 \dd |\gamma|]$, it holds $\Cols{\Exp{G}{\gamma[j]}} = w$.
  We then define
  \[
    \Exp{G}{N_i} =
      \Exp{G}{\gamma[1]} \hconcat
      \Exp{G}{\gamma[2]} \hconcat \dots \hconcat
      \Exp{G}{\gamma[|\gamma|]}.
  \]
\item If $N_i \in V_{v}$, then we require that there exists $h \in \Zp$
  such that, letting $\gamma = \Rhs{G}{N_i}$, for every $j \in
  [1 \dd |\gamma|]$, it holds $\Rows{\Exp{G}{\gamma[j]}} = h$.
  We then define
  \[
    \Exp{G}{N_i} = 
      \Exp{G}{\gamma[1]} \vconcat
      \Exp{G}{\gamma[2]} \vconcat \dots \vconcat
      \Exp{G}{\gamma[|\gamma|]}.
  \]
\end{itemize}
We denote $\Lang{G} := \{\Exp{G}{S}\}$. We define the size of a 2D SLG
$G = (V_{l}, V_{h}, V_{v}, \Sigma, R, S)$ as $|G| = \sum_{N \in V}
\max(|\Rhs{G}{N}|, 1)$, where $V = V_{l} \cup V_{h} \cup V_{v}$.

A 2D SLG $G = (V_{l}, V_{h}, V_{v}, \Sigma, R, S)$ in which, for every
$N \in V_{h} \cup V_{v}$, it holds $|\Rhs{G}{N}| = 2$, is called a
\emph{2D straight-line program (2D SLP)}.

\begin{observation}\label{ob:2d-slg-to-2d-slp}
  Every 2D SLG $G$ satisfying $\Lang{G} \neq \{\emptystring\}$ can be
  transformed in $\bigO(|G|)$ time into a 2D SLP $G'$ satisfying $|G'|
  = \Theta(|G|)$ and $\Lang{G'} = \Lang{G}$.
\end{observation}

\begin{lemma}\label{lm:2d-slp-min-size}
  If $G$ is a 2D SLP such that $\Lang{G} = \{T\}$, where
  $\Rows{T} = m$ and $\Cols{T} = n$, then $|G| = \Omega(\log m + \log n)$.
\end{lemma}
\begin{proof}
  The maximal number of rows in a 2D string represented by a 2D SLP of
  size $g$ is $2^g$. Thus, letting $T$ be such that $\Lang{G} = \{T\}$
  and $\Rows{T} = m$, we have $m \leq 2^{|G|}$, which implies $|G| =
  \Omega(\log m)$. Analogously, we also have $|G| = \Omega(\log n)$,
  and hence $|G| = \Omega(\max(\log m, \log n)) =
  \Omega(\log m + \log n)$.
\end{proof}

\section{Technical Overview}\label{sec:overview}

\subsection{New Upper Bounds}

We begin by giving an overview of our new data structure that supports
random access on 2D grammars in optimal time. Existing data
structures for random access on 2D strings~\cite{CarfagnaMRSU24}
employ techniques related to the heavy-light decomposition of the
grammar~\cite{BLRSRW15}. In contrast, our data structure is based on a different
approach, which we refer to as \emph{bookmarks}. This technique underlies
some known data structures for random access in 1D strings~\cite{attractors}.
However, in its basic form, it does not achieve the same results
for 2D compressed strings. To overcome this limitation, we develop a new variant
of bookmarks that we call \emph{internal bookmarks}. To explain their role,
we first recall the notion of regular bookmarks and highlight the obstacles
to using this technique directly. We then provide an overview of our new approach.

\paragraph{Difficulty of Using Regular Bookmarks in 2D for Random Access}

A string attractor~\cite{attractors} of a string $T \in \Sigma^{n}$ is
any set $\Gamma \subseteq [1 \dd n]$ that has the property that for
every substring $T[i \dd j]$, there exists an occurrence $T[i' \dd
j']$ with $T[i \dd j] = T[i' \dd j']$ such that $p \in [i' \dd j']$
holds for some $p \in \Gamma$. In other words, $\Gamma$ is a ``hitting
set'' for all substrings of $T$. Given an attractor $\Gamma$ for $T
\in \Sigma^{n}$ of size $|\Gamma| = \gamma$, we can support random
access to $T$ in $\bigO(\gamma \log n)$ space by partitioning every
block $T[i \dd j]$ between two positions $i,j \in \Gamma$ into the
so-called \emph{concentric exponential parse (CEP)}~\cite{attractors}.
That is, we first create blocks of size one (on each of the two ends of
the block), followed (toward the middle) by blocks of size two, four,
eight, and so on. We call them \emph{bookmarks}. For each of the
blocks, we store one of its occurrences containing a position in
$\Gamma$. With such partitioning, we can always map a position $i
\in [1 \dd n]$ so that the distance to the nearest element of $\Gamma$
is reduced by at least a factor of two. In~\cite{attractors}, it was
proved that if a 1D grammar $G$ represents a string $T$, then $T$ has
an attractor $\Gamma_{G}$ of size $|\Gamma_{G}| = \bigO(|G|)$. This
implies that the above structure yields random access in $\bigO(|G| \log
n)$ space. The difficulty of generalizing this idea to \emph{2D
attractors} (which have been defined and studied
in~\cite{CarfagnaM24,CarfagnaMRSU24}) is that:
\begin{itemize}
\item In~\cite{CarfagnaMRSU24}, it was proved that there exist 2D strings
  $M \in \Sigma^{r \times c}$ for which there exists a 2D grammar $G$
  representing $M$ such that the size $\gamma(M)$ of the smallest 2D
  attractor of $M$ satisfies $\gamma(M)/|G| = \Omega(rc / \polylog
  (rc))$ (see~\cite[Propositions 1 and 10]{CarfagnaMRSU24}).
  In other words, not every 2D grammar implies
  a 2D attractor of similar size. Hence, even if random access were
  possible in $\bigO(\gamma(M) \polylog (rc))$ space (which is
  currently not known), this would not imply random access in
  $\bigO(|G| \polylog (rc))$ space.
\item In addition to the above issue, it is currently not known how to
  generalize the idea of bookmarks into 2D so that after mapping some
  point $(a,b)$ according to a bookmark, it maintains \emph{both} the
  horizontal and vertical distance to the nearest element of
  $\Gamma_{2D}$. It is easy to ensure that the distance in one
  direction is reduced, but then the other may grow uncontrollably.
\end{itemize}

\paragraph{Our Approach}

To sidestep the above barriers, we take a slightly different approach
in our solution. To illustrate our techniques, let us first
revisit the problem of random access to 1D grammars. This serves only
as a warm-up and to illustrate our key ideas, since random access to 1D
grammars is already well
studied~\cite{Rytter03,Charikar05,balancing,BilleEGV18}.
The first idea is to forgo the concept of bookmarks defined in terms of
text positions. Instead, all our bookmarks are defined relative to
some grammar variable. To define a bookmark for a substring $S$ of the
expansion of a variable $X$ (denoted $\Exp{G}{X}$), we consider a
traversal in the parse tree of $X$, as long as $S$ falls entirely
within the expansion of some variable. The last
variable $H$ in the traversal is called the \emph{hook} of $S$, and
the position of $S$ within $\Exp{G}{H}$ is the \emph{offset} of $S$;
see \cref{def:hook,def:offset}, and \cref{fig:1d-hook-and-offset}.
Our data structure (see
\cref{sec:1d-log-structure}) consists of $\log n$ levels. At level $p$,
we store, for every variable $X \in V$, the bookmarks for two
non-overlapping blocks of size $2^{p}$ touching the left boundary of the
expansion of $X$. We store similar information for the right
boundary. At query time (see \cref{sec:1d-log-queries} and
\cref{fig:1d-log-access}), we start at level $p = \lceil \log n \rceil$
and maintain a variable $c$ (see \cref{fig:1d-log-access}) to track the
direction (left or right) with respect to which we measure the
\emph{current distance} $\delta$ to the expansion boundary. We then
continue mapping the position according to bookmarks, traversing all
levels $p, p-1, p-2, \ldots, 0$. The positioning of bookmarks
guarantees that:
\begin{enumerate}
\item The distance $\delta$ with respect to the current direction $c$
  is always halved (see \cref{lm:left-map,lm:right-map}).
\item The distance in the direction \emph{opposite} to the current one
  (indicated by variable $c$) never increases. This holds
  because the bookmark of a substring of the expansion of variable $X$
  always maps into a variable that occurs in the parse of $X$. We
  remark that this property is not needed in the 1D case but turns
  out to be crucial for random access to 2D grammars.
\end{enumerate}
In total, we obtain a structure with $\bigO(g \log n)$ space and
$\bigO(\log n)$ query time (\cref{pr:1d-log-query}).
By generalizing the degree of the above structure from $2$ to $\tau$
(i.e., at level $p \in [0 \dd \lceil \log_{\tau} n \rceil]$, we store
$2g \tau$ bookmarks for substrings of length $\tau^p$), we can reduce
the query time to $\bigO(\log_{\tau} n)$ at the cost of increasing
space to $\bigO(g \cdot \tau \cdot \log_{\tau} n)$; see
\cref{sec:1d-opt}. With $\tau = \log^{\epsilon} n$, we can thus
replicate the optimal-time $\bigO(\tfrac{\log n}{\log \log n})$ random
access to 1D grammars (\cref{th:1d-opt}). The main
purpose of describing this structure is to introduce our new concepts
before presenting the complete structure in 2D. Although the random
access result for 1D grammars is known, the structure we present is
new and has certain novel properties: if, during the mapping, we skip
one of the bookmarks, we do not lose progress; i.e., $\delta$ remains
the distance to the expansion boundary in the direction given by $c$.

Let us now consider random access to the 2D grammar $G$ encoding the
string $M \in \Sigma^{r \times c}$ and let $n = \max(r,c)$. To
generalize the above structure to 2D grammars, we first extend
the concepts of hook and offset to 2D substrings (see
\cref{def:2d-hook,def:2d-offset}, and \cref{fig:2d-hook-and-offset}). This requires
separately considering horizontal and vertical variables. In the data
structure, for every level $(p_r,p_c) \in [0 \dd \lceil \log r \rceil]
\times [0 \dd \lceil \log c \rceil]$ and for every variable $X$, we
store bookmarks for four non-overlapping 2D substrings of size
$2^{p_r} \times 2^{p_c}$ touching each of the four corners of the
expansion of $X$ (for a total of 16 bookmarks per variable at each
level); see \cref{sec:2d-log-structure}. In total, this requires
$\bigO(|G| \cdot \log^2 n)$ space. At query time (see
\cref{sec:2d-log-queries} and
\cref{fig:2d-top-left-map,fig:2d-log-access}),
we continue mapping the initial point by always decreasing one of the
levels $p_r$ or $p_c$. After $\bigO(\log n)$ time, we reach a variable
expanding to a single symbol.
This algorithm is correct
due to our stronger notion of bookmarks. Specifically, in the query algorithm
(\cref{fig:2d-log-access}), when using a bookmark that reduces the
distance $\delta_r$ (which happens when $H$ is a horizontal variable),
the value $\delta_c$ is guaranteed not to increase. An analogous
property holds when we reduce $\delta_c$. The above data structure can
be further generalized, as before, to achieve
$\bigO(\tfrac{\log n}{\log \log n})$ query time in
$\bigO(|G| \cdot \log^{2+\epsilon} n)$ space (see \cref{sec:2d-opt} and
\cref{th:2d-opt}).

\subsection{New Hardness Results}

\paragraph{Hardness of Pattern Matching on 2D SLPs}

Our first hardness result concerns the problem of matching a one-dimensional
pattern within a two-dimensional string represented by a 2D SLP
(\cref{sec:2d-pattern-matching-problem-def}).
To establish
this conditional lower bound, we reduce from the
\emph{Orthogonal Vectors (OV)} problem (\cref{sec:ov}).

Given an OV instance with $n$ Boolean vectors in $d$ dimensions, we construct,
for each coordinate $j \in [1 \dd d]$, a vertical column
$C_j \in \BinaryAlphabet^{n \times 1}$ whose $i$th entry indicates whether the
$i$th vector has a $1$ in position $j$. Each vector $a_i$ having ones at
coordinates $b_1, \dots, b_\ell$ is then represented by the vertical
concatenation
\[
  \Gamma_i = C_{b_1} \vconcat C_{b_2} \vconcat \dots \vconcat C_{b_\ell}.
\]
The $j$th row of $\Gamma_i$ encodes whether $a_i$ and $a_j$ share a common
$1$, that is, $\Gamma_i[j,t] = 1$ for some $t$ if and only if
$a_i \cdot a_j \neq 0$. Hence, $a_i$ and $a_j$ are orthogonal precisely
when the $j$th row of $\Gamma_i$ consists entirely of zeros.

To make this construction consistent, we require all $\Gamma_i$ to have the
same width. For this reason, we first preprocess the OV instance so that every
vector has the same number of ones, say $\ell$, using the transformation in
\cref{pr:uniform-ov}. This step preserves orthogonality while ensuring that all
$\Gamma_i$ are $n \times \ell$ matrices, allowing us to define a single fixed
pattern length.

By concatenating all $\Gamma_i$ and separating them with delimiter
columns, we obtain a 2D text in which a row segment of the form
\[
  P = \one \, \zero^{\ell} \, \one
\]
occurs if and only if there exists a pair of orthogonal vectors.

The resulting 2D text can be generated by a 2D SLP $G$ of size
$\bigO(n \cdot d)$, which can be efficiently constructed
(\cref{lm:ov-reduction}). Therefore, any algorithm solving
2D pattern matching on grammar-compressed texts in time
$\bigO(|G|^{2-\epsilon} \cdot |P|^{\bigO(1)})$ for some $\epsilon > 0$ would
violate the Orthogonal Vectors Conjecture, yielding the claimed conditional
lower bound.

\paragraph{Hardness of LCE and Related Queries on 2D SLPs}

Our second family of hardness results shows that many natural queries about
subrectangles of a grammar-compressed 2D string become hard once the text
is represented by a 2D SLP. The core idea is a \emph{marking}
reduction that turns 1D rank/symbol occurrence queries into simple
2D numeric or Boolean queries.

Concretely, for a 1D string $T$ over an integer alphabet
$\Sigma = [0 \dd \sigma)$, we build the \emph{alphabet marking matrix}
$\MarkAllChars{T}{\sigma}$ (\cref{def:mark-all-chars}): each
alphabet symbol $c$ becomes a single row whose entries are $1$ exactly
at positions where $c$ occurs in $T$. In \cref{lm:reduce-rank-to-line-sum},
we show a tight connection
$\Rank{T}{j}{c} = \LineSumQuery{\MarkAllChars{T}{\sigma}}{c+1}{j}{j}$, i.e.,
rank queries reduce to 2D \emph{line sum} queries. Crucially,
\cref{lm:mark-all-chars-grammar} shows that if $T$ is generated by a
1D SLP of size $g$, then $\MarkAllChars{T}{\sigma}$ can be produced
by a 2D SLG of size $\bigO(g + \sigma)$ (and hence by a 2D SLP of size
$\bigO(g + \sigma)$ after standard conversion). Combining these facts
(\cref{pr:reduce-rank-to-line-sum}) yields the conditional lower
bound for line sum (and hence sum) queries on binary 2D SLPs stated in
\cref{th:reduce-from-rank}.

A similar but slightly different marking, the \emph{extended alphabet
marking matrix} $\ExtMarkAllChars{T}{\sigma}$
(\cref{def:ext-mark-all-chars}), pads each symbol row with zero blocks so
that symbol occurrence queries on $T$ correspond exactly to
\emph{square all-zero} queries on the 2D matrix
(\cref{lm:reduce-symbol-occ-to-square-all-zero}).
\cref{lm:ext-mark-all-chars-grammar} again guarantees an efficient 2D SLP
construction, and \cref{pr:reduce-symbol-occ-to-square-all-zero}
(plus the alphabet reduction in
\cref{pr:slp-rank-and-symbol-occ-alphabet-reduction}) gives the
hardness result presented in \cref{th:reduce-from-symbol-occ} for
square-all-zero queries and, by simple reductions, for most of the other
queries in \cref{def:2d-integer-problems,def:2d-general-problems}
(e.g., square LCE, line LCE, equality).
The inter-reductions among query types are compactly handled by
\cref{pr:reduce-square-lce-to-line-lce,%
  pr:reduce-line-lce-to-equality,%
  pr:reduce-square-all-zero-to-square-lce}.

In summary, marking transforms 1D rank/symbol occurrence problems into
very local 2D sum or all-zero checks. The marking matrices are producible
by small 2D grammars. Therefore, any data structure that uses
$\bigO(|G| \cdot \log^{\bigO(1)} n)$ space and answers these 2D queries in
$\bigO(\log^{\bigO(1)} n)$ time would yield equally efficient rank or symbol occurrence
structures for 1D SLPs.

\section{Revisiting Random Access using SLPs in Logarithmic Time}\label{sec:1d-log}

Let
$G = (V, \Sigma, R, S)$
be an SLP (see \cref{sec:prelim-grammars})
fixed for the duration of this section.
Denote
$n = |\Exp{G}{S}|$
and
$V = \{N_{1}, \ldots, N_{|V|}\}$.
We assume that $S = N_{1}$.
The aim of this section is to develop a certain data structure using
$\bigO(|G| \log n)$ space
that implements random access queries to $\Exp{G}{S}$
in $\bigO(\log n)$ time.
As pointed out earlier, we remark that such (or even more space efficient)
data structures are already known~\cite{Rytter03,Charikar05,Jez16,balancing,blocktree,attractors},
and the only purpose of this section is to introduce concepts which will be
generalized into two-dimensional strings in \cref{sec:2d-log}. We remark,
however, that this specific formulation has not been used before.

\subsection{Preliminaries}\label{sec:1d-log-prelim}

\begin{definition}[Hook]\label{def:hook}
  Let $N \in V$. Denote $m = |\Exp{G}{N}|$ and let $b, e \in [0 \dd m]$
  be such that $b < e$. If $m = 1$, then we define $\Hook{G}{N}{b}{e} = N$.
  Otherwise, i.e., if $m \geq 2$, letting $A,B \in V$ be such that
  $\Rhs{G}{N} = AB$, and $\ell = |\Exp{G}{A}|$, we define
  \[
    \Hook{G}{N}{b}{e} =
      \begin{cases}
        N & \text{if }b < \ell < e,\\
        \Hook{G}{A}{b}{e} & \text{if }e \leq \ell,\\
        \Hook{G}{B}{b-\ell}{e-\ell} & \text{if }\ell \leq b.
      \end{cases}
  \]
\end{definition}

\begin{definition}[Offset]\label{def:offset}
  Let $N \in V$. Denote $m = |\Exp{G}{N}|$.
  Let $b, e \in [0 \dd m]$ be such that $b < e$.
  If $m = 1$, we define $\Offset{G}{N}{b}{e} = 0$.
  Otherwise, i.e., if $m \geq 2$, letting $A,B \in V$ be such that
  $\Rhs{G}{N} = AB$, and $\ell = |\Exp{G}{A}|$, we define
  \[
    \Offset{G}{N}{b}{e} =
      \begin{cases}
        b & \text{if }b < \ell < e,\\
        \Offset{G}{A}{b}{e} & \text{if }e \leq \ell,\\
        \Offset{G}{B}{b-\ell}{e-\ell} & \text{if }\ell \leq b.
      \end{cases}
  \]
\end{definition}

\begin{figure}[t!]
  \centering
  \begin{tikzpicture}[yscale=0.3, xscale=0.5]
    \node[] at (-3,1) {$\Exp{G}{N}:$};
    \draw (0,0) rectangle (20,2);
    \draw[dotted] (6,0) -- (6,2.5);
    \draw[dotted] (9,0) -- (9,3.5);
    \fill[pattern = north west lines, pattern color = gray] (6,0) rectangle (9,2);
    \draw[stealth-stealth] (0,2.5) -- (6,2.5) node[midway,yshift=0.5em]{$b$};
    \draw[stealth-stealth] (0,3.5) -- (9,3.5) node[midway,yshift=0.5em]{$e$};
    \draw[-stealth] (10,0) -- (7,-2);
    \draw[-stealth] (10,0) -- (17,-2);

    \draw (0,-2) rectangle (20,-4);
    \draw[dotted] (6,-2) -- (6,-4);
    \draw[dotted] (9,-2) -- (9,-4);
    \fill[pattern = north west lines, pattern color = gray] (6,-4) rectangle (9,-2);
    \draw[solid] (14,-2) -- (14,-4);
		
    \draw[-stealth] (7,-4) -- (2,-6);
    \draw[-stealth] (7,-4) -- (9,-6);
		
    \draw (0,-6) rectangle (14,-8);
    \draw[dotted] (6,-6) -- (6,-8);
    \draw[dotted] (9,-6) -- (9,-8);
    \fill[pattern = north west lines, pattern color = gray] (6,-8) rectangle (9,-6);
    \draw[solid] (4,-6) -- (4,-8);

    \draw[-stealth] (9,-8) -- (6,-10);
    \draw[-stealth] (9,-8) -- (11,-10);

    \node[] at (-3,-11) {$\Exp{G}{H}:$};
    \draw (4,-10) rectangle (14,-12);
    \draw[dotted] (6,-10) -- (6,-12);
    \draw[dotted] (9,-10) -- (9,-12);
    \fill[pattern = north west lines, pattern color = gray] (6,-12) rectangle (9,-10);
    \draw[solid] (8,-10) -- (8,-12);
    \draw[stealth-stealth] (4,-12.5) -- (6,-12.5) node[midway,yshift=-0.5em]{$\alpha$};
	\end{tikzpicture}
  \caption{Illustration of \cref{def:hook,def:offset}.
    Here, $N \in V$, $\Hook{G}{N}{b}{e} = H$, and
    $\Offset{G}{N}{b}{e} = \alpha$.}\label{fig:1d-hook-and-offset}
\end{figure}

\begin{lemma}\label{lm:hook}
  Let $N \in V$, $w = \Exp{G}{N}$, $m = |w|$, and $b, e \in [0 \dd m]$
  be such that $b < e$.
  Denote $H = \Hook{G}{N}{b}{e}$ (\cref{def:hook}),
  $\alpha = \Offset{G}{N}{b}{e}$ (\cref{def:offset}), and
  $\beta = \alpha + (e - b)$.
  Then, it holds
  $0 \leq \alpha < \beta \leq |\Exp{G}{H}|$ and
  $w(b \dd e] = \Exp{G}{H}(\alpha \dd \beta]$.
  Moreover,
  \begin{enumerate}
  \item\label{lm:hook-it-1}
    If $e - b = 1$, then $|\Exp{G}{H}| = 1$.
  \item\label{lm:hook-it-2}
    Otherwise, it holds $|\Exp{G}{H}| \geq 2$.
    Furthermore, letting $X, Y \in V$ be such that
    $\Rhs{G}{H} = XY$ and $\ell = |\Exp{G}{X}|$,
    it holds $\alpha < \ell < \beta$.
  \end{enumerate}
\end{lemma}
\begin{proof}

  By \cref{def:hook,def:offset},
  there exists a unique sequence of tuples
  $(N_i, b_{i}, e_{i})_{i \in [0 \dd q]}$,
  where $q \geq 0$, such that:
  \begin{itemize}
  \item for every $i \in [0 \dd q]$,
    $0 \leq b_{i} < e_{i} \leq |\Exp{G}{N_i}|$,
  \item $(N_{0}, b_{0}, e_{0}) = (N, b, e)$,
  \item for every $i \in [0 \dd q)$,
    \begin{itemize}
    \item $|\Exp{G}{N_i}| > 1$
      and, letting $X_i, Y_i \in V$ be such that
      $\Rhs{G}{N_i} = X_i Y_i$, $N_{i+1} \in \{X_i, Y_i\}$,
    \item $\Hook{G}{N_i}{b_{i}}{e_{i}} = \Hook{G}{N_{i+1}}{b_{i+1}}{e_{i+1}}$,
    \item $\Offset{G}{N_i}{b_{i}}{e_{i}} = \Offset{G}{N_{i+1}}{b_{i+1}}{e_{i+1}}$,
    \end{itemize}
  \item $\Hook{G}{N_{q}}{b_{q}}{e_{q}} = N_{q}$ and
    $\Offset{G}{N_{q}}{b_{q}}{e_{q}} = b_{q}$.
  \end{itemize}

  To prove the main claim, we first prove auxiliary properties of the above
  sequence. More specifically, we will first show that for every $i \in [0 \dd q)$,
  it holds:
  \begin{itemize}
  \item $e_{i+1} - b_{i+1} = e_{i} - b_{i}$,
  \item $\Exp{G}{N_{i+1}}(b_{i+1} \dd e_{i+1}] =
    \Exp{G}{N_i}(b_{i} \dd e_{i}]$.
  \end{itemize}
  Let $i \in [0 \dd q)$. As noted above, it holds
  $|\Exp{G}{N_i}| > 1$.
  Let $X_i, Y_i \in V$ be such that $\Rhs{G}{N_i} = X_i Y_i$.
  Denote $\ell = |\Exp{G}{X_i}|$.
  Above, we observe that $N_{i+1} \in \{X_i, Y_i\}$ and
  $\Hook{G}{N_i}{b_i}{e_i} = \Hook{G}{N_{i+1}}{b_{i+1}}{e_{i+1}}$.
  By \cref{def:hook}, this implies that we must either have
  $e_{i} \leq \ell$ or $\ell \leq b_{i}$.
  We consider two cases:
  \begin{itemize}
  \item First, assume that $e_{i} \leq \ell$. By \cref{def:hook}, we
    have $(N_{i+1}, b_{i+1}, e_{i+1}) =
    (X_i, b_{i}, e_{i})$.
    We thus immediately obtain
    $e_{i+1} - b_{i+1} = e_{i} - b_{i}$.
    On the other hand,
    $e_{i} \leq \ell$ implies that it holds $\Exp{G}{N_i}(b_{i} \dd e_{i}]
    = \Exp{G}{X_i}(b_{i} \dd e_{i}]$. Putting everything
    together, we thus obtain
    \begin{align*}
      \Exp{G}{N_{i+1}}(b_{i+1} \dd e_{i+1}]
        = \Exp{G}{X_i}(b_{i} \dd e_{i}]
        = \Exp{G}{N_i}(b_{i} \dd e_{i}].
    \end{align*}
  \item Assume now that $\ell \leq b_{i}$.
    By \cref{def:hook}, we
    have $(N_{i+1}, b_{i+1}, e_{i+1}) =
    (Y_i, b_{i} - \ell, e_{i} - \ell)$.
    We thus again obtain that
    $e_{i+1} - b_{i+1} = (e_{i} - \ell) - (b_{i} - \ell) = e_{i} - b_{i}$.
    On the other hand,
    $\ell \leq b_{i}$ implies that it holds $\Exp{G}{N_i}(b_{i} \dd e_{i}]
    = \Exp{G}{Y_i}(b_{i} - \ell \dd e_{i} - \ell]$. Putting everything
    together, we thus obtain
    \begin{align*}
      \Exp{G}{N_{i+1}}(b_{i+1} \dd e_{i+1}]
       = \Exp{G}{Y_i}(b_{i} - \ell \dd e_{i} - \ell]
       = \Exp{G}{N_i}(b_{i} \dd e_{i}].
    \end{align*}
  \end{itemize}

  Using the properties above, we are now ready to prove the main claims.
  First, observe that, from the properties of the sequence
  $(N_i, b_{i}, e_{i})_{i \in [0 \dd q]}$,
  we immediately obtain that $H = N_{q}$, $\alpha = b_{q}$, and
  $\beta
    = \alpha + (e - b)
    = b_{q} + (e_{q} - b_{q})
    = e_{q}$.
  This immediately implies the first claim, i.e.,
  $0 \leq \alpha < \beta \leq |\Exp{G}{H}|$,
  since by the above we have
  $0 \leq b_{q} < e_{q} \leq |\Exp{G}{N_{q}}|$.
  To obtain the second
  claim note that, by the above, it holds
  \begin{align*}
    w(b \dd e]
      &= \Exp{G}{N}(b \dd e]\\
      &= \Exp{G}{N_{0}}(b_{0} \dd e_{0}]\\
      &= \Exp{G}{N_{1}}(b_{1} \dd e_{1}]\\
      &= \dots\\
      &= \Exp{G}{N_{q}}(b_{q} \dd e_{q}]\\
      &= \Exp{G}{H}(\alpha \dd \beta].\\
  \end{align*}

  We now prove the second part of the main claim, considering each of the items separately:
  \begin{enumerate}

  \item First, assume that $e - b = 1$. Suppose the claim does not
    hold, i.e., $|\Exp{G}{H}| > 1$. By the above, this is equivalent
    to $|\Exp{G}{N_{q}}| > 1$. Then, there exist $X,Y \in V$ such that
    $\Rhs{G}{N_{q}} = XY$. Denote $\ell = |\Exp{G}{X}|$. Observe that
    we cannot have $b_{q} < \ell < e_{q}$, since that implies $e_{q} - b_{q} >
    1$, which contradicts $e - b = 1$ (recall that $e_{q} - b_{q} = e -
    b$). Thus, we must have either $\ell \leq b_{q}$ or $e_{q} \leq
    \ell$. In either case, by \cref{def:hook}, it holds
    $\Hook{G}{N_{q}}{b_{q}}{e_{q}} \neq N_{q}$, a contradiction. We thus must
    have $|\Exp{G}{N_{q}}| = 1$, or equivalently, $|\Exp{G}{H}| = 1$.

  \item Let us now assume that $e - b > 1$.
    Recall that above we observed that
    $H = N_{q}$,
    $e_{q} - b_{q} = e - b$, and
    $0 \leq b_{q} < e_{q} \leq |\Exp{G}{N_{q}}|$.
    This immediately implies that
    $|\Exp{G}{H}| = |\Exp{G}{N_{q}}| \geq e_{q} - b_{q} = e - b > 1$, i.e.,
    the first part of the claim.
    Let then $X, Y \in V$ be as in the claim,
    i.e., such that $\Rhs{G}{H} = XY$.
    Let $\ell = |\Exp{G}{X}|$.
    Recall that we have $\Hook{G}{N_{q}}{b_{q}}{e_{q}} = N_{q}$.
    By \cref{def:hook}, this implies that $b_{q} < \ell < e_{q}$. By $\alpha =
    b_{q}$ and $\beta = e_{q}$ (observed above), we thus obtain the second
    part of the claim, i.e., $\alpha < \ell < \beta$.
    \qedhere
  \end{enumerate}
\end{proof}

\subsection{The Data Structure}\label{sec:1d-log-structure}

\paragraph{Definitions}

Let $H^{\rm left}$ and $O^{\rm left}$ denote 3-dimensional arrays defined
so that for every
$i \in [1 \dd |V|]$,
$\LevelId \in [0 \dd \ceil{\log n}]$, and
$k \in \{0,1\}$
satisfying
$k \cdot 2^{\LevelId} < |\Exp{G}{N_i}|$,
it holds
\begin{align*}
  N_{H^{\rm left}[i,\LevelId,k]} &= \Hook{G}{N_i}{b}{e},\\
  O^{\rm left}[i,\LevelId,k] &= \Offset{G}{N_i}{b}{e},
\end{align*}
where
\begin{itemize}
\item $m = |\Exp{G}{N_i}|$,
\item $b = k \cdot 2^{\LevelId}$, and
\item $e = \min(m, (k+1) \cdot 2^{\LevelId})$.
\end{itemize}
Similarly, for every $i, \LevelId, k$ as above, we define
$H^{\rm right}$ and $O^{\rm right}$ such that it holds
\begin{align*}
  N_{H^{\rm right}[i,\LevelId,k]} &= \Hook{G}{N_i}{m-e}{m-b},\\
  O^{\rm right}[i,\LevelId,k] &= \Offset{G}{N_i}{m-e}{m-b},
\end{align*}
where $m, b, e$ are as above.

By $A_{\rm expsize}[1 \dd |V|]$ we denote an array defined by
$A_{\rm expsize}[i] = |\Exp{G}{N_i}|$.
By $A_{\rm rhs}[1 \dd |V|]$ we denote an array defined such that,
for $i \in [1 \dd |V|]$,
\begin{itemize}
\item If $|\Rhs{G}{N_i}| = 1$, then $A_{\rm rhs}[i] = \Rhs{G}{N_i}$,
\item Otherwise, $A_{\rm rhs}[i] = (i',i'')$, where $i', i'' \in [1 \dd |V|]$
  are such that $\Rhs{G}{N_i} = N_{i'} N_{i''}$.
\end{itemize}

\paragraph{Components}

The data structure consists of the following components:
\begin{enumerate}
\item The array $A_{\rm rhs}$ using $\bigO(|V|)$ space.
\item The array $A_{\rm expsize}$ using $\bigO(|V|)$ space.
\item The arrays
  $H^{\rm left}$, $O^{\rm left}$,
  $H^{\rm right}$, $O^{\rm right}$
  using
  $\bigO(|V| \log n)$
  space.
\end{enumerate}
In total, the data structure needs
$\bigO(|V| \log n) = \bigO(|G| \log n)$
space.

\subsection{Implementation of Queries}\label{sec:1d-log-queries}

\begin{definition}[Access function]\label{def:access}
  Let $N \in V$ and $m = |\Exp{G}{N}|$.
  For every
  $\delta \in [1 \dd m]$ and $c \in \{\DirLeft, \DirRight\}$,
  we define
  \vspace{1ex}
  \[
    \Access{G}{N}{\delta}{c} =
      \begin{cases}
        \Exp{G}{N}[\delta]         & \text{if }c = \DirLeft,\\
        \Exp{G}{N}[m - \delta + 1] & \text{otherwise}.
      \end{cases}
  \]
  \vspace{1ex}
\end{definition}

\begin{definition}[Left boundary mapping]\label{def:left-map}
  Let $N \in V$.
  Denote $m = |\Exp{G}{N}|$.
  Consider any $\LevelId \in \Zn$ and $\delta \in [1 \dd m]$
  satisfying
  $\delta \leq 2^{\LevelId+1}$.
  Denote
  $k = \lceil \tfrac{\delta}{2^{\LevelId}} \rceil - 1 \in \{0,1\}$,
  $b = k \cdot 2^{\LevelId}$,
  $e = \min(m, (k + 1) \cdot 2^{\LevelId})$, and
  $H = \Hook{G}{N}{b}{e}$.
  \begin{itemize}
  \item
    If $e - b = 1$,
    then we define
    \[
      \LeftMap{G}{N}{\LevelId}{\delta} = (H, 1, \DirLeft).
    \]
  \item
    Otherwise, letting
    $\alpha = \Offset{G}{N}{b}{e}$,
    $X, Y \in V$ be such that $\Rhs{G}{H} = XY$,
    and $\ell = |\Exp{G}{X}|$, we define
    \vspace{1ex}
    \[
      \LeftMap{G}{N}{\LevelId}{\delta} =
        \begin{cases}
          (X, (\ell - \alpha) - (\delta - b) + 1, \DirRight)   & \text{if }\delta - b \leq \ell - \alpha,\\
          (Y, (\delta - b) - (\ell - \alpha), \DirLeft)      & \text{otherwise}.\\
        \end{cases}
    \]
    \vspace{1ex}
  \end{itemize}
\end{definition}

\begin{definition}[Right boundary mapping]\label{def:right-map}
  Let $N \in V$ and $m = |\Exp{G}{N}|$.
  Consider any $\LevelId \in \Zn$ and any $\delta \in [1 \dd m]$ satisfying $\delta \leq 2^{\LevelId+1}$.
  Denote $k = \lceil \tfrac{\delta}{2^{\LevelId}} \rceil - 1 \in \{0,1\}$,
  $b = k \cdot 2^{\LevelId}$,
  $e = \min(m, (k+1) \cdot 2^{\LevelId})$, and
  $H = \Hook{G}{N}{m - e}{m - b}$.
  \begin{itemize}
  \item
    If $e-b = 1$, then we define
    \[
      \RightMap{G}{N}{\LevelId}{\delta} = (H, 1, \DirLeft).
    \]
  \item
    Otherwise,
    letting $\beta = |\Exp{G}{H}| - (\Offset{G}{N}{m - e}{m - b} + (e - b))$,
    $X, Y \in V$ be such that $\Rhs{G}{H} = XY$, and $\ell = |\Exp{G}{Y}|$, we define
    \vspace{2ex}
    \[
      \RightMap{G}{N}{\LevelId}{\delta} =
        \begin{cases}
          (Y, (\ell-\beta)-(\delta-b)+1, \DirLeft)  & \text{if }\delta-b \leq \ell-\beta,\\
          (X, (\delta-b)-(\ell-\beta), \DirRight)   & \text{otherwise}.\\
        \end{cases}
    \]
    \vspace{1ex}
  \end{itemize}
\end{definition}

\begin{remark}\label{rm:map}
  Note that $X$ and $Y$ in \cref{def:left-map,def:right-map} are well-defined when
  $e - b > 1$,
  since in this case we have
  $|\Exp{G}{H}| \geq 2$ (\cref{lm:hook}).
\end{remark}

\begin{lemma}\label{lm:left-map}
  Let $N \in V$ and
  $m = |\Exp{G}{N}|$.
  Consider any
  $\LevelId \in \Zn$ and
  $\delta \in [1 \dd m]$
  satisfying
  $\delta \leq 2^{\LevelId+1}$.
  Then, the tuple
  $(N', \delta', c) =
  \LeftMap{G}{N}{\LevelId}{\delta}$
  (\cref{def:left-map})
  satisfies:
  \begin{enumerate}
  \item
    $\delta' \in [1 \dd |\Exp{G}{N'}|]$,
  \item
    $\delta' \leq 2^{\LevelId}$,
  \item
    $\Access{G}{N}{\delta}{\DirLeft} = \Access{G}{N'}{\delta'}{c}$ (\cref{def:access}).
  \end{enumerate}
  Moreover, if
  $\LevelId = 0$,
  then
  $|\Exp{G}{N'}| = 1$.
\end{lemma}
\begin{proof}

  Denote:
  \begin{itemize}
  \item
    $k = \lceil \tfrac{\delta}{2^{\LevelId}} \rceil - 1 \in \{0,1\}$,
    $b = k \cdot 2^{\LevelId}$,
    $e = \min(m, (k + 1) \cdot 2^{\LevelId})$,
  \item $H = \Hook{G}{N}{b}{e}$ (\cref{def:hook}),
  \item $\alpha = \Offset{G}{N}{b}{e}$ (\cref{def:offset}), and
  \item $\beta = \alpha + (e - b)$.
  \end{itemize}

  Observe that
  $b = k \cdot 2^{\LevelId} =
  (\lceil \tfrac{\delta}{2^{\LevelId}} \rceil - 1) \cdot 2^{\LevelId} =
  (\lfloor \tfrac{\delta + 2^{\LevelId} - 1}{2^{\LevelId}} \rfloor - 1) \cdot 2^{\LevelId} \leq
  \delta + 2^{\LevelId} - 1 - 2^{\LevelId} < \delta$.
  On the other hand, by
  $\delta \leq \lceil \tfrac{\delta}{2^{\LevelId}} \rceil 2^{\LevelId} = (k + 1) \cdot 2^{\LevelId}$
  and $\delta \leq m$,
  it follows that
  $\delta \leq \min(m, (k + 1) \cdot 2^{\LevelId}) = e$.
  Note also that
  $e - b \leq (k + 1) \cdot 2^{\LevelId} - k \cdot 2^{\LevelId} = 2^{\LevelId}$.
  We have thus established
  that:
  \begin{itemize}
  \item $b < \delta \leq e$ and
  \item $e - b \leq 2^{\LevelId}$.
  \end{itemize}

  We prove each of the three main claims separately:
  \begin{enumerate}

  \item First, we prove that $\delta' \in [1 \dd |\Exp{G}{N'}|]$.
    If $e - b = 1$, then by \cref{def:left-map}, it holds $\delta' = 1$,
    and the claim follows immediately.
    Let us thus assume that
    $e - b \geq 2$. Then,
    it holds $|\Exp{G}{H}| \geq 2$ (see also \cref{rm:map}).
    Let $X, Y \in V$ be such that $\Rhs{G}{H} = XY$.
    Denote $\ell = |\Exp{G}{X}|$.
    We consider two cases:
    \begin{itemize}

    \item If $\delta - b \leq \ell - \alpha$ then, by \cref{def:left-map}, it holds
      $N' = X$ and $\delta' = (\ell - \alpha) - (\delta - b) + 1$.
      By the assumption $\delta - b \leq \ell - \alpha$, we immediately obtain
      $\delta' = (\ell - \alpha) - (\delta - b) + 1 \geq 1$.
      On the other hand, $b < \delta$ and $\alpha \geq 0$ (see \cref{lm:hook}) imply
      $\delta' = (\ell - \alpha) - (\delta - b) + 1 \leq \ell = |\Exp{G}{X}| = |\Exp{G}{N'}|$.
      We have thus proved $\delta' \in [1 \dd |\Exp{G}{N'}|]$.

    \item Otherwise (i.e., if $\delta - b > \ell - \alpha$), by \cref{def:left-map}, it holds
      $N' = Y$ and $\delta' = (\delta - b) - (\ell - \alpha)$.
      The assumption $\delta - b > \ell - \alpha$ immediately yields
      $\delta' = (\delta - b) - (\ell - \alpha) \geq 1$.
      On the other hand, $\delta \leq e$ and
      $(e - b) - (\ell - \alpha) = ((e - b) + \alpha) - \ell = \beta - \ell \leq
      |\Exp{G}{H}| - \ell = |\Exp{G}{Y}|$
      (where $\beta \leq |\Exp{G}{H}|$ follows by \cref{lm:hook}) imply
      $\delta' = (\delta - b) - (\ell - \alpha) \leq (e - b) - (\ell - \alpha)
      \leq |\Exp{G}{Y}| = |\Exp{G}{N'}|$.
      We have thus proved $\delta' \in [1 \dd |\Exp{G}{N'}|]$.
    \end{itemize}

  \item Second, we prove that
    $\delta' \leq 2^{\LevelId}$.
    If $e - b = 1$, then by \cref{def:left-map}, it holds $\delta' = 1$,
    and the claim follows immediately.
    Let us thus assume that
    $e - b > 1$.
    Then, it holds
    $|\Exp{G}{H}| \geq 2$ (see also \cref{rm:map}).
    Let $X, Y \in V$ be such that $\Rhs{G}{H} = XY$.
    Denote $\ell = |\Exp{G}{X}|$.
    We consider two cases:
    \begin{itemize}

    \item If $\delta - b \leq \ell - \alpha$ then, by \cref{def:left-map}, we have
      $\delta' = (\ell - \alpha) - (\delta - b) + 1$.
      Note that $\delta > b$ implies that
      $\delta' = (\ell - \alpha) - (\delta - b) + 1 \leq \ell - \alpha$.
      On the other hand, by
      $\ell < \beta$ (\cref{lm:hook}) and the definition of $\beta$,
      we have $\ell - \alpha \leq \beta - \alpha = e - b$.
      Thus,
      $\delta' \leq \ell - \alpha \leq e - b \leq 2^{\LevelId}$.

    \item Otherwise (i.e., if $\delta - b > \ell - \alpha$), by \cref{def:left-map},
      we have $\delta' = (\delta - b) - (\ell - \alpha)$.
      To show $\delta' \leq 2^{\LevelId}$,
      it suffices to note that by $\ell > \alpha$ (\cref{lm:hook})
      and $\delta' \leq (e - b) - (\ell - \alpha)$
      (shown above in the analogous case),
      it follows that
      $\delta' \leq (e - b) - (\ell - \alpha) \leq e - b \leq 2^{\LevelId}$.
    \end{itemize}

  \item Finally, we prove that it holds
    $\Access{G}{N}{\delta}{\DirLeft} = \Access{G}{N'}{\delta'}{c}$.
    If $e - b = 1$, then by \cref{def:left-map},
    we have
    $(N', \delta', c) = (H, 1, \DirLeft)$.
    Observe that by $b < \delta \leq e$,
    we then must have
    $\delta = e$.
    On the other hand, by
    \cref{lm:hook}\eqref{lm:hook-it-1},
    we then have
    $\Exp{G}{N}[e] = \Rhs{G}{H}$.
    Thus,
    \begin{align*}
      \Access{G}{N}{\delta}{\DirLeft}
        &= \Access{G}{N}{e}{\DirLeft}
        = \Exp{G}{N}[e]
        = \Rhs{G}{H}\\
        &= \Rhs{G}{N'}[\delta']
        = \Access{G}{N'}{\delta'}{\DirLeft}\\
        &= \Access{G}{N'}{\delta'}{c}.
    \end{align*}
    Let us now assume that $e - b > 1$.
    Let $X,Y \in V$ be such that $\Rhs{G}{H} = XY$.
    Denote $\ell = |\Exp{G}{X}|$.
    We consider two cases:
    \begin{itemize}

    \item If $\delta - b \leq \ell - \alpha$ then, by \cref{def:left-map}, we have
      $(N', \delta', c) = (X, (\ell - \alpha) - (\delta - b) + 1,\DirRight)$.
      Combining $\delta \leq b + (\ell - \alpha)$ (the assumption) with
      $b < \delta$ (proved above),
      we obtain
      $b < \delta \leq b + (\ell - \alpha)$.
      On the other hand, it holds
      $w(b \dd e] = \Exp{G}{H}(\alpha \dd \beta]$
      and $\alpha < \ell < \beta$ (\cref{lm:hook})
      Thus, for every
      $t \in (b \dd b + (\ell - \alpha)]$,
      it holds
      $w[t]
        = \Exp{G}{H}(\alpha + (t - b)]
        = \Exp{G}{X}(\alpha + (t - b)]
        = \Exp{G}{X}[\ell - ((\ell - \alpha) - (t - b))]$.
      In particular,
      $w[\delta] =
      \Exp{G}{X}[\ell - \delta' + 1]$.
      Thus,
      \begin{align*}
        \Access{G}{N}{\delta}{\DirLeft}
          &= w[\delta]
           = \Exp{G}{X}[\ell - \delta' + 1]\\
          &= \Access{G}{X}{\delta'}{\DirRight}
           = \Access{G}{N'}{\delta'}{c}.
      \end{align*}

    \item Otherwise (i.e., if $\delta - b > \ell - \alpha$),
      by \cref{def:left-map}, we have
      $(N', \delta', c) = (Y, (\delta - b) - (\ell - \alpha), \DirLeft)$.
      Combining
      $\delta > b + (\ell - \alpha)$ (the assumption) with
      $\delta \leq e$ (proved above),
      we obtain
      $b + (\ell - \alpha) < \delta \leq e$.
      On the other hand, it holds
      $w(b \dd e] = \Exp{G}{H}(\alpha \dd \beta]$
      and $\alpha < \ell < \beta$ (\cref{lm:hook}).
      Thus, for every
      $t \in (b + (\ell - \alpha) \dd e]$,
      it holds
      $w[t]
        = \Exp{G}{H}[\alpha + (t - b)]
        = \Exp{G}{Y}[\alpha + (t - b) - \ell]
        = \Exp{G}{Y}[(t - b) - (\ell - \alpha)]$.
      In particular,
      $w[\delta] = \Exp{G}{Y}[\delta']$.
      Thus,
      \begin{align*}
        \Access{G}{N}{\delta}{\DirLeft}
          &= w[\delta]
           = \Exp{G}{Y}[\delta']\\
          &= \Access{G}{Y}{\delta'}{\DirLeft}
           = \Access{G}{N'}{\delta'}{c}.
      \end{align*}
    \end{itemize}
  \end{enumerate}

  We now show the remaining claim, i.e., that
  $\LevelId = 0$ implies $|\Exp{G}{N'}| = 1$.
  By the above, $b < e$ and $e - b \leq 2^{\LevelId}$.
  Thus, if $\LevelId = 0$, we must have $b = 0$ and $e = 1$.
  By \cref{def:left-map}, we then have $N' = \Hook{G}{N}{b}{e} = H$.
  It remains to observe that by \cref{lm:hook}\eqref{lm:hook-it-1},
  we then have $|\Exp{G}{H}| = 1$.
\end{proof}

\begin{lemma}\label{lm:right-map}
  Let $N \in V$ and
  $m = |\Exp{G}{N}|$.
  Consider any
  $\LevelId \in \Zn$ and
  $\delta \in [1 \dd m]$
  satisfying
  $\delta \leq 2^{\LevelId+1}$.
  Then, the tuple
  $(N',\delta',c) = \RightMap{G}{N}{\LevelId}{\delta}$ (\cref{def:right-map})
  satisfies:
  \begin{enumerate}
  \item
    $\delta' \in [1 \dd |\Exp{G}{N'}|]$,
  \item
    $\delta' \leq 2^{\LevelId}$,
  \item
    $\Access{G}{N}{\delta}{\DirRight} = \Access{G}{N'}{\delta'}{c}$.
  \end{enumerate}
  Moreover, if
  $\LevelId = 0$,
  then
  $|\Exp{G}{N'}| = 1$.
\end{lemma}
\begin{proof}
  the proof is analogous to the proof of \cref{lm:left-map}.
\end{proof}

\begin{proposition}\label{pr:1d-log-query}
  Given the data structure from
  \cref{sec:1d-log-structure}
  and any
  $i \in [1 \dd n]$,
  we can compute
  $\Exp{G}{S}[i]$ in
  $\bigO(\log n)$
  time.
\end{proposition}
\begin{proof}
  The pseudocode of the algorithm is presented in
  \cref{fig:1d-log-access}.
  To compute the symbol
  $\Exp{G}{S}[i]$,
  we call
  $\texttt{RandomAccess}(i)$.
  At the beginning of the algorithm, we set
  $t$, $\delta$, and $c$
  so that it holds
  $\Access{G}{N_t}{\delta}{c} = \Exp{G}{S}[i]$.
  In each iteration of the while loop we then
  compute variables
  $t'$, $\delta'$, and $c'$
  such that
  \[
    \Access{G}{N_{t}}{\delta}{c} =
    \Access{G}{N_{t'}}{\delta'}{c'}.
  \]
  By \cref{lm:left-map,lm:right-map}, each iteration of the while loop
  decreases the variable $p$
  indicating the current level in the data structure.
  Once the algorithm reaches $p < 0$,
  it must hold 
  $|\Exp{G}{N_t}| = 1$.
  Thus, $\Rhs{G}{N_t} \in \Sigma$, and hence
  $\Rhs{G}{N_t} = \Exp{G}{N_t}[1] = \Exp{G}{S}[i]$.
  The total running time is $\bigO(\log n)$.
\end{proof}

\section{Revisiting Random Access using SLPs in Optimal Time}\label{sec:1d-opt}

In this section, we generalize the structure from \cref{sec:1d-log-structure}
to support faster random access queries. Specifically, instead of creating two blocks
of size $2^{\LevelId}$ for each of the endpoints of every variable expansion, we
choose a parameter $\tau \geq 2$, and create $\tau$ blocks of size $\tau^{\LevelId}$
at each endpoint.

\subsection{The Data Structure}\label{sec:1d-opt-structure}

\paragraph{Definitions}

The arrays
$H^{\rm left}$, $O^{\rm left}$,
are now defined for every
$i \in [1 \dd |V|]$,
$\LevelId \in [0 \dd \ceil{\log_{\tau} n}]$, and
$k \in [0 \dd \tau)$
satisfying
$k \cdot \tau^{\LevelId} < |\Exp{G}{N_i}|$
such that it holds
\begin{align*}
  N_{H^{\rm left}[i,\LevelId,k]} &= \Hook{G}{N_i}{b}{e},\\
  O^{\rm left}[i,\LevelId,k] &= \Offset{G}{N_i}{b}{e},
\end{align*}
where
\begin{itemize}
\item $m = |\Exp{G}{N_i}|$,
\item $b = k \cdot \tau^{\LevelId}$, and
\item $e = \min(m, (k+1) \cdot \tau^{\LevelId})$.
\end{itemize}
The arrays $H^{\rm right}$ and $O^{\rm right}$
are generalized analogously.

\paragraph{Components}

The data structure consists of the following components:
\begin{enumerate}
\item The array $A_{\rm rhs}$ using $\bigO(|V|)$ space.
\item The array $A_{\rm expsize}$ using $\bigO(|V|)$ space.
\item The arrays
  $H^{\rm left}$, $O^{\rm left}$,
  $H^{\rm right}$, $O^{\rm right}$
  using
  $\bigO(|V| \cdot \tau \cdot \log_{\tau} n)$
  space.
\end{enumerate}
In total, the data structure needs
$\bigO(|V| \cdot \tau \cdot \log_{\tau} n) = \bigO(|G| \cdot \tau \cdot \log_{\tau} n)$
space.

\subsection{Implementation of Queries}\label{sec:1d-opt-queries}

\begin{definition}[Generalized left boundary mapping]\label{def:generalized-left-map}
  Let $N \in V$
  Denote $m = |\Exp{G}{N}|$.
  Consider any
  $\LevelId \in \Zn$, and
  $\delta \in [1 \dd m]$
  satisfying
  $\delta \leq \tau^{\LevelId+1}$.
  Denote
  $k = \lceil \tfrac{\delta}{\tau^{\LevelId}} \rceil - 1 \in [0 \dd \tau)$,
  $b = k \cdot \tau^{\LevelId}$,
  $e = \min(m, (k + 1) \cdot \tau^{\LevelId})$, and
  $H = \Hook{G}{N}{b}{e}$.
  The rest of the definition of the
  generalized left boundary mapping
  $\LeftMap{G}{N}{\LevelId}{\delta}$
  is as in \cref{def:left-map}.
\end{definition}

\noindent
The generalized boundary mapping
$\RightMap{G}{N}{\LevelId}{\delta}$
is defined analogously.

\begin{lemma}\label{lm:generalized-left-map}
  Let $N \in V$ and
  $m = |\Exp{G}{N}|$.
  Consider any
  $\LevelId \in \Zn$ and
  $\delta \in [1 \dd m]$
  satisfying
  $\delta \leq \tau^{\LevelId+1}$.
  Then, the tuple
  $(N', \delta', c) =
  \LeftMap{G}{N}{\LevelId}{\delta}$
  (\cref{def:generalized-left-map})
  satisfies:
  \begin{enumerate}
  \item
    $\delta' \in [1 \dd |\Exp{G}{N'}|]$,
  \item
    $\delta' \leq \tau^{\LevelId}$,
  \item
    $\Access{G}{N}{\delta}{\DirLeft} = \Access{G}{N'}{\delta'}{c}$ (\cref{def:access}).
  \end{enumerate}
  Moreover, if
  $\LevelId = 0$,
  then
  $|\Exp{G}{N'}| = 1$.
\end{lemma}
\begin{proof}
  The proof is analogous to the proof of \cref{lm:left-map}.
\end{proof}

\begin{proposition}\label{pr:1d-opt-query}
  Given the data structure from
  \cref{sec:1d-opt}
  and any
  $i \in [1 \dd n]$,
  we can compute
  $\Exp{G}{S}[i]$ in
  $\bigO(\log_{\tau} n)$
  time.
\end{proposition}
\begin{proof}
  The proof is analogous to the proof of \cref{pr:1d-log-query}, except
  instead of \cref{lm:left-map}, we use \cref{lm:generalized-left-map}.
\end{proof}

\begin{theorem}\label{th:1d-opt}
  For every SLP $G$ of size
  $|G| = g$ representing a string $T \in \Sigma^{n}$
  (see \cref{sec:prelim-grammars}),
  and any constant $\epsilon > 0$,
  there exists a data structure of size
  $\bigO(g \cdot \log^{1+\epsilon} n)$,
  that, given any
  $i \in [1 \dd n]$,
  returns
  $T[i]$
  in
  $\bigO(\tfrac{\log n}{\log \log n})$ time.
\end{theorem}
\begin{proof}
  We let $\tau = \log^{\epsilon} n$ and use the data structure from
  \cref{sec:1d-opt}.
  It needs
  $\bigO(g \cdot \tau \cdot \log_{\tau} n) =
  \bigO(g \cdot \log^{1+\epsilon} n)$ space.
  By \cref{pr:1d-opt-query}, the query time
  is then $\bigO(\log_{\tau} n) = \bigO(\tfrac{\log n}{\log \tau})
  = \bigO(\tfrac{\log n}{\log (\log^{\epsilon} n)}) =
  \bigO(\tfrac{\log n}{\log \log n})$.
\end{proof}

\section{Random Access using 2D SLPs in Logarithmic Time}\label{sec:2d-log}

Let
$G = (V_l, V_h, V_v, \Sigma, R, S)$
be a 2D SLP (see \cref{sec:prelim-2d-grammars})
fixed for the duration of this section.
Denote
$n = \max(\Rows{\Exp{G}{S}}, \Cols{\Exp{G}{S}})$
and
$V = V_l \cup V_h \cup V_v = \{N_{1}, \ldots, N_{|V|}\}$.
In this section, we generalize the data structure from \cref{sec:1d-log}
to 2D SLPs, i.e.,
we show that there exists a data structure using
$\bigO(|G| \log^2 n)$ space
that implements random access queries to $\Exp{G}{S}$
in $\bigO(\log n)$ time.

\subsection{Preliminaries}\label{sec:2d-log-prelim}

\begin{definition}[2D hook]\label{def:2d-hook}
  Consider any nonterminal $N \in V$. Let $b_r, e_r \in [0 \dd \Rows{\Exp{G}{N}}]$ and
  $b_c, e_c \in [0 \dd \Cols{\Exp{G}{N}}]$ be such that $b_r < e_r$
  and $b_c < e_c$.
  \begin{enumerate}
    \item If $N \in V_{l}$, then we define $\HookTwoDim{G}{N}{b_r}{b_c}{e_r}{e_c} = N$.
    \item Otherwise, if $N \in V_h$, then letting $A,B \in V$ be such that
      $\Rhs{G}{N} = AB$ and $\ell = \Rows{\Exp{G}{A}}$, we define:
      \[
        \HookTwoDim{G}{N}{b_r}{b_c}{e_r}{e_c} =
          \begin{cases}
            N & \text{if }b_r < \ell < e_r,\\
            \HookTwoDim{G}{A}{b_r}{b_c}{e_r}{e_c} & \text{if }e_r \leq \ell,\\
            \HookTwoDim{G}{B}{b_r - \ell}{b_c}{e_r - \ell}{e_c} & \text{if }\ell \leq b_r.
          \end{cases}
      \]
    \item Finally, if $N \in V_v$ then, letting $A,B \in V$ be such that
      $\Rhs{G}{N} = AB$ and $\ell = \Cols{\Exp{G}{A}}$, we define:
      \[
        \HookTwoDim{G}{N}{b_r}{b_c}{e_r}{e_c} =
          \begin{cases}
            N & \text{if }b_c < \ell < e_c,\\
            \HookTwoDim{G}{A}{b_r}{b_c}{e_r}{e_c} & \text{if }e_c \leq \ell,\\
            \HookTwoDim{G}{B}{b_r}{b_c-\ell}{e_r}{e_c-\ell} & \text{if }\ell \leq b_c.
          \end{cases}
      \]
  \end{enumerate}
\end{definition}

\begin{definition}[2D offset]\label{def:2d-offset}
  Consider any nonterminal $N \in V$.
  Let $b_r, e_r \in [0 \dd \Rows{\Exp{G}{N}}]$ and
  $b_c, e_c \in [0 \dd \Cols{\Exp{G}{N}}]$ be such that $b_r < e_r$ and $b_c < e_c$.
  \begin{enumerate}
    \item If $N \in V_{l}$, then we define $\OffsetTwoDim{G}{N}{b_r}{b_c}{e_r}{e_c} = (0,0)$.
    \item Otherwise, if $N \in V_h$, then letting $A,B \in V$ be such that
      $\Rhs{G}{N} = AB$ and $\ell = \Rows{\Exp{G}{A}}$, we define:
      \[
        \OffsetTwoDim{G}{N}{b_r}{b_c}{e_r}{e_c} =
          \begin{cases}
            (b_r,b_c) & \text{if }b_r < \ell < e_r,\\
            \OffsetTwoDim{G}{A}{b_r}{b_c}{e_r}{e_c} & \text{if }e_r \leq \ell,\\
            \OffsetTwoDim{G}{B}{b_r - \ell}{b_c}{e_r - \ell}{e_c} & \text{if }\ell \leq b_r.
          \end{cases}
      \]
    \item Finally, if $N \in V_v$ then, letting $A,B \in V$ be such that
      $\Rhs{G}{N} = AB$ and $\ell = \Cols{\Exp{G}{A}}$, we define:
      \[
        \OffsetTwoDim{G}{N}{b_r}{b_c}{e_r}{e_c} =
          \begin{cases}
            (b_r,b_c) & \text{if }b_c < \ell < e_c,\\
            \OffsetTwoDim{G}{A}{b_r}{b_c}{e_r}{e_c} & \text{if }e_c \leq \ell,\\
            \OffsetTwoDim{G}{B}{b_r}{b_c-\ell}{e_r}{e_c-\ell} & \text{if }\ell \leq b_c.
          \end{cases}
      \]
  \end{enumerate}
\end{definition}

\begin{figure}[t!]
	\centering
	\begin{tikzpicture}[scale=0.35]
    \small
		\draw[black, line width = 0.8mm] (0,0) rectangle (20,20);
    \draw[black, line width = 0.9mm] (4.97,0) -- (4.97,20);
		\draw[blue, line width = 0.8mm] (5.2,0.2) rectangle (19.77,19.77);
    \draw[blue, line width = 0.9mm] (5.2,15.04) -- (19.8,15.04);
		\draw[red, line width = 0.8mm] (5.4,0.4) rectangle (19.6,14.8);
		\draw[red, line width = 0.8mm] (12,0.4) -- (12,14.7);
		\draw[dotted, teal, line width = 0.6mm] (8,8) rectangle (14,12);

		\draw[stealth-stealth] (0,21.2) -- (8,21.2) node[midway,yshift=0.6em]{$b_c$};
		\draw[stealth-stealth] (0,22.4) -- (14,22.4) node[midway,yshift=0.6em]{$e_c$};
		\draw[stealth-stealth] (-1.2,20) -- (-1.2,12) node[midway,xshift=-0.7em]{$b_r$};
		\draw[stealth-stealth] (-2.4,20) -- (-2.4,8) node[midway,xshift=-0.7em]{$e_r$};

		\draw[stealth-stealth] (5.5,10) -- (8,10) node[midway,yshift=0.6em]{$\alpha_c$};
		\draw[stealth-stealth] (10,12) -- (10,14.6) node[midway,xshift=-0.7em]{$\alpha_r$};

		\draw[decorate,decoration={brace,amplitude=15pt},line width = 0.4mm]
      (-3,0) -- (-3,20) node[midway,xshift=-4em]{$\Exp{G}{N}$};
		\draw[red,decorate,decoration={brace,amplitude=12pt,mirror},line width = 0.4mm]
      (20.5,0.4) -- (20.5,14.8) node[midway,xshift=4em]{$\Exp{G}{H}$};

		\draw [dashed,gray, line width = 0.1mm] (8,12) -- (8,20);
		\draw [dashed,gray, line width = 0.1mm] (14,12) -- (14,20);
		\draw [dashed,gray, line width = 0.1mm] (0,8) -- (8,8);
		\draw [dashed,gray, line width = 0.1mm] (0,12) -- (8,12);
	\end{tikzpicture}
	\caption{Illustration of \cref{def:2d-hook,def:2d-offset}. Here $N \in V$, $\HookTwoDim{G}{N}{b_r}{b_c}{e_r}{e_c} = H$, and $\OffsetTwoDim{G}{N}{b_r}{b_c}{e_r}{e_c} = (\alpha_r,\alpha_c)$.}
	\label{fig:2d-hook-and-offset}
\end{figure}

\begin{lemma}\label{lm:2d-hook}
  Let $N \in V$, $w = \Exp{G}{N}$, and $b_r, b_e \in [0 \dd \Rows{w}]$ and
  $b_c, e_c \in [0 \dd \Cols{w}]$ be such that $b_r < e_r$ and $b_c < e_c$.
  Denote
  $H = \HookTwoDim{G}{N}{b_r}{b_c}{e_r}{e_c}$ (\cref{def:2d-hook}),
  $(\alpha_r, \alpha_c) = \OffsetTwoDim{G}{N}{b_r}{b_c}{e_r}{e_c}$ (\cref{def:2d-offset}), and
  $(\beta_r, \beta_c) = (\alpha_r + (e_r - b_r), \alpha_c + (e_c - b_c))$.
  Then, it holds
  \begin{itemize}
  \item $\alpha_r \leq b_r$,
  \item $\alpha_c \leq b_c$,
  \item $0 \leq \alpha_r < \beta_r \leq \Rows{\Exp{G}{H}}$,
  \item $0 \leq \alpha_c < \beta_c \leq \Cols{\Exp{G}{H}}$, and
  \item $w(b_r \dd e_r](b_c \dd e_c] = \Exp{G}{H}(\alpha_r \dd \beta_r](\alpha_c \dd \beta_c]$.
  \end{itemize}
  Moreover,
  \begin{enumerate}
  \item\label{lm:2d-hook-it-1}
    If $e_r - b_r = 1$ and $e_c - b_c = 1$, then $H \in V_{l}$.
  \item\label{lm:2d-hook-it-2}
    Otherwise, it holds $H \in V_{h} \cup V_{v}$.
    Furthermore, letting $X, Y \in V$ be such that
    $\Rhs{G}{H} = XY$,
    \begin{itemize}
    \item $H \in V_{h}$ implies that it holds $\alpha_r < \ell < \beta_r$, where $\ell = \Rows{\Exp{G}{X}}$.
    \item $H \in V_{v}$ implies that it holds $\alpha_c < \ell < \beta_c$, where $\ell = \Cols{\Exp{G}{X}}$.
    \end{itemize}
  \end{enumerate}
\end{lemma}
\begin{proof}

  By \cref{def:2d-hook,def:2d-offset},
  there exists a unique sequence of tuples
  $(N_i, b_{r,i}, b_{c,i}, e_{r,i}, e_{c,i})_{i \in [0 \dd q]}$,
  where $q \geq 0$, such that:
  \begin{itemize}
  \item for every $i \in [0 \dd q]$,
  \begin{itemize}
    \item $0 \leq b_{r,i} < e_{r,i} \leq \Rows{\Exp{G}{N_i}}$,
    \item $0 \leq b_{c,i} < e_{c,i} \leq \Cols{\Exp{G}{N_i}}$,
  \end{itemize}
  \item $(N_{0}, b_{r,0}, b_{c,0}, e_{r,0}, e_{c,0}) = (N, b_r, b_c, e_r, e_c)$,
  \item for every $i \in [0 \dd q)$,
    \begin{itemize}
    \item $N_i \in V_{h} \cup V_{v}$
      and, letting $X_i, Y_i \in V$ be such that
      $\Rhs{G}{N_i} = X_i Y_i$, $N_{i+1} \in \{X_i, Y_i\}$,
    \item $\HookTwoDim{G}{N_i}{b_{r,i}}{b_{c,i}}{e_{r,i}}{e_{c,i}} =
      \HookTwoDim{G}{N_{i+1}}{b_{r,i+1}}{b_{c,i+1}}{e_{r,i+1}}{e_{c,i+1}}$,
    \item $\OffsetTwoDim{G}{N_i}{b_{r,i}}{b_{c,i}}{e_{r,i}}{e_{c,i}} =
      \OffsetTwoDim{G}{N_{i+1}}{b_{r,i+1}}{b_{c,i+1}}{e_{r,i+1}}{e_{c,i+1}}$,
    \end{itemize}
  \item $\HookTwoDim{G}{N_q}{b_{r,q}}{b_{c,q}}{e_{r,q}}{e_{c,q}} = N_{q}$ and
    $\OffsetTwoDim{G}{N_q}{b_{r,q}}{b_{c,q}}{e_{r,q}}{e_{c,q}} = (b_{r,q},b_{c,q})$.
  \end{itemize}

  To prove the main claim, we first prove auxiliary properties of the above
  sequence. More specifically, we will first show that for every $i \in [0 \dd q)$,
  it holds:
  \begin{itemize}
  \item $b_{r,i+1} \leq b_{r,i}$,
  \item $b_{c,i+1} \leq b_{c,i}$,
  \item $e_{r,i+1} - b_{r,i+1} = e_{r,i} - b_{r,i}$,
  \item $e_{c,i+1} - b_{c,i+1} = e_{c,i} - b_{c,i}$,
  \item $\Exp{G}{N_{i+1}}(b_{r,i+1} \dd e_{r,i+1}](b_{c,i+1} \dd e_{c,i+1}] =
    \Exp{G}{N_i}(b_{r,i} \dd e_{r,i}](b_{c,i} \dd e_{c,i}]$.
  \end{itemize}
  Let $i \in [0 \dd q)$. As noted above, it holds
  $N_i \in V_h \cup V_v$.
  Let $X_i, Y_i \in V$ be such that $\Rhs{G}{N_i} = X_i Y_i$.
  We consider two cases:
  \begin{itemize}
  \item First, assume that $N_i \in V_h$. Denote $\ell = \Rows{\Exp{G}{X_i}}$.
    Above, we observed that $N_{i+1} \in \{X_i, Y_i\}$ and
    $\HookTwoDim{G}{N_i}{b_{r,i}}{b_{c,i}}{e_{r,i}}{e_{c,i}} =
    \HookTwoDim{G}{N_{i+1}}{b_{r,i+1}}{b_{c,i+1}}{e_{r,i+1}}{e_{c,i+1}}$.
    By \cref{def:2d-hook}, this implies that we must either have
    $e_{r,i} \leq \ell$ or $\ell \leq b_{r,i}$. Consider two cases:
    \begin{itemize}
    \item First, assume that $e_{r,i} \leq \ell$. By \cref{def:2d-hook}, we
      have $(N_{i+1}, b_{r,i+1}, b_{c,i+1}, e_{r,i+1}, e_{c,i+1}) =
      (X_i, b_{r,i}, b_{c,i}, e_{r,i}, e_{c,i})$.
      We thus immediately obtain
      $b_{r,i+1} \leq b_{r,i}$,
      $b_{c,i+1} \leq b_{c,i}$,
      $e_{r,i+1} - b_{r,i+1} = e_{r,i} - b_{r,i}$, and
      $e_{c,i+1} - b_{c,i+1} = e_{c,i} - b_{c,i}$.
      On the other hand,
      $e_{r,i} \leq \ell$ implies that it holds $\Exp{G}{N_i}(b_{r,i} \dd e_{r,i}](b_{c,i} \dd e_{c,i}]
       = \Exp{G}{X_i}(b_{r,i} \dd e_{r,i}](b_{c,i} \dd e_{c,i}]$. Putting everything
      together, we thus obtain
      \begin{align*}
        \Exp{G}{N_{i+1}}(b_{r,i+1} \dd e_{r,i+1}](b_{c,i+1} \dd e_{c,i+1}]
          &= \Exp{G}{X_i}(b_{r,i} \dd e_{r,i}](b_{c,i} \dd e_{c,i}]\\
          &= \Exp{G}{N_i}(b_{r,i} \dd e_{r,i}](b_{c,i} \dd e_{c,i}].
      \end{align*}
    \item Assume now that $\ell \leq b_{r,i}$.
      By \cref{def:2d-hook}, we
      have $(N_{i+1}, b_{r,i+1}, b_{c,i+1}, e_{r,i+1}, e_{c,i+1}) =
      (Y_i, b_{r,i} - \ell, b_{c,i}, e_{r,i} - \ell, e_{c,i})$.
      We thus again obtain that
      $b_{r,i+1} = b_{r,i} - \ell \leq b_{r,i}$,
      $b_{c,i+1} \leq b_{c,i}$,
      $e_{r,i+1} - b_{r,i+1} = (e_{r,i} - \ell) - (b_{r,i} - \ell) = e_{r,i} - b_{r,i}$,
      and $e_{c,i+1} - b_{c,i+1} = e_{c,i} - b_{c,i}$.
      On the other hand,
      $\ell \leq b_{r,i}$ implies that it holds $\Exp{G}{N_i}(b_{r,i} \dd e_{r,i}](b_{c,i} \dd e_{c,i}]
       = \Exp{G}{Y_i}(b_{r,i} - \ell \dd e_{r,i} - \ell](b_{c,i} \dd e_{c,i}]$. Putting everything
      together, we thus obtain
      \begin{align*}
        \Exp{G}{N_{i+1}}(b_{r,i+1} \dd e_{r,i+1}](b_{c,i+1} \dd e_{c,i+1}]
          &= \Exp{G}{Y_i}(b_{r,i} - \ell \dd e_{r,i} - \ell](b_{c,i} \dd e_{c,i}]\\
          &= \Exp{G}{N_i}(b_{r,i} \dd e_{r,i}](b_{c,i} \dd e_{c,i}].
      \end{align*}
    \end{itemize}
  \item The proof of the case $N_i \in V_v$ is analogous.
  \end{itemize}

  Using the properties above, we are now ready to prove the main claims.
  First, observe that, from the properties of the sequence
  $(N_i, b_{r,i}, b_{c,i}, e_{r,i}, e_{c,i})_{i \in [0 \dd q]}$,
  we immediately obtain that $H = N_q$, $(\alpha_r, \alpha_c) = (b_{r,q}, b_{c,q})$, and
  \begin{align*}
    (\beta_r, \beta_c)
      &= (\alpha_r + (e_r - b_r), \alpha_c + (e_c - b_c))\\
      &= (b_{r,q} + (e_{r,q} - b_{r,q}), b_{c,q} + (e_{c,q} - b_{c,q}))\\
      &= (e_{r,q}, e_{c,q}).
  \end{align*}
  This immediately implies the first four claims, i.e.,
  $\alpha_r \leq b_r$,
  $\alpha_c \leq b_c$,
  $0 \leq \alpha_r < \beta_r \leq \Rows{\Exp{G}{H}}$, and
  $0 \leq \alpha_c < \beta_c \leq \Cols{\Exp{G}{H}}$,
  since by the above we have
  $\alpha_r = b_{r,q} \leq b_{r,q-1} \leq \ldots \leq b_{r,0} = b_r$,
  $\alpha_c = b_{c,q} \leq b_{c,q-1} \leq \ldots \leq b_{c,0} = b_c$,
  $0 \leq b_{r,q} < e_{r,q} \leq \Rows{\Exp{G}{N_{q}}}$, and
  $0 \leq b_{c,q} < e_{c,q} \leq \Cols{\Exp{G}{N_{q}}}$.
  To obtain the fifth
  claim note that, by the above, it holds
  \begin{align*}
    w(b_r \dd e_r](b_c \dd e_c]
      &= \Exp{G}{N}(b_r \dd e_r](b_c \dd e_c]\\
      &= \Exp{G}{N_{0}}(b_{r,0} \dd e_{r,0}](b_{c,0} \dd e_{c,0}]\\
      &= \Exp{G}{N_{1}}(b_{r,1} \dd e_{r,1}](b_{c,1} \dd e_{c,1}]\\
      &= \dots\\
      &= \Exp{G}{N_{q}}(b_{r,q} \dd e_{r,q}](b_{c,q} \dd e_{c,q}]\\
      &= \Exp{G}{H}(\alpha_r \dd \beta_r](\alpha_c \dd \beta_c].\\
  \end{align*}

  We now prove the second part of the main claim, considering each of the items separately:
  \begin{enumerate}

  \item First, assume that $e_r - b_r = 1$ and $e_c - b_c = 1$.
    Suppose the claim does not hold, i.e.,
    $H \not\in V_l$. Then, we have $H \in V_h \cup V_v$.
    By the above, this is equivalent
    to $N_{q} \in V_h \cup V_v$.
    Then, there exist $X,Y \in V$ such that
    $\Rhs{G}{N_{q}} = XY$. Assume that $N_{q} \in V_h$.
    Denote $\ell = \Rows{\Exp{G}{X}}$.
    Observe that we cannot have
    $b_{r,q} < \ell < e_{r,q}$, since that implies
    $e_{r,q} - b_{r,q} > 1$, which contradicts $e_r - b_r = 1$
    (recall that $e_{r,q} - b_{r,q} = e_r - b_r$).
    Thus, we must have either $\ell \leq b_{r,q}$ or $e_{r,q} \leq \ell$.
    In either case, by \cref{def:2d-hook}, it holds
    $\HookTwoDim{G}{N_{q}}{b_{r,q}}{b_{r,q}}{e_{r,q}}{e_{c,q}} \neq N_{q}$,
    a contradiction. The proof in the case $N_{q} \in V_v$ is analogous,
    also leading to contradiction.
    We thus cannot have $N_{q} \in V_h \cup V_v$, and hence
    it holds $N_{q} \in V_l$, or equivalently, $H \in V_l$.

  \item Let us now assume that $e_r - b_r > 1$ or $e_c - b_c > 1$.
    Recall that above we observed that
    $H = N_{q}$,
    $e_{r,q} - b_{r,q} = e_r - b_r$,
    $e_{c,q} - b_{c,q} = e_c - b_c$,
    $0 \leq b_{r,q} < e_{r,q} \leq \Rows{\Exp{G}{N_{q}}}$, and
    $0 \leq b_{c,q} < e_{c,q} \leq \Cols{\Exp{G}{N_{q}}}$.
    This immediately implies that
    $\Rows{\Exp{G}{H}} = \Rows{\Exp{G}{N_{q}}} \geq e_{r,q} - b_{r,q} = e_r - b_r$ and
    $\Cols{\Exp{G}{H}} = \Cols{\Exp{G}{N_{q}}} \geq e_{c,q} - b_{c,q} = e_c - b_c$.
    Consequently, by the assumptions about
    $e_r - b_r$ and $e_c - b_c$,
    we either have
    $\Rows{\Exp{G}{H}} > 1$ or $\Cols{\Exp{G}{H}} > 1$,
    and thus $H \in V_h \cup V_v$, i.e.,
    the first part of the claim.
    Let then $X, Y \in V$ be as in the claim,
    i.e., such that $\Rhs{G}{H} = XY$.
    Consider two cases:
    \begin{itemize}
    \item Let us first assume $H \in V_h$.
      Denote $\ell = \Rows{\Exp{G}{X}}$. Recall that above we proved
      that in this case, we cannot have $e_{r,q} \leq \ell$ or $\ell \leq b_{r,q}$, because
      each case leads to a contradiction. Thus, we must have $b_{r,q} < \ell < e_{r,q}$.
      It remains to recall that above we also showed that $\alpha_r = b_{r,q}$
      and $\beta_r = e_{r,q}$. Thus, we obtain $\alpha_r < \ell < \beta_r$, i.e., the claim.
    \item The case $H \in V_v$ analogously yields $\alpha_c < \ell < \beta_c$,
      where $\ell = \Cols{\Exp{G}{X}}$.
      \qedhere
    \end{itemize}
  \end{enumerate}
\end{proof}

\subsection{The Data Structure}\label{sec:2d-log-structure}

\paragraph{Definitions}

Let $H^{NW}$ and $O^{NW}$ denote 5-dimensional arrays defined
so that, for every
$i \in [1 \dd |V|]$,
$\LevelId_{r}, \LevelId_{c} \in [0 \dd \ceil{\log n}]$, and
$k_r, k_c \in \{0,1\}$
satisfying
$k_r \cdot 2^{\LevelId_r} < \Rows{\Exp{G}{N_i}}$ and
$k_c \cdot 2^{\LevelId_c} < \Cols{\Exp{G}{N_i}}$,
it holds
\begin{align*}
  N_{H^{NW}[i,\LevelId_r,\LevelId_c,k_r,k_c]} &= \HookTwoDim{G}{N_i}{b_r}{b_c}{e_r}{e_c},\\
  O^{NW}[i,\LevelId_r,\LevelId_c,k_r,k_c] &= \OffsetTwoDim{G}{N_i}{b_r}{b_c}{e_r}{e_c},
\end{align*}
where
\begin{itemize}
\item $m_r = \Rows{\Exp{G}{N_i}}$,
\item $m_c = \Cols{\Exp{G}{N_i}}$,
\item $b_r = k_r \cdot 2^{\LevelId_r}$,
\item $b_c = k_c \cdot 2^{\LevelId_c}$,
\item $e_r = \min(m_r, (k_r+1) \cdot 2^{\LevelId_r})$, and
\item $e_c = \min(m_c, (k_c+1) \cdot 2^{\LevelId_c})$.
\end{itemize}
Similarly, for every $i, \LevelId_r, \LevelId_c, k_r, k_c$ as above, we define
$H^{NE}$ and $O^{NE}$ such that it holds
\begin{align*}
  N_{H^{NE}[i,\LevelId_r,\LevelId_c,k_r,k_c]} &= \HookTwoDim{G}{N_i}{b_r}{m_c - e_c}{e_r}{m_c - b_c},\\
  O^{NE}[i,\LevelId_r,\LevelId_c,k_r,k_c] &= \OffsetTwoDim{G}{N_i}{b_r}{m_c - e_c}{e_r}{m_c - b_c},
\end{align*}
where $m_r, m_c, b_r, b_c, e_r, e_c$ are as above.

We analogously define the arrays $H^{SW}$, $O^{SW}$, $H^{SE}$, and $O^{SE}$.
By $A_{\rm rows}[1 \dd |V|]$ and $A_{\rm cols}[1 \dd |V|]$ we denote arrays defined by
$A_{\rm rows}[i] = \Rows{\Exp{G}{N_i}}$ and
$A_{\rm cols}[i] = \Cols{\Exp{G}{N_i}}$.
By $A_{\rm rhs}[1 \dd |V|]$ we denote an array defined such that,
for $i \in [1 \dd |V|]$,
\begin{itemize}
\item If $|\Rhs{G}{N_i}| = 1$, then $A_{\rm rhs}[i] = \Rhs{G}{N_i}$,
\item Otherwise, $A_{\rm rhs}[i] = (i',i'')$, where $i', i'' \in [1 \dd |V|]$
  are such that $\Rhs{G}{N_i} = N_{i'} N_{i''}$.
\end{itemize}
By $A_{\rm horiz}[1 \dd |V|]$ we denote an array defined so
that, for every $i \in [1 \dd |V|]$, $A_{\rm horiz}[i] \in \{0,1\}$ is
such that $A_{\rm horiz}[i] = 1$ holds if and only if $N_i \in V_h$.

\paragraph{Components}

The data structure consists of the following components:
\begin{enumerate}
\item The array $A_{\rm rhs}$ using $\bigO(|V|)$ space.
\item The arrays $A_{\rm rows}$ and $A_{\rm cols}$ using $\bigO(|V|)$ space.
\item The array $A_{\rm horiz}$ using $\bigO(|V|)$ space.
\item The arrays
  $H^{NW}$, $O^{NW}$,
  $H^{NE}$, $O^{NE}$,
  $H^{SW}$, $O^{SW}$,
  $H^{SE}$, $O^{SE}$
  using
  $\bigO(|V| \log^2 n)$
  space.
\end{enumerate}
In total, the data structure needs
$\bigO(|V| \log^2 n) = \bigO(|G| \log^2 n)$
space.

\subsection{Implementation of Queries}\label{sec:2d-log-queries}

\begin{definition}[2D access function]\label{def:2d-access}
  Let $N \in V$,
  $m_r = \Rows{\Exp{G}{N}}$, and
  $m_c = \Cols{\Exp{G}{N}}$.
  For every
  $\delta_r \in [1 \dd m_r]$,
  $\delta_c \in [1 \dd m_c]$,
  $c_r \in \{\DirTop, \DirBottom\}$, and
  $c_c \in \{\DirLeft, \DirRight\}$,
  we define
  \vspace{1ex}
  \[
    \AccessTwoDim{G}{N}{\delta_r}{\delta_c}{c_r}{c_c} =
      \begin{cases}
        \Exp{G}{N}[\delta_r, \delta_c]                     & \text{if } (c_r, c_c) = (\DirTop, \DirLeft),\\
        \Exp{G}{N}[\delta_r, m_c - \delta_c + 1]           & \text{if } (c_r, c_c) = (\DirTop, \DirRight),\\
        \Exp{G}{N}[m_r - \delta_r + 1, \delta_c]           & \text{if } (c_r, c_c) = (\DirBottom, \DirLeft),\\
        \Exp{G}{N}[m_r - \delta_r + 1, m_c - \delta_c + 1] & \text{if } (c_r, c_c) = (\DirBottom, \DirRight).\\
      \end{cases}
  \]
  \vspace{1ex}
\end{definition}

\begin{definition}[Top-left corner mapping]\label{def:top-left-map}
  Let $N \in V$.
  Denote
  $m_r = \Rows{\Exp{G}{N}}$ and
  $m_c = \Cols{\Exp{G}{N}}$.
  Consider any
  $\LevelId_r, \LevelId_c \in \Zn$,
  $\delta_r \in [1 \dd m_r]$, and
  $\delta_c \in [1 \dd m_c]$
  satisfying
  $\delta_r \leq 2^{\LevelId_r+1}$ and
  $\delta_c \leq 2^{\LevelId_c+1}$.
  Denote
  $k_r = \lceil \tfrac{\delta_r}{2^{\LevelId_r}} \rceil - 1 \in \{0,1\}$,
  $b_r = k_r \cdot 2^{\LevelId_r}$,
  $e_r = \min(m_r, (k_r + 1) \cdot 2^{\LevelId_r})$,
  $k_c = \lceil \tfrac{\delta_c}{2^{\LevelId_c}} \rceil - 1 \in \{0,1\}$,
  $b_c = k_c \cdot 2^{\LevelId_c}$,
  $e_c = \min(m_c, (k_c + 1) \cdot 2^{\LevelId_c})$, and
  $H = \HookTwoDim{G}{N}{b_r}{b_c}{e_r}{e_c}$.
  \begin{itemize}
  \item
    If $e_r - b_r = 1$ and $e_c - b_c = 1$,
    then we define
    \[
      \TopLeftMap{G}{N}{\LevelId_r}{\LevelId_c}{\delta_r}{\delta_c} = (H, 1, 1, \DirTop, \DirLeft).
    \]
  \item
    Otherwise, letting
    $(\alpha_r, \alpha_c) = \OffsetTwoDim{G}{N}{b_r}{b_c}{e_r}{e_c}$ and
    $X, Y \in V$ be such that $\Rhs{G}{H} = XY$,
    we consider two cases:
    \begin{itemize}
    \item If $H \in V_h$, then, letting $\ell = \Rows{\Exp{G}{X}}$,
      we define $\TopLeftMap{G}{N}{\LevelId_r}{\LevelId_c}{\delta_r}{\delta_c}$ as:
      \vspace{1ex}
      \[
        \begin{cases}
          (X, (\ell - \alpha_r) - (\delta_r - b_r) + 1, \alpha_c + (\delta_c - b_c), \DirBottom, \DirLeft) & \text{if }\delta_r - b_r \leq \ell - \alpha_r,\\
          (Y, (\delta_r - b_r) - (\ell - \alpha_r), \alpha_c + (\delta_c - b_c), \DirTop, \DirLeft)        & \text{otherwise}.\\
        \end{cases}
      \]
      \vspace{1ex}
    \item If $H \in V_v$, then, letting $\ell = \Cols{\Exp{G}{X}}$,
      we define $\TopLeftMap{G}{N}{\LevelId_r}{\LevelId_c}{\delta_r}{\delta_c}$ as:
      \vspace{1ex}
      \[
        \begin{cases}
          (X, \alpha_r + (\delta_r - b_r), (\ell - \alpha_c) - (\delta_c - b_c) + 1, \DirTop, \DirRight) & \text{if }\delta_c - b_c \leq \ell - \alpha_c,\\
          (Y, \alpha_r + (\delta_r - b_r), (\delta_c - b_c) - (\ell - \alpha_c), \DirTop, \DirLeft)     & \text{otherwise}.\\
        \end{cases}
      \]
      \vspace{1ex}
    \end{itemize}
  \end{itemize}
\end{definition}

\begin{remark}\label{rm:top-left-map}
  Note that $X$ and $Y$ in \cref{def:top-left-map} are well-defined when
  $e_r - b_r > 1$ or $e_c - b_c > 1$,
  since in this case we have
  $H \in V_h \cup V_v$ (\cref{lm:2d-hook}).
\end{remark}

\noindent
We analogously define the other three corner mappings:
\begin{itemize}
\item $\TopRightMap{G}{N}{\LevelId_r}{\LevelId_c}{\delta_r}{\delta_c}$,
\item $\BottomLeftMap{G}{N}{\LevelId_r}{\LevelId_c}{\delta_r}{\delta_c}$, and
\item $\BottomRightMap{G}{N}{\LevelId_r}{\LevelId_c}{\delta_r}{\delta_c}$.
\end{itemize}

\begin{lemma}\label{lm:top-left-map}
  Let $N \in V$,
  $m_r = \Rows{\Exp{G}{N}}$, and
  $m_c = \Cols{\Exp{G}{N}}$.
  Consider any $\LevelId_r, \LevelId_c \in \Zn$,
  $\delta_r \in [1 \dd m_r]$ and
  $\delta_c \in [1 \dd m_c]$
  satisfying
  $\delta_r \leq 2^{\LevelId_r+1}$ and
  $\delta_c \leq 2^{\LevelId_c+1}$.
  Then, the tuple
  $(N', \delta_r', \delta_c', c_r, c_c) = \TopLeftMap{G}{N}{\LevelId_r}{\LevelId_c}{\delta_r}{\delta_c}$
  (\cref{def:top-left-map})
  satisfies:
  \begin{enumerate}
  \item
    $\delta_r' \in [1 \dd \Rows{\Exp{G}{N'}}]$ and
    $\delta_c' \in [1 \dd \Cols{\Exp{G}{N'}}]$,
  \item
    ($\delta_r' \leq 2^{\LevelId_r}$ and $\delta_c' \leq \delta_c$) or
    ($\delta_c' \leq 2^{\LevelId_c}$ and $\delta_r' \leq \delta_r$), and
  \item
    $\AccessTwoDim{G}{N}{\delta_r}{\delta_c}{\DirTop}{\DirLeft} =
    \AccessTwoDim{G}{N'}{\delta_r'}{\delta_c'}{c_r}{c_c}$ (\cref{def:2d-access}).
  \end{enumerate}
  Moreover, if $\LevelId_r = 0$ and $\LevelId_c = 0$, then $N' \in V_l$.
\end{lemma}
\begin{proof}

  Denote:
  \begin{itemize}
  \item
    $k_r = \lceil \tfrac{\delta_r}{2^{\LevelId_r}} \rceil - 1 \in \{0,1\}$,
    $b_r = k_r \cdot 2^{\LevelId_r}$,
    $e_r = \min(m_r, (k_r + 1) \cdot 2^{\LevelId_r})$,
  \item
    $k_c = \lceil \tfrac{\delta_c}{2^{\LevelId_c}} \rceil - 1 \in \{0,1\}$,
    $b_c = k_c \cdot 2^{\LevelId_c}$,
    $e_c = \min(m_c, (k_c + 1) \cdot 2^{\LevelId_c})$,
  \item $H = \HookTwoDim{G}{N}{b_r}{b_c}{e_r}{e_c}$ (\cref{def:2d-hook}),
  \item $(\alpha_r, \alpha_c) = \OffsetTwoDim{G}{N}{b_r}{b_r}{e_r}{e_c}$ (\cref{def:2d-offset}), and
  \item $(\beta_r, \beta_c) = (\alpha_r + (e_r - b_r), \alpha_c + (e_c - b_c))$.
  \end{itemize}

  Observe that
  $b_r = k_r \cdot 2^{\LevelId_r} =
  (\lceil \tfrac{\delta_r}{2^{\LevelId_r}} \rceil - 1) \cdot 2^{\LevelId_r} =
  (\lfloor \tfrac{\delta_r + 2^{\LevelId_r} - 1}{2^{\LevelId_r}} \rfloor - 1) \cdot 2^{\LevelId_r} \leq
  \delta_r + 2^{\LevelId_r} - 1 - 2^{\LevelId_r} < \delta_r$.
  On the other hand, by
  $\delta_r \leq \lceil \tfrac{\delta_r}{2^{\LevelId_r}} \rceil 2^{\LevelId_r} = (k_r + 1) \cdot 2^{\LevelId_r}$
  and $\delta_r \leq m_r$,
  it follows that
  $\delta_r \leq \min(m_r, (k_r + 1) \cdot 2^{\LevelId_r}) = e_r$.
  Note also that
  $e_r - b_r \leq (k_r + 1) \cdot 2^{\LevelId_r} - k_r \cdot 2^{\LevelId_r} = 2^{\LevelId_r}$.
  Analogous calculations for $b_c$, $e_c$, and $\delta_c$, altogether establish
  that:
  \begin{itemize}
  \item $b_r < \delta_r \leq e_r$, $e_r - b_r \leq 2^{\LevelId_r}$,
  \item $b_c < \delta_c \leq e_c$, $e_c - b_c \leq 2^{\LevelId_c}$.
  \end{itemize}

  We prove each of the three main claims separately:
  \begin{enumerate}

  \item First, we prove that
    $\delta_r' \in [1 \dd \Rows{\Exp{G}{N'}}]$ and
    $\delta_c' \in [1 \dd \Cols{\Exp{G}{N'}}]$.
    If $e_r - b_r = 1$ and $e_c - b_c = 1$, then by \cref{def:top-left-map}, it holds
    $\delta_r' = 1$ and $\delta_c' = 1$, and the claim follows immediately.
    Let us thus assume that $e_r - b_r > 1$ or $e_c - b_c > 1$.
    Then, it holds $H \in V_h \cup V_v$ (see also \cref{rm:top-left-map}).
    Assume that $H \in V_{h}$ (the claim for $H \in V_{v}$ is proved analogously).
    \begin{itemize}

    \item We first prove that $\delta_c' \in [1 \dd \Cols{\Exp{G}{N'}}]$.
      To this end, first observe that by \cref{def:top-left-map}, we
      have $\delta_c' = \alpha_c + (\delta_c - b_c)$. By \cref{lm:2d-hook}, it holds
      $\alpha_c \geq 0$. On the other hand, above we noted that $\delta_c > b_c$.
      This yields $\delta_c' = \alpha_c + (\delta_c - b_c) \geq 1$.
      To show $\delta_c' \leq \Cols{\Exp{G}{N'}}$, we first observe that regardless
      of which of the two cases holds in \cref{def:top-left-map}, in the case $H \in V_h$
      we always have $\Cols{\Exp{G}{N'}} = \Cols{\Exp{G}{H}}$. Thus, by $\delta_c \leq e_c$
      and $\beta_c \leq \Exp{G}{H}$ (\cref{lm:2d-hook}), we obtain
      $\delta_c' = \alpha_c + (\delta_c - b_c) \leq \alpha_c + (e_c - b_c) =
      \beta_c \leq \Cols{\Exp{G}{H}} = \Cols{\Exp{G}{N'}}$.
      We have thus proved that $\delta_c' \in [1 \dd \Cols{\Exp{G}{N'}}]$.

    \item We now show that $\delta_r' \in [1 \dd \Rows{\Exp{G}{N'}}]$.    
      Let $X, Y \in V$ be such that $\Rhs{G}{H} = XY$ and
      let $\ell = \Rows{\Exp{G}{X}}$.
      We consider two cases:
      \begin{itemize}

      \item If $\delta_r - b_r \leq \ell - \alpha_r$ then, by \cref{def:top-left-map}, it holds
        $N' = X$ and $\delta_r' = (\ell - \alpha_r) - (\delta_r - b_r) + 1$.
        By the assumption $\delta_r - b_r \leq \ell - \alpha_r$, we immediately obtain
        $\delta_r' = (\ell - \alpha_r) - (\delta_r - b_r) + 1 \geq 1$.
        On the other hand, $b_r < \delta_r$ and $\alpha_r \geq 0$ (see \cref{lm:2d-hook}) imply
        $\delta_r' = (\ell - \alpha_r) - (\delta_r - b_r) + 1 \leq \ell = \Rows{\Exp{G}{X}} = \Rows{\Exp{G}{N'}}$.
        We have thus proved $\delta_r' \in [1 \dd \Rows{\Exp{G}{N'}}]$.

      \item Otherwise (i.e., if $\delta_r - b_r > \ell - \alpha_r$), by \cref{def:top-left-map}, it holds
        $N' = Y$ and $\delta_r' = (\delta_r - b_r) - (\ell - \alpha_r)$.
        The assumption $\delta_r - b_r > \ell - \alpha_r$ immediately yields
        $\delta_r' = (\delta_r - b_r) - (\ell - \alpha_r) \geq 1$.
        On the other hand, $\delta_r \leq e_r$ and
        $(e_r - b_r) - (\ell - \alpha_r) = ((e_r - b_r) + \alpha_r) - \ell = \beta_r - \ell \leq
        \Rows{\Exp{G}{H}} - \ell = \Rows{\Exp{G}{Y}}$
        (where $\beta_r \leq \Rows{\Exp{G}{H}}$ follows by \cref{lm:2d-hook}) imply
        $\delta_r' = (\delta_r - b_r) - (\ell - \alpha_r) \leq (e_r - b_r) - (\ell - \alpha_r)
        \leq \Rows{\Exp{G}{Y}} = \Rows{\Exp{G}{N'}}$.
        We have thus proved $\delta_r' \in [1 \dd \Rows{\Exp{G}{N'}}]$.
      \end{itemize}
    \end{itemize}

  \item Second, we prove that
    ($\delta_r' \leq 2^{\LevelId_r}$ and $\delta_c' \leq \delta_c$) or
    ($\delta_c' \leq 2^{\LevelId_c}$ and $\delta_r' \leq \delta_r$).
    If $e_r - b_r = 1$ and $e_c - b_c = 1$, then by \cref{def:top-left-map}, it holds
    $\delta_r' = 1$ and $\delta_c' = 1$,
    and the claim follows immediately.
    Let us thus assume that
    $e_r - b_r > 1$ or $e_c - b_c > 1$.
    Then, it holds
    $H \in V_h \cup V_v$ (see also \cref{rm:top-left-map}).
    Assume that $H \in V_{h}$ (the claim for $H \in V_{v}$ is proved analogously).
    We will show that in this case, it holds
    $\delta_r' \leq 2^{\LevelId_r}$ and $\delta_c' \leq \delta_c$.
    To show that $\delta_c' \leq \delta_c$, we first note that,
    by \cref{def:top-left-map}, in the case $H \in V_h$, we always have
    $\delta_c' = \alpha_c + (\delta_c - b_c)$. By \cref{lm:2d-hook},
    it holds $\alpha_c \leq b_c$.
    We therefore obtain
    $\delta_c' = \alpha_c + (\delta_c - b_c) = \delta_c - (b_c - \alpha_c) \leq \delta_c$.
    It remains to prove that $\delta_r' \leq 2^{\LevelId_r}$.
    Let $X, Y \in V$ be such that $\Rhs{G}{H} = XY$.
    Denote $\ell = \Rows{\Exp{G}{X}}$.
    We consider two cases:
    \begin{itemize}

    \item If $\delta_r - b_r \leq \ell - \alpha_r$ then, by \cref{def:top-left-map}, we have
      $\delta_r' = (\ell - \alpha_r) - (\delta_r - b_r) + 1$.
      Note that $\delta_r > b_r$ implies that
      $\delta_r' = (\ell - \alpha_r) - (\delta_r - b_r) + 1 \leq \ell - \alpha_r$.
      On the other hand, by
      $\ell < \beta_r$ (\cref{lm:2d-hook}) and the definition of $\beta_r$,
      we have $\ell - \alpha_r \leq \beta_r - \alpha_r = e_r - b_r$.
      Thus,
      $\delta_r' \leq \ell - \alpha_r \leq e_r - b_r \leq 2^{\LevelId_r}$.

    \item Otherwise (i.e., if $\delta_r - b_r > \ell - \alpha_r$), by \cref{def:top-left-map},
      we have $\delta_r' = (\delta_r - b_r) - (\ell - \alpha_r)$.
      To show $\delta_r' \leq 2^{\LevelId_r}$,
      it suffices to note that by $\ell > \alpha_r$ (\cref{lm:2d-hook})
      and $\delta_r' \leq (e_r - b_r) - (\ell - \alpha_r)$
      (shown above in the analogous case),
      it follows that
      $\delta_r' \leq (e_r - b_r) - (\ell - \alpha_r) \leq e_r - b_r \leq 2^{\LevelId_r}$.
    \end{itemize}

  \item Finally, we prove that it holds
    $\AccessTwoDim{G}{N}{\delta_r}{\delta_c}{\DirTop}{\DirLeft} =
    \AccessTwoDim{G}{N'}{\delta_r'}{\delta_c'}{c_r}{c_c}$.
    If $e_r - b_r = 1$ and $e_c - b_c = 1$, then by \cref{def:top-left-map},
    we have
    $(N', \delta_r', \delta_c', c_r, c_c) = (H, 1, 1, \DirTop, \DirLeft)$.
    Observe that by $b_r < \delta_r \leq e_r$ and $b_c < \delta_c \leq e_c$,
    we then must have
    $\delta_r = e_r$ and $\delta_c = e_c$.
    On the other hand, by
    \cref{lm:2d-hook},
    we then have
    $\Exp{G}{N}[e_r,e_c] = \Exp{G}{H}[1,1]$.
    Thus,
    \begin{align*}
      \AccessTwoDim{G}{N}{\delta_r}{\delta_c}{\DirTop}{\DirLeft}
        &= \AccessTwoDim{G}{N}{e_r}{e_c}{\DirTop}{\DirLeft}\\
        &= \Exp{G}{N}[e_r,e_c]\\
        &= \Exp{G}{H}[1,1]\\
        &= \Exp{G}{N'}[\delta_r',\delta_c']\\
        &= \AccessTwoDim{G}{N'}{\delta_r'}{\delta_c'}{\DirTop}{\DirLeft}\\
        &= \AccessTwoDim{G}{N'}{\delta_r'}{\delta_c'}{c_r}{c_c}.
    \end{align*}
    Let us now assume that $e_r - b_r > 1$ or $e_c - b_c > 1$.
    Then, it holds $H \in V_h \cup V_v$ (see also \cref{rm:top-left-map}).
    Assume that $H \in V_h$ (the proof in the case $H \in V_v$ is analogous).
    Let $X, Y \in V$ be such that $\Rhs{G}{H} = XY$.
    Denote $\ell = \Rows{\Exp{G}{X}}$.
    We consider two cases:
    \begin{itemize}

    \item If $\delta_r - b_r \leq \ell - \alpha_r$ then, by \cref{def:top-left-map}, we have
      $(N', \delta_r', \delta_c', c_r, c_c) =
      (X, (\ell - \alpha_r) - (\delta_r - b_r) + 1, \alpha_c + (\delta_c - b_c), \DirBottom, \DirLeft)$.
      Combining $\delta_r \leq b_r + (\ell - \alpha_r)$ (the assumption) with
      $b_r < \delta_r$ (proved above),
      we obtain
      $b_r < \delta_r \leq b_r + (\ell - \alpha_r)$.
      Furthermore, $b_c < \delta_c \leq e_c$.
      On the other hand, it holds
      $w(b_r \dd e_r](b_c \dd e_c] = \Exp{G}{H}(\alpha_r \dd \beta_r](\alpha_c \dd \beta_c]$
      and $\alpha_r < \ell < \beta_r$ (\cref{lm:2d-hook}).
      Thus, for every
      $t_r \in (b_r \dd b_r + (\ell - \alpha_r)]$ and $t_c \in (b_c \dd e_c]$,
      it holds
      $w[t_r, t_c]
        = \Exp{G}{H}[\alpha_r + (t_r - b_r), \alpha_c + (t_c - b_c)]
        = \Exp{G}{X}[\alpha_r + (t_r - b_r), \alpha_c + (t_c - b_c)]
        = \Exp{G}{X}[\ell - ((\ell - \alpha_r) - (t_r - b_r)), \alpha_c + (t_c - b_c)]$.
      In particular,
      $w[\delta_r,\delta_c] =
      \Exp{G}{X}[\ell - \delta_r' + 1, \delta_c']$.
      Thus,
      \begin{align*}
        \AccessTwoDim{G}{N}{\delta_r}{\delta_c}{\DirTop}{\DirLeft}
          &= w[\delta_r,\delta_c]\\
          &= \Exp{G}{X}[\ell - \delta_r' + 1,\delta_c']\\
          &= \AccessTwoDim{G}{X}{\delta_r'}{\delta_c'}{\DirBottom}{\DirLeft}\\
          &= \AccessTwoDim{G}{N'}{\delta_r'}{\delta_c'}{c_r}{c_c}.
      \end{align*}

    \item Otherwise (i.e., if $\delta_r - b_r > \ell - \alpha_r$),
      by \cref{def:top-left-map}, we have
      $(N', \delta_r', \delta_c', c_r, c_c) =
      (Y, (\delta_r - b_r) - (\ell - \alpha_r), \alpha_c + (\delta_c - b_c), \DirTop, \DirLeft)$.
      Combining
      $\delta_r > b_r + (\ell - \alpha_r)$ (the assumption) with
      $\delta_r \leq e_r$ (proved above),
      we obtain
      $b_r + (\ell - \alpha_r) < \delta_r \leq e_r$.
      Furthermore, $b_c < \delta_c \leq e_c$.
      On the other hand, it holds
      $w(b_r \dd e_r](b_c \dd e_c] = \Exp{G}{H}(\alpha_r \dd \beta_r](\alpha_c \dd \beta_c]$
      and $\alpha_r < \ell < \beta_r$ (\cref{lm:2d-hook}).
      Thus, for every
      $t_r \in (b_r + (\ell - \alpha_r) \dd e_r]$ and $t_c \in (b_c \dd e_c]$,
      it holds
      $w[t_r,t_c]
        = \Exp{G}{H}[\alpha_r + (t_r - b_r), \alpha_c + (t_c - b_c)]
        = \Exp{G}{Y}[\alpha_r + (t_r - b_r) - \ell, \alpha_c + (t_c - b_c)]
        = \Exp{G}{Y}[(t_r - b_r) - (\ell - \alpha_r), \alpha_c + (t_c - b_c)]$.
      In particular,
      $w[\delta_r, \delta_c] = \Exp{G}{Y}[\delta_r', \delta_c']$.
      Thus,
      \begin{align*}
        \AccessTwoDim{G}{N}{\delta_r}{\delta_c}{\DirTop}{\DirLeft}
          &= w[\delta_r,\delta_c]\\
          &= \Exp{G}{Y}[\delta_r',\delta_c']\\
          &= \AccessTwoDim{G}{Y}{\delta_r'}{\delta_c'}{\DirTop}{\DirLeft}\\
          &= \AccessTwoDim{G}{N'}{\delta_r'}{\delta_c'}{c_r}{c_c}.
      \end{align*}
    \end{itemize}
  \end{enumerate}

  We now show the remaining claim, i.e., that
  $\LevelId_r = 0$ and $\LevelId_c = 0$ imply $N' \in V_l$.
  By the above, $b_r < e_r$, $b_c < e_c$, $e_r - b_r \leq 2^{\LevelId_r}$, and $e_c - b_c \leq 2^{\LevelId_c}$.
  Thus, if $\LevelId_r = 0$ and $\LevelId_c = 0$, we must have $(b_r,e_r) = (0,1)$ and $(b_c,e_c) = (0,1)$.
  By \cref{def:top-left-map}, we then have $N' = \HookTwoDim{G}{N}{b_r}{b_c}{e_r}{e_c} = H$.
  It remains to observe that by \cref{lm:2d-hook}\eqref{lm:2d-hook-it-1},
  we then have $H \in V_l$.
\end{proof}

\begin{proposition}\label{pr:2d-log-query}
  Given the data structure from
  \cref{sec:2d-log-structure}
  and any
  $(i,j) \in [1 \dd \Rows{\Exp{G}{S}}] \times [1 \dd \Cols{\Exp{G}{S}}]$,
  we can compute $\Exp{G}{S}[i,j]$ in
  $\bigO(\log n)$
  time.
\end{proposition}
\begin{proof}
  The pseudocode of the algorithm is presented in
  \cref{fig:2d-log-access}.
  To compute the symbol
  $\Exp{G}{S}[i,j]$,
  we call
  $\texttt{RandomAccess}(i,j)$.
  At the beginning of the algorithm, we set
  $t$, $\delta_r$, $\delta_c$, $c_r$, and $c_c$
  so that it holds
  $\AccessTwoDim{G}{N_t}{\delta_r}{\delta_c}{c_r}{c_c} = \Exp{G}{S}[i,j]$.
  In each iteration of the while loop we then
  compute variables
  $t'$, $\delta'_r$, $\delta'_c$, $c'_r$, and $c'_c$
  such that
  \[
    \AccessTwoDim{G}{N_{t}}{\delta_r}{\delta_c}{c_r}{c_c} =
    \AccessTwoDim{G}{N_{t'}}{\delta'_r}{\delta'_c}{c'_r}{c'_c}
  \]
  By \cref{lm:top-left-map}, each iteration of the while loop
  decreases one of the two variables $p_r$ and $p_c$
  indicating the current level in the data structure.
  Once the algorithm reaches $p_r = 0$ and $p_c = 0$,
  it must hold 
  $\Rows{\Exp{G}{N_t}} = \Cols{\Exp{G}{N_t}} = 1$.
  Thus, $\Rhs{G}{N_t} \in \Sigma$, and hence
  $\Rhs{G}{N_t} = \Exp{G}{N_t}[1,1] = \Exp{G}{S}[i,j]$.
  The total running time is $\bigO(\log n)$.
\end{proof}

We thus obtain the following result.

\begin{theorem}\label{th:2d-log}
  For every 2D SLP $G$ of size $|G| = g$ representing a 2D string $T \in \Sigma^{r \times c}$, there exists
  a data structure of size $\bigO(g \log^2 n)$, where $n = \max(r,c)$, that,
  given any $(i,j) \in [1 \dd r] \times [1 \dd c]$, returns $T[i,j]$ in $\bigO(\log n)$ time.
\end{theorem}

\section{Random Access using 2D SLPs in Optimal Time}\label{sec:2d-opt}

In this section, we generalize the structure from \cref{sec:2d-log-structure} to
support faster random access queries. The structure from \cref{sec:2d-log-structure}
creates four blocks of size $2^{\LevelId_r} \times 2^{\LevelId_c}$ at level
$(\LevelId_r,\LevelId_c)$ for each of the four corners of each variable expansion,
and stores information regarding each of those blocks. In the generalized version
of the data structure, we instead choose a parameter
$\tau \geq 2$, and for each corner we store $\tau^2$ blocks of size
$\tau^{\LevelId_r} \times \tau^{\LevelId_c}$ at level $(\LevelId_r,\LevelId_c)$.

\subsection{The Data Structure}\label{sec:2d-opt-structure}

\paragraph{Definitions}

The arrays
$H^{NW}$ and $O^{NW}$
are now defined for every
$i \in [1 \dd |V|]$,
$\LevelId_{r}, \LevelId_{c} \in [0 \dd \ceil{\log_{\tau} n}]$, and
$k_r, k_c \in [0 \dd \tau)$
satisfying
$k_r \cdot \tau^{\LevelId_r} < \Rows{\Exp{G}{N_i}}$ and
$k_c \cdot \tau^{\LevelId_c} < \Cols{\Exp{G}{N_i}}$
such that it holds
\begin{align*}
  N_{H^{NW}[i,\LevelId_r,\LevelId_c,k_r,k_c]} &= \HookTwoDim{G}{N_i}{b_r}{b_c}{e_r}{e_c},\\
  O^{NW}[i,\LevelId_r,\LevelId_c,k_r,k_c] &= \OffsetTwoDim{G}{N_i}{b_r}{b_c}{e_r}{e_c},
\end{align*}
where
\begin{itemize}
\item $m_r = \Rows{\Exp{G}{N_i}}$,
\item $m_c = \Cols{\Exp{G}{N_i}}$,
\item $b_r = k_r \cdot \tau^{\LevelId_r}$,
\item $b_c = k_c \cdot \tau^{\LevelId_c}$,
\item $e_r = \min(m_r, (k_r+1) \cdot \tau^{\LevelId_r})$, and
\item $e_c = \min(m_c, (k_c+1) \cdot \tau^{\LevelId_c})$.
\end{itemize}
The arrays $H^{NE}$, $O^{NE}$, $H^{SW}$, $O^{SW}$, $H^{SE}$, and $O^{SE}$
are generalized analogously.

\paragraph{Components}

The data structure consists of the following components:
\begin{enumerate}
\item The array $A_{\rm rhs}$ using $\bigO(|V|)$ space.
\item The arrays $A_{\rm rows}$ and $A_{\rm cols}$ using $\bigO(|V|)$ space.
\item The array $A_{\rm horiz}$ using $\bigO(|V|)$ space.
\item The arrays
  $H^{NW}$, $O^{NW}$,
  $H^{NE}$, $O^{NE}$,
  $H^{SW}$, $O^{SW}$,
  $H^{SE}$, $O^{SE}$
  using
  $\bigO(|V| \cdot \tau^2 \cdot \log^2_{\tau} n)$
  space.
\end{enumerate}
In total, the data structure needs
$\bigO(|V| \cdot \tau^2 \cdot \log^2_{\tau} n) = \bigO(|G| \cdot \tau^2 \cdot \log^2_{\tau} n)$
space.

\subsection{Implementation of Queries}\label{sec:2d-opt-queries}

\begin{definition}[Generalized top-left corner mapping]\label{def:generalized-top-left-map}
  Let $N \in V$.
  Denote
  $m_r = \Rows{\Exp{G}{N}}$ and
  $m_c = \Cols{\Exp{G}{N}}$.
  Consider any
  $\LevelId_r, \LevelId_c \in \Zn$,
  $\delta_r \in [1 \dd m_r]$, and
  $\delta_c \in [1 \dd m_c]$
  satisfying
  $\delta_r \leq \tau^{\LevelId_r+1}$ and
  $\delta_c \leq \tau^{\LevelId_c+1}$.
  Denote
  $k_r = \lceil \tfrac{\delta_r}{\tau^{\LevelId_r}} \rceil - 1 \in [0 \dd \tau)$,
  $b_r = k_r \cdot \tau^{\LevelId_r}$,
  $e_r = \min(m_r, (k_r + 1) \cdot \tau^{\LevelId_r})$,
  $k_c = \lceil \tfrac{\delta_c}{\tau^{\LevelId_c}} \rceil - 1 \in [0 \dd \tau)$,
  $b_c = k_c \cdot \tau^{\LevelId_c}$,
  $e_c = \min(m_c, (k_c + 1) \cdot \tau^{\LevelId_c})$, and
  $H = \HookTwoDim{G}{N}{b_r}{b_c}{e_r}{e_c}$.
  The rest of the definition of the
  generalized top-left corner mapping
  $\TopLeftMap{G}{N}{\LevelId_r}{\LevelId_c}{\delta_r}{\delta_c}$
  is as in \cref{def:top-left-map}.
\end{definition}

\noindent
The generalized corner mappings:
\begin{itemize}
\item $\TopRightMap{G}{N}{\LevelId_r}{\LevelId_c}{\delta_r}{\delta_c}$,
\item $\BottomLeftMap{G}{N}{\LevelId_r}{\LevelId_c}{\delta_r}{\delta_c}$, and
\item $\BottomRightMap{G}{N}{\LevelId_r}{\LevelId_c}{\delta_r}{\delta_c}$
\end{itemize}
are defined analogously.

\begin{lemma}\label{lm:generalized-top-left-map}
  Let $N \in V$,
  $m_r = \Rows{\Exp{G}{N}}$, and
  $m_c = \Cols{\Exp{G}{N}}$.
  Consider any
  $\LevelId_r, \LevelId_c \in \Zn$,
  $\delta_r \in [1 \dd m_r]$ and
  $\delta_c \in [1 \dd m_c]$
  satisfying
  $\delta_r \leq \tau^{\LevelId_r+1}$ and
  $\delta_c \leq \tau^{\LevelId_c+1}$.
  Then, the tuple
  $(N', \delta_r', \delta_c', c_r, c_c) =
  \TopLeftMap{G}{N}{\LevelId_r}{\LevelId_c}{\delta_r}{\delta_c}$
  (\cref{def:generalized-top-left-map})
  satisfies:
  \begin{enumerate}
  \item
    $\delta_r' \in [1 \dd \Rows{\Exp{G}{N'}}]$ and
    $\delta_c' \in [1 \dd \Cols{\Exp{G}{N'}}]$,
  \item
    ($\delta_r' \leq \tau^{\LevelId_r}$ and $\delta_c' \leq \delta_c$) or
    ($\delta_c' \leq \tau^{\LevelId_c}$ and $\delta_r' \leq \delta_r$), and
  \item
    $\AccessTwoDim{G}{N}{\delta_r}{\delta_c}{\DirTop}{\DirLeft} =
    \AccessTwoDim{G}{N'}{\delta_r'}{\delta_c'}{c_r}{c_c}$ (\cref{def:2d-access}).
  \end{enumerate}
  Moreover, if
  $\LevelId_r = 0$ and $\LevelId_c = 0$,
  then
  $N' \in V_l$.
\end{lemma}
\begin{proof}
  The proof is analogous to the proof of \cref{lm:top-left-map}.
\end{proof}

\begin{proposition}\label{pr:2d-opt-query}
  Given the data structure from
  \cref{sec:2d-opt}
  and any
  pair $(i,j) \in [1 \dd \Rows{\Exp{G}{S}}] \times [1 \dd \Cols{\Exp{G}{S}}]$,
  we can compute
  $\Exp{G}{S}[i,j]$ in
  $\bigO(\log_{\tau} n)$
  time.
\end{proposition}
\begin{proof}
  The proof is analogous to the proof of \cref{pr:2d-log-query}, except
  instead of \cref{lm:top-left-map}, we use \cref{lm:generalized-top-left-map}.
\end{proof}

\thtwodopt*
\begin{proof}
  We let $\tau = \log^{\epsilon/2} n$ and use the data structure from
  \cref{sec:2d-opt}.
  It needs
  $\bigO(g \cdot \tau^2 \cdot \log^2_{\tau} n) =
  \bigO(g \cdot \log^{2+\epsilon} n)$ space.
  By \cref{pr:2d-opt-query}, the query time
  is then $\bigO(\log_{\tau} n) = \bigO(\tfrac{\log n}{\log \tau})
  = \bigO(\tfrac{\log n}{\log (\log^{\epsilon/2} n)}) =
  \bigO(\tfrac{\log n}{\log \log n})$.
\end{proof}

\section{Hardness of Computation over 2D SLPs}\label{sec:hardness}

\subsection{Hardness of Compressed Pattern Matching
  for 2D Strings}\label{sec:hardness-pattern-matching}

\subsubsection{Problem Definition}\label{sec:2d-pattern-matching-problem-def}
\vspace{-1.5ex}

\begin{framed}
  \noindent
  \probname{Pattern Matching of a 1D Pattern over a Grammar-Compressed 2D Text}
  \begin{description}[style=sameline,itemsep=0ex,font={\normalfont\bf}]
  \item[Input:] A single-row 2D string $P \in \Sigma^{1 \times k}$ (the pattern)
    represented using a length-$k$ array, and a 2D string
    $T \in \Sigma^{r \times c}$ (the text) represented using
    a 2D SLP (see \cref{sec:prelim}).
  \item[Output:] A \texttt{YES}/\texttt{NO} answer indicating whether
    $P$ occurs in $T$, i.e., if there exists
    $(i,j) \in [1 \dd r] \times [1 \dd c-k+1]$ such that
    $T[i \dd i+1)[j \dd j+k) = P$.
  \end{description}
  \vspace{-1.3ex}
\end{framed}
\vspace{1ex}

\subsubsection{Hardness Assumptions}\label{sec:ov}
\vspace{-1.5ex}

\begin{framed}
  \noindent
  \probname{Orthogonal Vectors (OV)}
  \begin{description}[style=sameline,itemsep=0ex,font={\normalfont\bf}]
  \item[Input:] A set of vectors $A \subseteq \BinaryAlphabet^d$ with $|A| = n$.
  \item[Output:] Determine whether there exist vectors $x,y \in A$ such that $x \cdot y = \sum_{i=1}^{d} x[i] \cdot y[i] = 0$.
  \end{description}
  \vspace{-1.3ex}
\end{framed}
\vspace{1ex}

A naive algorithm for the Orthogonal Vectors problem that checks all
possible pairs of vectors runs in $\bigO(n^2 d)$ time.
For $d = \omega(\log n)$, there is currently no algorithm running
in time $\bigO(n^{2 - \epsilon} \cdot d^{\bigO(1)})$ for any $\epsilon > 0$. This
difficulty in solving the OV problem
has led to the following hypothesis being used as a basis for fine-grained
complexity hardness results.

\begin{conjecture}\label{con:ov}
  The Orthogonal Vectors (OV) Conjecture
  states that given any set $A \subseteq \BinaryAlphabet^{d}$ of $|A| = n$
  binary vectors in $d = \omega(\log n)$ dimensions, it is not
  possible to determine whether there exists a pair $x, y \in A$ satisfying
  $x \cdot y = \sum_{i=1}^{d}x[i]y[i] = 0$ in
  $\bigO(n^{2-\epsilon} \cdot d^{\bigO(1)})$ time,
  for any constant $\epsilon > 0$.
\end{conjecture}

Below we prove a useful lemma showing that every Orthogonal Vectors
instance can be efficiently reduced to the case where the number of ones
in every vector is the same. We will use this transformation in our
reduction in \cref{sec:reducing-ov-to-2d-compressed-pattern-matching}.

\begin{proposition}\label{pr:uniform-ov}
  Let $A \subseteq \BinaryAlphabet^{d}$ and $|A| = n$. Given the set
  $A$, we can, in $\bigO(dn)$ time, construct a set $A' \subseteq
  \BinaryAlphabet^{d'}$, where $d' = 3d$, that satisfies the following
  conditions:
  \begin{enumerate}
  \item $|A'| = 2n$,
  \item For every $x' \in A'$, it holds $\sum_{i=1}^{d'}x'[i] = d$,
  \item There exist $x,y \in A$ such that $x \cdot y = 0$ if and only
    if there exist $x',y' \in A'$ such that $x' \cdot y' = 0$.
  \end{enumerate}
\end{proposition}
\begin{proof}

  For every $x \in A$, denote
  \begin{align*}
    f_0(x) &=
      x \cdot \one^{d-k} \cdot \zero^{d+k} \in \BinaryAlphabet^{3d},\\
    f_1(x)
      &= x \cdot \zero^{d} \cdot \one^{d-k} \cdot \zero^{k} \in \BinaryAlphabet^{3d},
  \end{align*}
  where $k = \sum_{i=1}^{d} x[i]$, i.e., $k$ is the number of ones in $x$.
  The set $A'$ is then defined as
  \[
    A' = \{f_0(x) : x \in A\} \cup \{f_1(x) : x \in A\}.
  \]
  Construction of $A'$ from $A$ is easily done in $\bigO(dn)$ time.

  We now prove the claimed properties of $A'$:
  \begin{enumerate}
  \item For every $x \in A$, we have $f_0(x) \neq
    f_1(x)$. Moreover, since the first $d$ symbols of $f_0(x)$ and
    $f_1(x)$ match $x$, it follows that for all distinct elements
    $x, y \in A$, it holds $f_i(x) \neq f_j(y)$ for any $i,j \in
    \{0,1\}$. Hence, $|A'| = 2n$.
  \item To show the second property, note that, by
    definition, for every $x \in A$, both $f_0(x)$ and $f_1(x)$
    contain precisely $d$ ones.
  \item Let us first consider any $x, y \in A$ such that $x \cdot y =
    0$.  Then $f_0(x) \cdot f_1(y) = 0$, since
    the positions of ones among the first $d$ positions of $f_0(x)$
    and $f_1(y)$ differ by $x \cdot y = 0$, and the positions
    of ones in the remaining positions differ by the definition of
    $f_0(x)$ and $f_1(y)$. Thus, letting $x' = f_0(x)$ and $y' =
    f_1(y)$, we have $x', y' \in A'$ and $x' \cdot y' = 0$.  Conversely,
    assume that for some $x', y' \in A'$ it holds $x' \cdot y' =
    0$.  Then there exist $x, y \in A$ and $i,j \in \{0,1\}$ such
    that $x' = f_i(x)$ and $y' = f_j(y)$. The assumption
    $f_i(x) \cdot f_j(y) = 0$ implies that $f_i(x)[1 \dd
    d] \cdot f_j(y)[1 \dd d] = 0$.  By $x = f_i(x)[1 \dd d]$ and $y
    = f_j(y)[1 \dd d]$, we obtain $x \cdot y = 0$.
    \qedhere
  \end{enumerate}
\end{proof}

\subsubsection{Reducing Orthogonal Vectors
  to 2D Compressed Pattern Matching}\label{sec:reducing-ov-to-2d-compressed-pattern-matching}

\begin{lemma}\label{lm:ov-reduction}
  Let $A \subseteq \BinaryAlphabet^{d}$ and $|A| = n$. Assume that there
  exists $\ell \in [1 \dd d]$ such that, for every $x \in A$, it holds
  $\sum_{i=1}^{d} x[i] = \ell$.
  Given the set $A$, we can, in $\bigO(dn)$ time, construct a 2D string
  $P \in \BinaryAlphabet^{1 \times (\ell+2)}$ and a 2D SLG
  (see \cref{sec:prelim}) representing a 2D string $T \in \BinaryAlphabet^{n \times (\ell+2)n}$
  such that the following conditions are equivalent:
  \begin{enumerate}
  \item\label{lm:ov-reduction-it-1}
    There exist $x,y \in A$ such that $x \cdot y = 0$.
  \item\label{lm:ov-reduction-it-2}
    The pattern $P$ occurs in the text $T$; that is, there exists
    $(r,c) \in [1 \dd \Rows{T}-\Rows{P} + 1] \times [1 \dd \Cols{T} - \Cols{P} + 1]$
    such that $T[r \dd r+\Rows{P})[c \dd c+\Cols{P}) = P$.
  \end{enumerate}
\end{lemma}
\begin{proof}

  We begin by defining the text $T$ (and the associated 2D SLG)
  and the pattern $P$ from the claim.
  Denote $A = \{a_1, \dots, a_n\}$.
  For every $i \in [1 \dd d]$, let $C_i \in \BinaryAlphabet^{n \times 1}$ be such
  that, for every $j \in [1 \dd n]$, $C_i[j,1] = a_j[i]$
  (see \cref{fig:ov}). Let $D \in \BinaryAlphabet^{n \times 1}$ be such that, for
  every $j \in [1 \dd n]$, $D[j,1] = \one$.
  For every $i \in [1 \dd n]$, let $(b_{i,j})_{j \in [1 \dd \ell]}$
  denote an increasing sequence containing
  the positions of ones in $a_i$, i.e., such that
  $b_{i,1} < b_{i,2} < \dots < b_{i,\ell}$ and
  $\{t \in [1 \dd d] : a_i[t] = \one\} = \{b_{i,j}\}_{j \in [1 \dd \ell]}$.
  For every $i \in [1 \dd n]$, let
  \[
    \Gamma_i = C_{b_{i,1}} \vconcat C_{b_{i,2}} \vconcat \dots \vconcat C_{b_{i,\ell}} \in \BinaryAlphabet^{n \times \ell}.
  \]
  We then let $P \in \BinaryAlphabet^{1 \times (\ell+2)}$ be such that
  $P[1,1] = P[1,\ell+2] = \one$ and, for every $j \in [2 \dd \ell+1]$,
  $P[1,j] = \zero$. Finally, we define $T$ as follows:
  \[
    T = \vconcat_{i=1}^{n} \big( D \vconcat \Gamma_i \vconcat D \big)
        \in \BinaryAlphabet^{n \times (\ell+2)n}.
  \]
  Next, we define a 2D SLG representing the 2D string $T$.
  Let $G = (V_{l}, V_{h}, V_{v}, \Sigma, R, S)$, where
  \begin{itemize}
  \item $\Sigma = \BinaryAlphabet$,
  \item $V_{l} = \{N_{\zero}, N_{\one}\}$,
  \item $V_{h} = \{N_{C,1}, \dots, N_{C,d}, N_{D}\}$, and
  \item $V_{v} = \{S\}$.
  \end{itemize}
  Denote $V = V_{l} \cup V_{h} \cup V_{v}$.
  We set $\Rhs{G}{N_{\zero}} = \zero$ and $\Rhs{G}{N_{\one}} = \one$.
  For every $i \in [1 \dd d]$, we set $\Rhs{G}{N_{C,i}} =
  \bigodot_{j=1}^{n} N_{C_i[j,1]} \in V^{n}$.
  Next, we set $\Rhs{G}{N_{D}} = N_{\one}^{n} \in V^{n}$. Finally, we set
  \[
    \Rhs{G}{S} = \textstyle\bigodot_{i=1}^{n}
      \big(
        N_{D} \cdot
        N_{C,b_{i,1}} \cdot N_{C,b_{i,2}} \cdots N_{C,b_{i,\ell}} \cdot
        N_{D}
      \big)
      \in V^{(\ell+2)n}.
  \]
  The size of $G$ is then
  \begin{align*}
    |G|
      &= \textstyle\sum_{N \in V_{l}} |\Rhs{G}{N}| +
         \textstyle\sum_{N \in V_{h}} |\Rhs{G}{N}| +
         \textstyle\sum_{N \in V_{v}} |\Rhs{G}{N}|\\
      &= 2 + (d+1)n + (\ell+2)n = \bigO(dn).
  \end{align*}

  The construction of $P$ takes $\bigO(\ell) = \bigO(d)$ time.
  To construct $G$, we proceed as follows:
  \begin{enumerate}
  \item Compute the sequence $(b_{i,j})_{j \in [1 \dd \ell]}$ for every $i \in [1 \dd n]$.
    Given $A$, this takes $\bigO(dn)$ time.
  \item In $\bigO(dn)$ time, compute the strings $C_i$, where $i \in [1 \dd d]$.
  \item Given $\{C_1, \dots, C_d\}$ and all sequences $(b_{i,j})_{j \in [1 \dd \ell]}$ (where $i \in [1 \dd n]$),
    we can construct $G$ in $\bigO(dn)$ time.
  \end{enumerate}
  In total, the construction of $G$ takes $\bigO(dn)$ time.

  It remains to prove the equivalence of the two conditions in the claim:

  (\ref{lm:ov-reduction-it-1}) $\Rightarrow$ (\ref{lm:ov-reduction-it-2}):
  Assume that there exist $x,y \in A$ such that $x \cdot y = 0$. Let
  $i,j \in [1 \dd n]$ be such that $a_i = x$ and $a_j = y$. The assumption
  $a_i \cdot a_j = 0$ implies, by definition of the sequence
  $(b_{i,t})_{t \in [1 \dd \ell]}$, that, for every $t \in [1 \dd \ell]$,
  it holds $a_j[b_{i,t}] = 0$. By definition of $\Gamma_i$, we thus obtain
  that, for every $t \in [1 \dd \ell]$, it holds
  $\Gamma_{i}[j,t] = C_{b_{i,t}}[j,1] = a_j[b_{i,t}] = 0$.
  This implies that the 2D pattern $P$ occurs in the 2D string
  $D \vconcat \Gamma_i \vconcat D$. Since $D \vconcat \Gamma_i \vconcat D$
  occurs in $T$, we thus obtain an occurrence of $P$ in $T$.

  (\ref{lm:ov-reduction-it-2}) $\Rightarrow$ (\ref{lm:ov-reduction-it-1}):
  Let us now assume that $P$ occurs in $T$, i.e., there exists
  $(r,c) \in [1 \dd \Rows{T} - \Rows{P} + 1] \times [1 \dd \Cols{T} - \Cols{P} + 1]$
  such that $T[r \dd r + \Rows{P})[c \dd c + \Cols{P}) = P$. Observe that,
  by definition of $T$ and $P$, it follows that there exists
  $i \in [1 \dd n]$ such that $c = (\ell+2)(i-1) + 1$.
  This implies that, for every $t \in [1 \dd \ell]$,
  it holds $\Gamma_i[r,t] = P[1+t] = \zero$. We thus obtain
  $\zero = \Gamma_i[r,t] = C_{b_{i,t}}[r,1] = a_r[b_{i,t}]$ for all $t \in [1 \dd \ell]$.
  By definition of the sequence $(b_{i,t})_{t \in [1 \dd \ell]}$, this implies
  that $a_i \cdot a_r = 0$. Thus, letting $x = a_i$ and $y = a_r$, we obtain that
  $x,y \in A$ and $x \cdot y = 0$.
\end{proof}

\begin{figure}[ht!]
	\centering
	\begin{tikzpicture}[scale=0.5]

		\draw node at (2,0){$a_1 = (\one,\zero,\zero,\one)$};
		\draw node at (2,-1){$a_2 = (\one,\one,\zero,\zero)$};
		\draw node at (2,-2){$a_3 = (\zero,\one,\zero,\one)$};
		\draw node at (2,-3){$a_4 = (\zero,\zero,\one,\one)$};
		\draw node at (2,-4){$a_5 = (\one,\zero,\one,\zero)$};

		\draw [] (8,0.5) rectangle (9,-4.5);
		\draw node at (8.5,0){\one};
		\draw node at (8.5,-1){\one};
		\draw node at (8.5,-2){\zero};
		\draw node at (8.5,-3){\zero};
		\draw node at (8.5,-4){\one};
		\draw node at (7,-2){$C_1=$};

		\draw [] (12,0.5) rectangle (13,-4.5);
		\draw node at (12.5,0){\zero};
		\draw node at (12.5,-1){\one};
		\draw node at (12.5,-2){\one};
		\draw node at (12.5,-3){\zero};
		\draw node at (12.5,-4){\zero};
		\draw node at (11,-2){$C_2=$};

		\draw [] (16,0.5) rectangle (17,-4.5);
		\draw node at (16.5,0){\zero};
		\draw node at (16.5,-1){\zero};
		\draw node at (16.5,-2){\zero};
		\draw node at (16.5,-3){\one};
		\draw node at (16.5,-4){\one};
		\draw node at (15,-2){$C_3=$};

		\draw [] (20,0.5) rectangle (21,-4.5);
		\draw node at (20.5,0){\one};
		\draw node at (20.5,-1){\zero};
		\draw node at (20.5,-2){\one};
		\draw node at (20.5,-3){\one};
		\draw node at (20.5,-4){\zero};
		\draw node at (19,-2){$C_4=$};

		\draw node at (0,-10){$T=$};

		\draw node at (1.5, -8){$\one$};
		\draw node at (1.5, -9){$\one$};
		\draw node at (1.5, -10){$\one$};
		\draw node at (1.5, -11){$\one$};
		\draw node at (1.5, -12){$\one$};

		\draw node at (3 , -7){$C_1 C_4$};
		\draw node at (3 , -8){\one\one};
		\draw node at (3 , -9){\one\zero};
		\draw node at (3 , -10){\zero\one};
		\draw node at (3 , -11){\zero\one};
		\draw node at (3 , -12){\one\zero};

		\draw node at (4.5, -8){$\one\one$};
		\draw node at (4.5, -9){$\one\one$};
		\draw node at (4.5, -10){$\one\one$};
		\draw node at (4.5, -11){$\one\one$};
		\draw node at (4.5, -12){$\one\one$};

		\draw node at (6 , -7){$C_1 C_2$};
		\draw node at (6,-8){\one\zero};
		\draw node at (6,-9){\one\one};
		\draw node at (6,-10){\zero\one};
		\draw node at (6,-11){\zero\zero};
		\draw node at (6,-12){\one\zero};

		\draw node at (7.5, -8){$\one\one$};
		\draw node at (7.5, -9){$\one\one$};
		\draw node at (7.5, -10){$\one\one$};
		\draw node at (7.5, -11){$\one\one$};
		\draw node at (7.5, -12){$\one\one$};

		\draw node at (9 , -7){$C_2 C_4$};
		\draw node at (9,-8){\zero\one};
		\draw node at (9,-9){\one\zero};
		\draw node at (9,-10){\one\one};
		\draw node at (9,-11){\zero\one};
		\draw node at (9,-12){\zero\zero};

		\draw node at (10.5, -8){$\one\one$};
		\draw node at (10.5, -9){$\one\one$};
		\draw node at (10.5, -10){$\one\one$};
		\draw node at (10.5, -11){$\one\one$};
		\draw node at (10.5, -12){$\one\one$};

		\draw node at (12 , -7){$C_3 C_4$};
		\draw node at (12,-8){\zero\one};
		\draw node at (12,-9){\zero\zero};
		\draw node at (12,-10){\zero\one};
		\draw node at (12,-11){\one\one};
		\draw node at (12,-12){\one\zero};

		\draw node at (13.5, -8){$\one\one$};
		\draw node at (13.5, -9){$\one\one$};
		\draw node at (13.5, -10){$\one\one$};
		\draw node at (13.5, -11){$\one\one$};
		\draw node at (13.5, -12){$\one\one$};

		\draw node at (15 , -7){$C_1 C_3$};
		\draw node at (15,-8){\one\zero};
		\draw node at (15,-9){\one\zero};
		\draw node at (15,-10){\zero\zero};
		\draw node at (15,-11){\zero\one};
		\draw node at (15,-12){\one\one};

		\draw node at (16.5, -8){$\one$};
		\draw node at (16.5, -9){$\one$};
		\draw node at (16.5, -10){$\one$};
		\draw node at (16.5, -11){$\one$};
		\draw node at (16.5, -12){$\one$};

		\draw[darkgray,fill=green,fill opacity=0.3,thin] (5,-10.5) rectangle (7,-11.5);
		\draw[darkgray,fill=green,fill opacity=0.3,thin] (8,-11.5) rectangle (10,-12.5);
		\draw[darkgray,fill=green,fill opacity=0.3,thin] (11,-8.5) rectangle (13,-9.5);
		\draw[darkgray,fill=green,fill opacity=0.3,thin] (14,-9.5) rectangle (16,-10.5);

		\draw node at (20,-10){$P = \one\zero\zero\one$};
	\end{tikzpicture}
	\caption{Illustration of the construction presented in \cref{lm:ov-reduction}.
    The highlighted blocks of zeros in $T$ correspond to orthogonal pairs
    $(a_2,a_4)$ and $(a_3,a_5)$.}\label{fig:ov}
\end{figure}

\twodpatternmatchinghardness*
\begin{proof}

  Suppose that there exists a constant $\epsilon > 0$ such that, given
  any $P \in \BinaryAlphabet^{1 \times p}$ and a 2D SLP $G$ encoding a
  2D text $T$ over the alphabet $\BinaryAlphabet$, one can determine whether $P$
  occurs in $T$ in $\bigO(|G|^{2-\epsilon} \cdot p^c)$ time, where
  $c > 0$ is some constant. We will show that this implies that \cref{con:ov}
  does not hold.

  Consider any $A \subseteq \BinaryAlphabet^{d}$, where $d > 0$,
  $|A| = n$, and $d = \omega(\log n)$. Given the set $A$, we can determine
  whether there exist $x, y \in A$ such that $x \cdot y = 0$ as follows:
  \begin{enumerate}
  \item Using \cref{pr:uniform-ov}, in $\bigO(dn)$ time we construct
    $A' \subseteq \BinaryAlphabet^{d'}$, where $d' = 3d$ and $|A'| = 2|A|$,
    such that:
    \begin{itemize}
    \item For every $a' \in A'$, it holds $\sum_{i=1}^{d'} a'[i] = d$.
    \item There exist $x,y \in A$ such that $x \cdot y = 0$ if and only if
      there exist $x',y' \in A'$ such that $x' \cdot y' = 0$.
    \end{itemize}
  \item Applying \cref{lm:ov-reduction} to $A'$, in $\bigO(d' \cdot |A'|) =
    \bigO(dn)$ time we construct a 2D string $P \in \BinaryAlphabet^{1 \times (d+2)}$
    and a 2D SLG $G$ representing a 2D string $T \in \BinaryAlphabet^{n \times (d+2)n}$
    such that the following conditions are equivalent:
    \begin{itemize}
    \item There exist $x',y' \in A'$ such that $x' \cdot y' = 0$.
    \item The pattern $P$ occurs in the text $T$.
    \end{itemize}
    Observe that the upper bound on the runtime in \cref{lm:ov-reduction}
    implies that $|G| = \bigO(dn)$.
  \item Using \cref{ob:2d-slg-to-2d-slp}, in $\bigO(|G|) = \bigO(dn)$
    time we construct a 2D SLP $G'$ such that $|G'| = \Theta(|G|)$ and
    $\Lang{G'} = \Lang{G} = \{T\}$.
  \item Applying the above hypothetical algorithm for 2D pattern matching to the 2D
    SLP $G'$ and the pattern $P$, we
    check whether $P$ occurs in $T$ in
    \[
      \bigO(|G'|^{2-\epsilon} \cdot (d+2)^c) =
      \bigO(|G|^{2-\epsilon} \cdot d^c) = 
      \bigO((dn)^{2-\epsilon} \cdot d^c) =
      \bigO(n^{2-\epsilon} \cdot d^{c+2})
    \]
    time. By the above discussion, this check is equivalent to determining
    whether there exist $x,y \in A$ such that $x \cdot y = 0$.
  \end{enumerate}
  In total, the above algorithm takes $\bigO(dn + n^{2-\epsilon} \cdot d^{c+2})
  = \bigO(n^{2-\epsilon} \cdot d^{\bigO(1)})$ time and determines
  whether there exist $x,y \in A$ such that $x \cdot y = 0$. This implies
  that \cref{con:ov} does not hold.
  To see why $\bigO(dn + n^{2-\epsilon} \cdot d^{c+2}) = \bigO(n^{2-\epsilon} \cdot d^{\bigO(1)})$,
  consider two cases:
  \begin{itemize}
  \item If $d = \bigO(n^{1-\epsilon})$, then $\bigO(dn + n^{2-\epsilon} \cdot d^{c+2}) = \bigO(n^{2-\epsilon} \cdot d^{\bigO(1)})$.
  \item If $d = \Omega(n^{1-\epsilon})$, then $n = \bigO(d^{1/(1-\epsilon)})$, and hence it holds
    \[
      \bigO(dn + n^{2-\epsilon} \cdot d^{c+2}) =
      \bigO(d^{1+1/(1-\epsilon)} + n^{2-\epsilon} \cdot d^{c+2})
      = \bigO(n^{2-\epsilon} \cdot d^{\bigO(1)}).
      \qedhere
    \]
  \end{itemize}
\end{proof}

\subsection{Hardness of Compressed Indexing for 2D Strings}\label{sec:hardness-data-structures}

\subsubsection{Problem Definition}\label{sec:2d-queries}

\begin{definition}[2D queries for integer alphabet]\label{def:2d-integer-problems}
  Let $T \in \Zn^{r \times c}$, where $r,c \geq 1$. We define
  the following queries:
  \begin{description}[style=sameline,itemsep=1ex]
  \item[Sum query:]
    Given any $b_r,e_r \in [0 \dd r]$ and $b_c,e_c \in [0 \dd c]$, return
    the value $\SumQuery{T}{b_r}{b_c}{e_r}{e_c}$, defined as
    $\sum_{(i,j) \in (b_r \dd e_r] \times (b_c \dd e_c]} T[i,j]$.
  \item[Line sum query:]
    Given any $\ell \geq 0$ and any $e_r \in [1 \dd r]$ and $e_c \in [0 \dd c]$
    satisfying $e_c \geq \ell$, return the value $\LineSumQuery{T}{e_r}{e_c}{\ell}$,
    defined as $\SumQuery{T}{e_r-1}{e_c-\ell}{e_r}{e_c}$.
  \item[All-zero query:]
    Given any $b_r,e_r \in [0 \dd r]$ and $b_c,e_c \in [0 \dd c]$, return
    $\AllZero{T}{b_r}{b_c}{e_r}{e_c} \in \{0,1\}$, defined such
    that $\AllZero{T}{b_r}{b_c}{e_r}{e_c} = 1$ holds if and only if,
    for every $(i,j) \in (b_r \dd e_r] \times (b_c \dd e_c]$, we have
    $T[i,j] = \zero$.
  \item[Square all-zero query:]
    Given any $\ell \geq 0$ and any $e_r \in [0 \dd r]$, $e_c \in [0 \dd c]$
    satisfying $e_r \geq \ell$ and $e_c \geq \ell$, return the value
    $\SquareAllZero{T}{e_r}{e_c}{\ell}$, defined such that
    $\SquareAllZero{T}{e_r}{e_c}{\ell} = \AllZero{T}{e_r-\ell}{e_c-\ell}{e_r}{e_c}$.
  \end{description}
\end{definition}
\vspace{0.5ex}

\begin{definition}[2D queries for general alphabet]\label{def:2d-general-problems}
  Let $T \in \Sigma^{r \times c}$, where $r,c \geq 1$. We define
  the following queries:
  \begin{description}[style=sameline,itemsep=1ex]
  \item[Equality query:] Given any $b_r,b_r',h \in [1 \dd r]$ and
  $b_c,b_c',w \in [1 \dd c]$ satisfying $\max(b_r,b_r')+h \leq r+1$
  and $\max(b_c,b_c')+w \leq c+1$, return the value
    $\EqualQuery{T}{b_r}{b_c}{b_r'}{b_c'}{h}{w} \in \{0,1\}$ defined
    such that $\EqualQuery{T}{b_r}{b_c}{b_r'}{b_c'}{h}{w} = 1$
    if and only if
    $T[b_r \dd b_r + h)[b_c \dd b_c + w) =
     T[b_r' \dd b_r' + h)[b_c' \dd b_c' + w)$.
  \item[Square LCE query:] Given any $b_r,b_r' \in [1 \dd r]$
    and $b_c,b_c' \in [1 \dd c]$, return the value
    $\SquareLCE{T}{b_r}{b_c}{b_r'}{b_c'}$ defined as the
    largest $t \in \Zn$ such that
    $T[b_r \dd b_r + t)[b_c \dd b_c + t) =
    T[b_r' \dd b_r' + t)[b_c' \dd b_c' + t)$.
  \item[Line LCE query:] Given any $b_r,b_r' \in [1 \dd r]$,
    $b_c, b_c' \in [1 \dd c]$, and $\ell$ satisfying
    $b_r + \ell \leq r+1$ and $b_r' + \ell \leq r+1$, return
    the value $\LineLCE{T}{b_r}{b_c}{b_r'}{b_c'}{\ell}$,
    defined as the largest $t \in \Zn$ such that
    $T[b_r \dd b_r + \ell)[b_c \dd b_c + t) =
    T[b_r' \dd b_r' + \ell)[b_c' \dd b_c' + t)$.
  \end{description}
\end{definition}

\subsubsection{Hardness Assumptions}\label{sec:hardness-data-structures-hardness-assumptions}

The results in \cref{sec:hardness-data-structures} are based on the
hardness of the following two types of queries.

\begin{definition}[Rank queries]\label{def:rank}
  Let $T \in \Sigma^{n}$.
  For every $j \in [0 \dd n]$ and every $c \in \Sigma$, we define
  \vspace{-0.5ex}
  \[
    \Rank{T}{j}{c} = |\Occ{c}{T(0 \dd j]}| = |\{i \in (0 \dd j] : T[i] = c\}|.
  \]
\end{definition}

\begin{definition}[Symbol occurrence queries]\label{def:symbol-occ}
  Let $T \in \Sigma^{n}$.
  For every $b, e \in [0 \dd n]$ and every $c \in \Sigma$, we define
  \vspace{1ex}
  \[
    \Occurs{T}{b}{e}{c} = 
      \begin{cases}
        1 & \text{if there exists } j \in (b \dd e] \text{ such that }T[j] = c,\\
        0 & \text{otherwise}.
      \end{cases}
  \]
  \vspace{-0.5ex}

  Equivalently, $\Occurs{T}{b}{e}{c} \in \{0,1\}$ is defined such that
  $\Occurs{T}{b}{e}{c} = 1$ holds if and only if $\Occ{c}{T(b \dd e]}
  \neq \emptyset$.
\end{definition}

Currently, there is no
known data structure that, given a 1D SLP $G$ representing a string $T
\in \Sigma^{n}$, where $\Sigma = [0 \dd \sigma)$ and
$\sigma = n^{\bigO(1)}$, uses $\bigO(|G| \cdot \log^{\bigO(1)} n)$
space and answers either rank (\cref{def:rank}) or symbol occurrence
(\cref{def:symbol-occ}) queries in $\bigO(\log^{\bigO(1)} n)$ time.

\subsubsection{Basic Reductions among 2D Queries}\label{sec:hardness-data-structures-basic-reductions}

\begin{proposition}\label{pr:reduce-square-lce-to-line-lce}
  Let $G = (V_l, V_h, V_v, \Sigma, R, S)$ be a 2D SLP representing a
  string $T \in \Sigma^{r \times c}$, and let $n = \max(r, c)$. If
  there exists a data structure using $\bigO(|G| \cdot \log^{\bigO(1)}
  n)$ space that answers line LCE queries
  (\cref{def:2d-general-problems}) on $T$ in $\bigO(\log^{\bigO(1)} n)$ time,
  then there exists a data structure of the same size that answers
  square LCE queries (\cref{def:2d-general-problems}) on $T$ also in
  $\bigO(\log^{\bigO(1)} n)$ time.
\end{proposition}
\begin{proof}
  It suffices to observe that, for any $b_r,b_r' \in [1 \dd m]$,
  $b_c,b_c' \in [1 \dd n]$, and any integer $t \in \Zn$,
  $\SquareLCE{T}{b_r}{b_c}{b_r'}{b_c'} \geq t$ holds if and only if
  $\LineLCE{T}{b_r}{b_c}{b_r'}{b_c'}{t} \geq t$. Thus, we can use
  binary search to reduce a square LCE query to line LCE queries. The
  query slows down only by a factor of $\bigO(\log n)$.
\end{proof}

\begin{proposition}\label{pr:reduce-line-lce-to-equality}
  Let $G = (V_l, V_h, V_v, \Sigma, R, S)$ be a 2D SLP representing a
  string $T \in \Sigma^{r \times c}$, and let $n = \max(r, c)$. If
  there exists a data structure using $\bigO(|G| \cdot \log^{\bigO(1)}
  n)$ space that answers subrectangle equality queries
  (\cref{def:2d-general-problems}) on $T$ in $\bigO(\log^{\bigO(1)} n)$ time,
  then there exists a data structure of the same size that answers
  line LCE queries (\cref{def:2d-general-problems}) on $T$ also in
  $\bigO(\log^{\bigO(1)} n)$ time.
\end{proposition}
\begin{proof}
  It suffices to observe that, for any
  $b_r,b_r' \in [1 \dd m]$, $b_c, b_c' \in [1 \dd n]$, $\ell$ satisfying
  $b_r + \ell \leq m+1$ and $b_r' + \ell \leq m+1$, and any integer $t \in \Zn$,
  $\LineLCE{T}{b_r}{b_c}{b_r'}{b_c'}{\ell} \geq t$ holds if and only if
  $\EqualQuery{T}{b_r}{b_c}{b_r'}{b_c'}{\ell}{t} = 1$. Thus, we can use
  binary search to reduce a line LCE query to subrectangle equality
  queries. The query slows down only by a factor of $\bigO(\log n)$.
\end{proof}

\begin{proposition}\label{pr:reduce-square-all-zero-to-square-lce}
  Let $G = (V_l, V_h, V_v, \Sigma, R, S)$, where $\Sigma = \BinaryAlphabet$,
  be a 2D SLP representing a string $T \in \Sigma^{r \times c}$, and let $n = \max(r, c)$. If
  there exists a data structure using $\bigO(|G| \cdot \log^{\bigO(1)} n)$
  space that answers square LCE queries
  (\cref{def:2d-general-problems}) on $T$ in $\bigO(\log^{\bigO(1)} n)$ time,
  then there exists a data structure of the same size that answers
  square all-zero queries (\cref{def:2d-integer-problems}) on $T$ also in
  $\bigO(\log^{\bigO(1)} n)$ time.
\end{proposition}
\begin{proof}

  Let $T' = T \vconcat \zero^{r \times c} \in \Zn^{r \times 2c}$.
  Let $G'$ be a 2D SLP such that $\Lang{G'} = \{T'\}$
  and $|G'| = |G| + \Theta(\log r + \log c)$
  (such a 2D SLP is straightforward to construct from $G$).
  By \cref{lm:2d-slp-min-size}, it holds that $|G| = \Omega(\log r + \log c)$,
  and hence $|G'| = \Theta(|G|)$.

  The data structure for answering square all-zero queries on $T$ consists of a single component:
  the data structure answering square LCE queries for $G'$.
  It needs $\bigO(|G'| \log^{\bigO(1)} \max(r, 2c)) =
  \bigO(|G| \log^{\bigO(1)} n)$ space and answers square LCE queries on $T'$ in $\bigO(\log^{\bigO(1)} \max(r, 2c)) =
  \bigO(\log^{\bigO(1)} n)$ time.

  Let $\ell \geq 0$ and $e_r \in [0 \dd r]$, $e_c \in [0 \dd c]$ be such that $e_r \geq \ell$ and $e_c \geq \ell$.
  Using the above structure, we compute $\SquareAllZero{T}{e_r}{e_c}{\ell}$ by computing
  $x = \SquareLCE{T'}{e_r - \ell}{e_c - \ell}{1}{c+1}$ and returning $\SquareAllZero{T}{e_r}{e_c}{\ell} = 1$ if and only
  if $x \geq \ell$. This takes $\bigO(\log^{\bigO(1)} n)$ time.
\end{proof}

\subsubsection{Alphabet Reduction for Rank and Symbol Occurrence Queries}\label{sec:rank-and-symbol-occ-alphabet-reduction}

\begin{proposition}\label{pr:slp-rank-and-symbol-occ-alphabet-reduction}
  Assume that, for every 1D SLP $G = (V, \Sigma, R, S)$ representing a string $T \in \Sigma^{n}$,
  where $\Sigma = [0 \dd \sigma)$ and $\sigma = \bigO(\min(|G|,|T|))$,
  there exists a data structure of size $\bigO(|G| \cdot \log^{\bigO(1)} n)$
  that answers rank queries (resp.\ symbol occurrence queries) (see \cref{def:rank,def:symbol-occ})
  on $T$ in $\bigO(\log^{\bigO(1)} n)$ time.
  Then, for every 1D SLP $\hat{G} = (\hat{V}, \hat{\Sigma}, \hat{R}, \hat{S})$
  representing $\hat{T} \in \hat{\Sigma}^{\hat{n}}$,
  where $\hat{\Sigma} = [0 \dd \hat{\sigma})$ and $\hat{\sigma} = \hat{n}^{\bigO(1)}$,
  there exists a data structure of size $\bigO(|\hat{G}| \cdot \log^{\bigO(1)} \hat{n})$
  that answers rank (resp.\ symbol occurrence) queries on $\hat{T}$ in $\bigO(\log^{\bigO(1)} \hat{n})$ time.
\end{proposition}
\begin{proof}

  Let $c > 0$ and $c' > 0$ be constants such that, for every 1D SLP
  $G = (V, \Sigma, R, S)$ representing a string $T \in \Sigma^{n}$,
  where $\Sigma = [0 \dd \sigma)$ and $\sigma = \bigO(\min(|G|,|T|))$,
  there exists a data structure of size $\bigO(|G| \cdot \log^{c} n)$
  that answers rank (resp.\ symbol occurrence) queries on $T$ in $\bigO(\log^{c'} n)$ time.
  Let $\Sigma_{\hat{T}} = \{\hat{T}[i] : i \in [1 \dd |\hat{T}|]\}$ be the set of characters occurring in $\hat{T}$.
  Let $f : \Sigma_{\hat{T}} \rightarrow [0 \dd |\Sigma_{\hat{T}}|)$ be a function defined so that, for every
  $c \in \Sigma_{\hat{T}}$, it holds $f(c) = |\{c' \in \Sigma_{\hat{T}} : c' \prec c\}|$.
  Let $G'$ be the SLP obtained from $\hat{G}$ by setting $\Rhs{G'}{N} = f(c)$ for every
  $N \in \hat{V}$ such that $\Rhs{\hat{G}}{N} = c \in [0 \dd \sigma)$,
  leaving right-hand sides of all other nonterminals unmodified, and removing from $\hat{G}$ all nonterminals
  that are not used when expanding $\hat{S}$ into $\hat{T}$. Note that we then have
  $|\Sigma_{\hat{T}}| \leq |G'| \leq |\hat{G}|$ and $|G'| = \bigO(\hat{n})$.
  In particular, $|\Sigma_{\hat{T}}| = \bigO(\min(|G'|, \hat{n})) = \bigO(\min(|G'|, |T'|))$, where
  $T' = \bigodot_{i=1,\dots,\hat{n}} f(T[i])$ is the string represented by $G'$.
  Finally, let $A[1 \dd |\Sigma_{\hat{T}}|]$ be an array containing all elements
  of $\Sigma_{\hat{T}}$ in increasing order.

  The data structure answering rank (resp.\ symbol occurrence) queries on $\hat{T}$ consists of the following components:
  \begin{enumerate}
  \item The data structure answering rank (resp.\ symbol occurrence) queries from the assumption in the claim applied to $G'$.
    It requires $\bigO(|G'| \cdot \log^{c} \hat{n})$ space.
    We can apply this data structure to $G'$ because the
    represented string $T'$ is over the alphabet $[0 \dd |\Sigma_{\hat{T}}|)$, and it holds
    $|\Sigma_{\hat{T}}| = \bigO(\min(|G'|, |T'|))$.
  \item The array $A$ stored in plain form, which requires $\bigO(|\Sigma_{\hat{T}}|) = \bigO(|G'|)$ space.
  \end{enumerate}
  In total, the structure uses $\bigO(|G'|  + |G'| \cdot \log^{c} \hat{n})
  = \bigO(|\hat{G}| \cdot \log^{c} \hat{n}) = \bigO(|\hat{G}| \cdot \log^{\bigO(1)} \hat{n})$ space.

  Let $j \in [0 \dd \hat{n}]$ (resp.\ $b,e \in [0 \dd \hat{n}]$) and $c \in [0 \dd \hat{\sigma})$.
  Using the above structure, we compute $\Rank{\hat{T}}{j}{c}$ (resp.\ $\Occurs{\hat{T}}{b}{e}{c}$)
  as follows:
  \begin{enumerate}
  \item Using binary search on array $A$, in $\bigO(\log |\Sigma_{\hat{T}}|) = \bigO(\log \hat{n})$ time we check
    whether $c \in \Sigma_{\hat{T}}$. If $c \not\in \Sigma_{\hat{T}}$, then
    we return that $\Rank{\hat{T}}{j}{c} = 0$ (resp.\ $\Occurs{\hat{T}}{b}{e}{c} = 0$).
    Let us thus assume
    that $c \in \Sigma_{\hat{T}}$, and let $k \in [1 \dd |\Sigma_{\hat{T}}|]$ be such that $A[k] = c$ (the index $k$
    is easy to compute during the binary search). Note that we then have $f(c) = k-1$.
  \item Using the data structure from the claim applied to $G'$, in
    $\bigO(\log^{c'} \hat{n})$ time,
    we compute $r = \Rank{T'}{j}{k-1}$ (resp.\ $r = \Occurs{T'}{b}{e}{k-1}$) and return $r$ as the answer.
    Correctness follows from $f(c) = k-1$ and the fact that
    $\Rank{\hat{T}}{j}{c} = \Rank{T'}{j}{f(c)}$ (resp.\ $\Occurs{\hat{T}}{b}{e}{c} = \Occurs{T'}{b}{e}{f(c)}$)
    holds for $c \in \Sigma_{\hat{T}}$.
  \end{enumerate}
  In total, the query takes $\bigO(\log \hat{n} + \log^{c'} \hat{n}) = \bigO(\log^{\bigO(1)} \hat{n})$ time.
\end{proof}

\subsubsection{Reductions from Rank Queries over 1D SLPs}\label{sec:reductions-from-rank}

\begin{definition}[Symbol marking vector]\label{def:mark-char}
  Let $\sigma \in \Zp$ and $T \in [0 \dd \sigma)^{n}$. For any
  $c \in [0 \dd \sigma)$, we define $\MarkChar{T}{c} \in
  \BinaryAlphabet^{1 \times n}$ such that, for every $j \in [1 \dd n]$,
  $\MarkChar{T}{c}[1,j] = 1$ if and only if $T[j] = c$.
\end{definition}

\begin{definition}[Alphabet marking matrix]\label{def:mark-all-chars}
  Let $\sigma \in \Zp$ and $T \in [0 \dd \sigma)^{n}$. We then define
  (see \cref{def:mark-char})
  \[
    \MarkAllChars{T}{\sigma} =
      \textstyle\hconcat_{c=0,\dots,\sigma-1} \MarkChar{T}{c}
      \in \BinaryAlphabet^{\sigma \times n}.
  \]
\end{definition}

\begin{observation}\label{ob:mark-all-chars-concat}
  Let $\sigma \in \Zp$, $T_1 \in [0 \dd \sigma)^{n_1}$,
  $T_2 \in [0 \dd \sigma)^{n_2}$, and $T_3 = T_1 \cdot T_2$.
  Then, it holds $\MarkAllChars{T_3}{\sigma} =
  \MarkAllChars{T_1}{\sigma} \vconcat \MarkAllChars{T_2}{\sigma}$
  (\cref{def:mark-all-chars}).
\end{observation}

\begin{lemma}\label{lm:reduce-rank-to-line-sum}
  Let $\sigma \in \Zp$ and $T \in [0 \dd \sigma)^{n}$. Denote
  $T' = \MarkAllChars{T}{\sigma}$ (\cref{def:mark-all-chars}).
  For every $j \in [0 \dd n]$ and every $c \in [0 \dd \sigma)$, it holds
  (see \cref{def:rank,def:2d-integer-problems}):
  \[
    \Rank{T}{j}{c} = \LineSumQuery{T'}{c+1}{j}{j}.
  \]
\end{lemma}
\begin{proof}
  By \cref{def:2d-integer-problems,def:mark-char,def:mark-all-chars}, it holds
  \[
    \LineSumQuery{T'}{c+1}{j}{j} =
    \SumQuery{T'}{c}{0}{c+1}{j} =
    \textstyle\sum_{t \in (0 \dd j]} \MarkChar{T}{c}[1,t].
  \]
  Consequently, by \cref{def:rank}, we obtain:
  \begin{align*}
    \Rank{T}{j}{c}
      &= |\{t \in (0 \dd j] : T[t] = c\}|\\
      &= |\{t \in (0 \dd j] : \MarkChar{T}{c}[1,t] = 1\}|\\
      &= \textstyle\sum_{t \in (0 \dd j]} \MarkChar{T}{c}[1,t]\\
      &= \LineSumQuery{T'}{c+1}{j}{j}.
      \qedhere
  \end{align*}
\end{proof}

\begin{lemma}\label{lm:mark-all-chars-grammar}
  Let $\sigma \in \Zp$, $T \in [0 \dd \sigma)^{n}$, and let $G$
  be a 1D SLP representing $T$, i.e., such that $\Lang{G} = \{T\}$.
  There exists a 2D SLG $G'$ of size $|G'| = \bigO(|G| + \sigma)$ such that
  $\Lang{G'} = \{\MarkAllChars{T}{\sigma}\}$ (\cref{def:mark-all-chars}).
\end{lemma}
\begin{proof}

  Denote $G = (V, \Sigma, R, S)$, $V = \{N_1, \dots, N_{g}\}$, and let $s \in [1 \dd g]$
  be such that $S = N_{s}$. We then let
  $G' = (V_{l}, V_{h}, V_{v}, \Sigma', R', S')$, where
  \begin{itemize}
  \item $\Sigma' = \{\zero, \one\}$,
  \item $V_{l} = \{M_{\zero}, M_{\one}\}$,
  \item $V_{h} = \{Z_0, \dots, Z_{\sigma-1}, X_0, \dots, X_{\sigma-1}\}$,
  \item $V_{v} = \{Y_1, \dots, Y_g\}$,
  \item $S' = Y_{s}$.
  \end{itemize}
  We set $\Rhs{G'}{M_{\zero}} = \zero$ and $\Rhs{G'}{M_{\one}} = \one$.
  We also set $\Rhs{G'}{Z_{0}} = \emptystring$, and, for every $i \in [1 \dd \sigma)$,
  we let $\Rhs{G'}{Z_{i}} = Z_{i-1} \cdot M_{\zero}$. For every $i \in [0 \dd \sigma)$, we
  also define $\Rhs{G'}{X_{i}} = Z_{i} \cdot M_{\one} \cdot Z_{\sigma-i-1}$.
  Lastly, we define $\Rhs{G'}{Y_i}$ for all $i \in [1 \dd g]$. We consider two cases:
  \begin{itemize}
  \item If $|\Rhs{G}{N_i}| = 1$, then, letting $c \in [0 \dd \sigma)$ be such that
    $\Rhs{G}{N_i} = c$, we set $\Rhs{G'}{Y_i} = X_{c}$.
  \item If $|\Rhs{G}{N_i}| = 2$, then, letting $j_1,j_2 \in [1 \dd g]$ be such that
    $\Rhs{G}{N_i} = N_{j_1} \cdot N_{j_2}$, we set $\Rhs{G'}{Y_i} = Y_{j_1} \cdot Y_{j_2}$.
  \end{itemize}
  The size of the above grammar is $|G'| = \bigO(g + \sigma) = \bigO(|G| + \sigma)$.

  To show that $\Lang{G'} = \{\MarkAllChars{T}{\sigma}\}$, first observe that,
  for every $a \in [0 \dd \sigma)$, it holds
  $\Exp{G'}{X_a} =
  \zero^{a \times 1} \hconcat
  \one^{1 \times 1} \hconcat
  \zero^{(\sigma-a-1) \times 1} \in
  \BinaryAlphabet^{\sigma \times 1}$.
  By induction on the length of the expansion, we thus
  obtain that, for every $i \in [1 \dd g]$, it holds
  \[
    \Exp{G'}{Y_i} = \MarkAllChars{\Exp{G}{N_i}}{\sigma} \in \BinaryAlphabet^{\sigma \times |\Exp{G}{N_i}|}.
  \]
  In particular,
  \[
    \Lang{G'}
      = \{\Exp{G'}{S'}\}
      = \{\Exp{G'}{Y_s}\}
      = \{\MarkAllChars{\Exp{G}{N_s}}{\sigma}\}
      = \{\MarkAllChars{T}{\sigma}\}.
      \qedhere
  \]
\end{proof}

\begin{proposition}\label{pr:reduce-rank-to-line-sum-small-sigma}
  Assume that, for every 2D SLP $G = (V_{l}, V_{h}, V_{v}, \Sigma, R, S)$, where $\Sigma = \BinaryAlphabet$,
  representing a 2D string $T \in \BinaryAlphabet^{r \times c}$,
  there exists a data structure of size $\bigO(|G| \cdot \log^{\bigO(1)} n)$
  that answers line sum queries (\cref{def:2d-integer-problems}) on $T$ in $\bigO(\log^{\bigO(1)} n)$ time, where
  $n = \max(r, c)$.
  Then, for every 1D SLP $\hat{G} = (\hat{V}, \hat{\Sigma}, \hat{R}, \hat{S})$ representing $\hat{T} \in \hat{\Sigma}^{\hat{n}}$,
  where $\hat{\Sigma} = [0 \dd \hat{\sigma})$ and $\hat{\sigma} = \bigO(\min(|\hat{G}|, |\hat{T}|))$,
  there exists a data structure of size $\bigO(|\hat{G}| \cdot \log^{\bigO(1)} \hat{n})$ that answers rank
  queries (\cref{def:rank}) in $\bigO(\log^{\bigO(1)} \hat{n})$ time.
\end{proposition}
\begin{proof}

  Let $c > 0$ and $c' > 0$ be constants such that, for every 2D SLP
  $G = (V_{l}, V_{h}, V_{v}, \Sigma, R, S)$ (where $\Sigma = \BinaryAlphabet$)
  representing a 2D string $T \in \BinaryAlphabet^{r \times c}$, the structure from the claim uses
  $\bigO(|G| \cdot \log^{c} n)$ space (where $n = \max(r, c)$) and
  answers line sum queries on $T$ in $\bigO(\log^{c'} n)$ time.
  Let $G'$ be a 2D SLP of size $|G'| = \bigO(|\hat{G}| + \hat{\sigma}) = \bigO(|\hat{G}|)$ such that
  $\Lang{G'} = \{\MarkAllChars{\hat{T}}{\hat{\sigma}}\}$ (\cref{def:mark-all-chars}). Such a 2D SLP exists
  by \cref{lm:mark-all-chars-grammar} and \cref{ob:2d-slg-to-2d-slp}.
  Denote $T' = \MarkAllChars{\hat{T}}{\hat{\sigma}}$ and $n' = \max(\Rows{T'}, \Cols{T'})$.
  Note that, since $\Rows{T'} = \hat{\sigma} = \bigO(\hat{n})$ and $\Cols{T'} = \hat{n}$, it holds $n' = \bigO(\hat{n})$.

  The data structure answering rank queries on $\hat{T}$ consists of a single component:
  the data structure answering line sum queries from the claim applied to $G'$. Note that
  we can apply this structure to $G'$ because $T'$ is over a binary alphabet. The structure
  needs $\bigO(|G'| \cdot \log^{c} n') = \bigO(|\hat{G}| \cdot \log^{c} \hat{n}) =
  \bigO(|\hat{G}| \cdot \log^{\bigO(1)} \hat{n})$ space.

  Let $j \in [0 \dd \hat{n}]$ and $c \in [0 \dd \hat{\sigma})$. Using the above structure, we compute
  $\Rank{\hat{T}}{j}{c}$ as follows:
  \begin{enumerate}
  \item First, in $\bigO(\log^{c'} n') = \bigO(\log^{c'} \hat{n})$ time, we compute
    $s = \LineSumQuery{T'}{c+1}{j}{j}$.
  \item By \cref{lm:reduce-rank-to-line-sum},
    it holds $s = \Rank{\hat{T}}{j}{c}$.
    We thus return $s$ as the answer.
  \end{enumerate}
  The query takes $\bigO(\log^{c'} \hat{n}) = \bigO(\log^{\bigO(1)} \hat{n})$ time.
\end{proof}

\begin{proposition}\label{pr:reduce-rank-to-line-sum}
  Assume that, for every 2D SLP $G = (V_{l}, V_{h}, V_{v}, \Sigma, R, S)$, where $\Sigma = \BinaryAlphabet$,
  representing a 2D string $T \in \BinaryAlphabet^{r \times c}$,
  there exists a data structure of size $\bigO(|G| \cdot \log^{\bigO(1)} n)$
  that answers line sum queries (\cref{def:2d-integer-problems}) on $T$ in $\bigO(\log^{\bigO(1)} n)$ time, where
  $n = \max(r, c)$.
  Then, for every 1D SLP $\hat{G} = (\hat{V}, \hat{\Sigma}, \hat{R}, \hat{S})$
  representing $\hat{T} \in \hat{\Sigma}^{\hat{n}}$,
  where $\hat{\Sigma} = [0 \dd \hat{\sigma})$ and $\hat{\sigma} = \hat{n}^{\bigO(1)}$,
  there exists a data structure of size $\bigO(|\hat{G}| \cdot \log^{\bigO(1)} \hat{n})$
  that answers rank queries on $\hat{T}$ in $\bigO(\log^{\bigO(1)} \hat{n})$ time.
\end{proposition}
\begin{proof}
  The result follows by combining \cref{pr:reduce-rank-to-line-sum-small-sigma} and
  \cref{pr:slp-rank-and-symbol-occ-alphabet-reduction}.
\end{proof}

\reducefromrank*
\begin{proof}
  The reduction from rank queries on a 1D SLP to line sum queries on a 2D SLP follows
  from \cref{pr:reduce-rank-to-line-sum}. The hardness for sum queries follows immediately,
  since line sum queries are a special case of sum queries.
\end{proof}

\subsubsection{Reductions from Symbol Occurrence
  Queries over 1D SLPs}\label{sec:reductions-from-symbol-occ}

\begin{definition}[Extended alphabet marking matrix]\label{def:ext-mark-all-chars}
  Let $\sigma \in \Zp$ and $T \in [0 \dd \sigma)^{n}$.
  Denote $Z = \zero^{(n-1) \times n}$.
  We then define (see \cref{def:mark-char})
  \[
    \ExtMarkAllChars{T}{\sigma} =
      \textstyle\hconcat_{a=0,\dots,\sigma-1}
        \big( \MarkChar{T}{a} \hconcat Z \big)
      \in \BinaryAlphabet^{n\sigma \times n}.
  \]
\end{definition}

\begin{lemma}\label{lm:reduce-symbol-occ-to-square-all-zero}
  Let $\sigma \in \Zp$ and $T \in [0 \dd \sigma)^{n}$. Denote
  $T' = \ExtMarkAllChars{T}{\sigma}$ (\cref{def:ext-mark-all-chars}). For
  every $a \in [0 \dd \sigma)$ and every $b, e \in [0 \dd n]$ satisfying
  $b < e$, the following conditions are equivalent:
  \begin{enumerate}
  \item\label{lm:reduce-symbol-occ-to-square-all-zero-it-1}
    $\Occurs{T}{b}{e}{a} = 0$ (\cref{def:symbol-occ}),
  \item\label{lm:reduce-symbol-occ-to-square-all-zero-it-2}
    $\SquareAllZero{T'}{n \cdot (a+1)}{e}{e-b} = 1$ (\cref{def:2d-integer-problems}).
  \end{enumerate}
\end{lemma}
\begin{proof}
  Observe that, by \cref{def:ext-mark-all-chars}, it holds
  \[
    T'(n \cdot (a+1)-(e-b) \dd n \cdot (a+1)](b \dd e]
    = \MarkChar{T}{a}(b \dd e] \hconcat \zero^{(e-b-1) \times (e-b)}.
  \]
  This implies that each adjacent pair of statements in the following
  sequence is equivalent:
  \begin{itemize}
  \item $\SquareAllZero{T'}{n \cdot (a+1)}{e}{e-b} = 1$,
  \item $T'(n \cdot (a+1)-(e-b) \dd n \cdot (a+1)](b \dd e] = \zero^{(e-b) \times (e-b)}$,
  \item $\MarkChar{T}{a}(b \dd e] = \zero^{1 \times (e-b)}$,
  \item for every $j \in (b \dd e]$, it holds $T[j] \neq a$,
  \item $\Occurs{T}{b}{e}{a} = 0$.
  \end{itemize}
  In particular, the first condition is equivalent to the last one.
\end{proof}

\begin{lemma}\label{lm:ext-mark-all-chars-grammar}
  Let $\sigma \in \Zp$, $T \in [0 \dd \sigma)^{n}$ (where $n \geq 2$),
  and let $G$ be a 1D SLP representing $T$, i.e.,
  such that $\Lang{G} = \{T\}$. There exists
  a 2D SLG $G'$ of size $|G'| = \bigO(|G| + \sigma)$ such that
  $\Lang{G'} = \{\ExtMarkAllChars{T}{\sigma}\}$ (\cref{def:ext-mark-all-chars}).
\end{lemma}
\begin{proof}

  Denote $G = (V, \Sigma, R, S)$, $V = \{N_1, \dots, N_{g}\}$, and let $s \in [1 \dd g]$
  be such that $S = N_{s}$. Denote $k = \lceil \log (n-1) \rceil$ and let
  $(p_1, \dots, p_{k'}) \in \Zn^{+}$ be an increasing sequence satisfying $n-1 = \sum_{i=1}^{k'} 2^{p_i}$.
  We then let $G' = (V_{l}, V_{h}, V_{v}, \Sigma', R', S')$, where
  \begin{itemize}
  \item $\Sigma' = \{\zero, \one\}$,
  \item $V_{l} = \{M_{\zero}, M_{\one}\}$,
  \item $V_{h} = \{Z_0, \dots, Z_{k-1}, Z, Z',
                   Z'_0, \dots, Z'_{\sigma-1},
                   X_0, \dots, X_{\sigma-1}\}$,
  \item $V_{v} = \{Y_1, \dots, Y_g\}$,
  \item $S' = Y_{s}$.
  \end{itemize}
  We set $\Rhs{G'}{M_{\zero}} = \zero$ and $\Rhs{G'}{M_{\one}} = \one$.
  We also set $\Rhs{G'}{Z_0} = M_{\zero}$, and for every $i \in [1 \dd k)$,
  $\Rhs{G'}{Z_i} = Z_{i-1} \cdot Z_{i-1}$.
  Next, we set $\Rhs{G'}{Z} = \bigodot_{i=1}^{k'} Z_{p_i}$ and $\Rhs{G'}{Z'} = Z \cdot M_{\zero}$.
  We also set $\Rhs{G'}{Z'_0} = \emptystring$,
  and for every $i \in [1 \dd \sigma)$, $\Rhs{G'}{Z'_{i}} = Z'_{i-1} \cdot Z'$.
  For every $i \in [0 \dd \sigma)$, we
  let $\Rhs{G'}{X_{i}} = Z'_{i} \cdot M_{\one} \cdot Z \cdot Z'_{\sigma-i-1}$.
  Lastly, we define $\Rhs{G'}{Y_i}$ for all $i \in [1 \dd g]$. We consider two cases:
  \begin{itemize}
  \item If $|\Rhs{G}{N_i}| = 1$, then letting $c \in [0 \dd \sigma)$ be such that
    $\Rhs{G}{N_i} = c$, we set $\Rhs{G'}{Y_i} = X_{c}$.
  \item If $|\Rhs{G}{N_i}| = 2$, then letting $j_1, j_2 \in [1 \dd g]$ be such that
    $\Rhs{G}{N_i} = N_{j_1} \cdot N_{j_2}$, we set $\Rhs{G'}{Y_i} = Y_{j_1} \cdot Y_{j_2}$.
  \end{itemize}
  The size of the above grammar is $|G'| = \bigO(k + g + \sigma) = \bigO(|G| + \sigma)$, where
  we used that $k = \Theta(\log n) = \bigO(g)$ (\cref{lm:slp-min-size}).

  To show that $\Lang{G'} = \{\ExtMarkAllChars{T}{\sigma}\}$, first observe that
  $\Exp{G'}{Z} = \zero^{(n-1) \times 1}$ and $\Exp{G'}{Z'} = \zero^{n \times 1}$. Thus,
  for every $a \in [0 \dd \sigma)$, it holds
  $\Exp{G'}{X_a} =
  \zero^{a n \times 1} \hconcat
  \one^{1 \times 1} \hconcat
  \zero^{(n-1) \times 1} \hconcat
  \zero^{(\sigma-a-1) n \times 1} \in
  \BinaryAlphabet^{\sigma n \times 1}$.
  By induction on the length of the expansion, we thus obtain that for every $i \in [1 \dd g]$, it holds
  \[
    \Exp{G'}{Y_i} = \ExtMarkAllChars{\Exp{G}{N_i}}{\sigma} \in \BinaryAlphabet^{\sigma n \times |\Exp{G}{N_i}|}.
  \]
  In particular,
  \begin{align*}
    \Lang{G'}
      &= \{\Exp{G'}{S'}\}
      = \{\Exp{G'}{Y_s}\}\\
      &= \{\ExtMarkAllChars{\Exp{G}{N_s}}{\sigma}\}
      = \{\ExtMarkAllChars{T}{\sigma}\}.
      \qedhere
  \end{align*}
\end{proof}

\begin{proposition}\label{pr:reduce-symbol-occ-to-square-all-zero-small-sigma}
  Assume that, for every 2D SLP $G = (V_{l}, V_{h}, V_{v}, \Sigma, R, S)$,
  where $\Sigma = \BinaryAlphabet$, representing a 2D string $T \in \BinaryAlphabet^{r \times c}$,
  there exists a data structure of size $\bigO(|G| \cdot \log^{\bigO(1)} n)$
  that answers square all-zero queries (\cref{def:2d-integer-problems}) on $T$
  in $\bigO(\log^{\bigO(1)} n)$ time, where $n = \max(r, c)$.
  Then, for every 1D SLP $\hat{G} = (\hat{V}, \hat{\Sigma}, \hat{R}, \hat{S})$
  representing $\hat{T} \in \hat{\Sigma}^{\hat{n}}$, where
  $\hat{\Sigma} = [0 \dd \hat{\sigma})$ and $\hat{\sigma} = \bigO(\min(|\hat{G}|, |\hat{T}|))$,
  there exists a data structure of size $\bigO(|\hat{G}| \cdot \log^{\bigO(1)} \hat{n})$ 
  that answers symbol occurrence queries (\cref{def:symbol-occ})
  in $\bigO(\log^{\bigO(1)} \hat{n})$ time.
\end{proposition}
\begin{proof}

  Let $c > 0$ and $c' > 0$ be constants such that, for every 2D SLP
  $G = (V_{l}, V_{h}, V_{v}, \Sigma, R, S)$ (where $\Sigma = \BinaryAlphabet$)
  representing a 2D string $T \in \BinaryAlphabet^{r \times c}$, the structure from the claim uses
  $\bigO(|G| \cdot \log^{c} n)$ space (where $n = \max(r, c)$), and
  answers square all-zero queries on $T$ in $\bigO(\log^{c'} n)$ time.
  Let $G'$ be a 2D SLP of size $|G'| = \bigO(|\hat{G}| + \hat{\sigma}) = \bigO(|\hat{G}|)$ such that
  $\Lang{G'} = \{\ExtMarkAllChars{\hat{T}}{\hat{\sigma}}\}$ (\cref{def:ext-mark-all-chars}). Such a 2D SLP exists
  by \cref{lm:ext-mark-all-chars-grammar} and \cref{ob:2d-slg-to-2d-slp}.
  Denote $T' = \ExtMarkAllChars{\hat{T}}{\hat{\sigma}}$ and $n' = \max(\Rows{T'}, \Cols{T'})$.
  Note that, since $\Rows{T'} = \hat{n} \cdot \hat{\sigma} = \bigO(\hat{n}^2)$
  and $\Cols{T'} = \hat{n}$, it follows that $n' = \bigO(\hat{n}^2)$.

  The data structure answering symbol occurrence queries on $\hat{T}$ consists of a single component:
  the data structure answering square all-zero queries from the claim applied to $G'$. Note that
  we can apply this structure to $G'$ because $T'$ is over a binary alphabet. The structure
  requires $\bigO(|G'| \cdot \log^{c} n') = \bigO(|\hat{G}| \cdot \log^{c} \hat{n}) =
  \bigO(|\hat{G}| \cdot \log^{\bigO(1)} \hat{n})$ space.

  Let $b,e \in [0 \dd \hat{n}]$ and $c \in [0 \dd \hat{\sigma})$. Using the above structure, we compute
  $\Occurs{\hat{T}}{b}{e}{c}$ as follows:
  \begin{enumerate}
  \item If $b \geq e$, we return that $\Occurs{\hat{T}}{b}{e}{c} = 0$. Assume now that $b < e$.
  \item In $\bigO(\log^{c'} n') = \bigO(\log^{c'} \hat{n})$ time we compute
    $z = \SquareAllZero{T'}{\hat{n} \cdot (c+1)}{e}{e-b}$.
  \item By \cref{lm:reduce-symbol-occ-to-square-all-zero},
    it holds $\Occurs{\hat{T}}{b}{e}{c} = 1 - z$.
    We thus return $1 - z$ as the answer.
  \end{enumerate}
  The query takes $\bigO(\log^{c'} \hat{n}) = \bigO(\log^{\bigO(1)} \hat{n})$ time.
\end{proof}

\begin{proposition}\label{pr:reduce-symbol-occ-to-square-all-zero}
  Assume that, for every 2D SLP $G = (V_{l}, V_{h}, V_{v}, \Sigma, R, S)$,
  where $\Sigma = \BinaryAlphabet$, representing a 2D string $T \in \BinaryAlphabet^{r \times c}$,
  there exists a data structure of size $\bigO(|G| \cdot \log^{\bigO(1)} n)$
  that answers square all-zero queries (\cref{def:2d-integer-problems}) on $T$
  in $\bigO(\log^{\bigO(1)} n)$ time, where $n = \max(r, c)$.
  Then, for every 1D SLP $\hat{G} = (\hat{V}, \hat{\Sigma}, \hat{R}, \hat{S})$
  representing $\hat{T} \in \hat{\Sigma}^{\hat{n}}$, where
  $\hat{\Sigma} = [0 \dd \hat{\sigma})$ and $\hat{\sigma} = \hat{n}^{\bigO(1)}$,
  there exists a data structure of size $\bigO(|\hat{G}| \cdot \log^{\bigO(1)} \hat{n})$ 
  that answers symbol occurrence queries (\cref{def:symbol-occ})
  in $\bigO(\log^{\bigO(1)} \hat{n})$ time.
\end{proposition}
\begin{proof}
  The result follows by combining \cref{pr:reduce-symbol-occ-to-square-all-zero-small-sigma} and
  \cref{pr:slp-rank-and-symbol-occ-alphabet-reduction}.
\end{proof}

\reducefromsymbolocc*
\begin{proof}
  The reduction from symbol occurrence queries on a 1D SLP to square all-zero queries on a 2D SLP follows
  by \cref{pr:reduce-symbol-occ-to-square-all-zero}. The hardness for other queries follows by reducing
  square all-zero queries to other queries as follows:
  \begin{itemize}
  \item square all-zero queries are a special case of all-zero queries,
  \item all-zero queries easily reduce to sum queries,
  \item square all-zero queries reduce to square LCE queries by \cref{pr:reduce-square-all-zero-to-square-lce},
  \item square LCE queries reduce to line LCE queries by \cref{pr:reduce-square-lce-to-line-lce},
  \item line LCE queries reduce to equality queries by \cref{pr:reduce-line-lce-to-equality}.
  \end{itemize}
  Note that all reductions hold for a binary alphabet.
\end{proof}

\clearpage

\appendix

\section{Appendix}\label{sec:app}

\subsection{Pseudo-code from Section~\ref{sec:1d-log}}

\begin{figure}[h!]
  \centering
  \begin{minipage}{.48\linewidth}
    \begin{algorithm}[H]
      \caption{\texttt{LeftMap}$(t,\LevelId,\delta)$}
      \Input{$t \in [1 \dd |V|]$, $\LevelId \in \Zn$, and\\
        $\delta \in [1 \dd |\Exp{G}{N_t}|]$ s.t.\\
        $\delta \leq 2^{\LevelId+1}$.}
      \Output{$t' \in [1 \dd |V|]$ s.t.\\
        $N_{t'} = \LeftMap{G}{N_t}{\LevelId}{\delta}$.}
      $m \gets A_{\rm expsize}[t]$\\
      $k \gets \lceil \delta / 2^{\LevelId} \rceil - 1$\\
      $(b,e) \gets (k \cdot 2^{\LevelId}, \min(m, (k+1) \cdot 2^{\LevelId}))$\\
      $h \gets H^{\rm left}[t,\LevelId,k]$\\
      \If{$e - b = 1$}{
        \Return $(h, 1, \DirLeft)$\\
      }
      \Else{
        $\alpha \gets O^{\rm left}[t,\LevelId,k]$\\
        $(x,y) \gets A_{\rm rhs}[h]$\\
        $\ell \gets A_{\rm expsize}[x]$\\
        \If{$\delta - b \leq \ell - \alpha$}{
          \Return $(x, (\ell {-} \alpha) {-} (\delta {-} b) {+} 1, \DirRight)$\\
        }
        \Else{
          \Return $(y, (\delta {-} b) {-} (\ell {-} \alpha), \DirLeft)$\\
        }
      }
    \end{algorithm}
  \end{minipage}
  \hfill
  \begin{minipage}{.51\linewidth}
    \begin{algorithm}[H]
      \caption{\texttt{RightMap}$(t,\LevelId,\delta)$}
      \Input{$t \in [1 \dd |V|]$, $\LevelId \in \Zn$, and\\
        $\delta \in [1 \dd |\Exp{G}{N_t}|]$ s.t.\\
        $\delta \leq 2^{\LevelId+1}$.}
      \Output{$t' \in [1 \dd |V|]$ s.t.\\
        $N_{t'} = \RightMap{G}{N_t}{\LevelId}{\delta}$.}
      $m \gets A_{\rm expsize}[t]$\\
      $k \gets \lceil \delta / 2^{\LevelId} \rceil - 1$\\
      $(b,e) \gets (k \cdot 2^{\LevelId}, \min(m, (k+1) \cdot 2^{\LevelId}))$\\
      $h \gets H^{\rm right}[t,\LevelId,k]$\\
      \If{$e - b = 1$}{
        \Return $(h, 1, \DirLeft)$\\
      }
      \Else{
        $\beta \,{\gets}\, A_{\rm expsize}[h] {-} (O^{\rm right}[t,\LevelId,k] {+} (e-b))$\\
        $(x,y) \gets A_{\rm rhs}[h]$\\
        $\ell \gets A_{\rm expsize}[y]$\\
        \If{$\delta - b \leq \ell - \beta$}{
          \Return $(y, (\ell {-} \beta) {-} (\delta {-} b) {+} 1, \DirLeft)$\\
        }
        \Else{
          \Return $(x, (\delta {-} b) {-} (\ell {-} \beta), \DirRight)$\\
        }
      }
    \end{algorithm}
  \end{minipage}

  \vspace{1ex}
  \begin{minipage}{.999\linewidth}
    \begin{algorithm}[H]
      \caption{\texttt{RandomAccess}$(i)$}
      \Input{Position $i \in [1 \dd n]$.}
      \Output{The symbol $\Exp{G}{S}[i]$.}
      $(t,\delta,c) \gets (1,i,\DirLeft)$\\
      $\LevelId \gets \lceil \log n \rceil$\\
      \While{$\LevelId \geq 0$}{
        \tcp{Invariant: $\delta \in [1 \dd |\Exp{G}{N_t}|]$, $\delta \leq 2^{\LevelId+1}$, and
          $\Exp{G}{S}[i] = \Access{G}{N_t}{\delta}{c}$}\label{fig:1d-log-access-invariant}
        \If{$c = \DirLeft$}{
          $(t',\delta',c') \gets \texttt{LeftMap}(t,\LevelId,\delta)$\\
        }
        \Else{
          $(t',\delta',c') \gets \texttt{RightMap}(t,\LevelId,\delta)$\\
        }
        \tcp{Invariant: $\delta' \in [1 \dd |\Exp{G}{N_{t'}}|]$, $\delta' \leq 2^{\LevelId}$, and
          $\Access{G}{N_t}{\delta}{c} = \Access{G}{N_{t'}}{\delta'}{c'}$}
        $(t,\delta,c) \gets (t',\delta',c')$\\
        $\LevelId \gets \LevelId-1$\\
      }
      \tcp{Invariant: $\delta = 1$ and $|\Exp{G}{N_t}| = 1$}
      \Return $A_{\rm rhs}[t]$\label{fig:1d-log-access-return}\\
    \end{algorithm}
  \end{minipage}
  \caption{Random access to the text in $\bigO(|G| \log n)$ space.}\label{fig:1d-log-access}
\end{figure}

\clearpage

\subsection{Pseudo-code from Section~\ref{sec:2d-log}}

\begin{figure}[h!]
  \centering
  \begin{minipage}{.99\linewidth}
    \begin{algorithm}[H]
      \caption{\texttt{TopLeftMap}$(t,\LevelId_r,\LevelId_c,\delta_r,\delta_c)$}
      \Input{$t \in [1 \dd |V|]$,
        $\LevelId_r,\LevelId_c \in \Zn$,
        $\delta_r \in [1 \dd \Rows{\Exp{G}{N_t}}]$, and
        $\delta_c \in [1 \dd \Cols{\Exp{G}{N_t}}]$ s.t.\\
        $\delta_r \leq 2^{\LevelId_r + 1}$ and $\delta_c \leq 2^{\LevelId_c + 1}$.}
      \Output{$t' \in [1 \dd |V|]$ s.t.\\
        $N_{t'} = \TopLeftMap{G}{N_t}{\LevelId_r}{\LevelId_c}{\delta_r}{\delta_c}$.}
      $m_r \gets A_{\rm rows}[t]$\\
      $m_c \gets A_{\rm cols}[t]$\\
      $k_r \gets \lceil \delta_r / 2^{\LevelId_r} \rceil - 1$\\
      $k_c \gets \lceil \delta_c / 2^{\LevelId_c} \rceil - 1$\\
      $(b_r,e_r) \gets (k_r \cdot 2^{\LevelId_r}, \min(m_r , (k_r + 1) \cdot 2^{\LevelId_r}))$\\
      $(b_c,e_c) \gets (k_c \cdot 2^{\LevelId_c}, \min(m_c , (k_c + 1) \cdot 2^{\LevelId_c}))$\\
      $h \gets H^{NW}[t,\LevelId_r,\LevelId_c,k_r,k_c]$\\
      \If{$e_r - b_r = 1$ {\bf and} $e_c - b_c = 1$}{
        \Return $(h,1,1,\DirTop,\DirLeft)$ \\	
      }
      \Else{
        $(\alpha_r,\alpha_c) \gets O^{NW}[t,\LevelId_r,\LevelId_c,k_r,k_c]$\\
        $(x,y) \gets A_{\rm rhs}[h]$\\
        \If{$A_{\rm horiz}[h] = 1$}{
          $\ell \gets A_{\rm rows}[x]$\\
          \If{$\delta_r - b_r \leq \ell - \alpha_r$}{
            \Return $(x, (\ell - \alpha_r) - (\delta_r - b_r) + 1 , \alpha_c + (\delta_c - b_c), \DirBottom, \DirLeft)$\\
          }
          \Else{
            \Return $(y, (\delta_r - b_r) - (\ell - \alpha_r), \alpha_c + (\delta_c - b_c), \DirTop, \DirLeft)$\\
          }
        }
        \Else{
          $\ell \gets A_{\rm cols}[x]$\\
          \If{$\delta_c - b_c \leq \ell - \alpha_c$}{
            \Return $(x, \alpha_r + (\delta_r - b_r), (\ell - \alpha_c) - (\delta_c - b_c) + 1, \DirTop, \DirRight)$\\
          }
          \Else{
            \Return $(y, \alpha_r + (\delta_r - b_r), (\delta_c - b_c) - (\ell - \alpha_c), \DirTop, \DirLeft)$\\
          }
        }
      }
    \end{algorithm}
  \end{minipage}
  \caption{The implementation of the $\texttt{TopLeftMap}$ subroutine used
    in the algorithm in \cref{fig:2d-log-access}.}\label{fig:2d-top-left-map}
\end{figure}

\clearpage

\begin{figure}[h!]
  \begin{minipage}{.999\linewidth}
    \begin{algorithm}[H]
      \caption{\texttt{RandomAccess}$(i,j)$}
      \Input{A pair $(i,j) \in [1 \dd \Rows{\Exp{G}{S}}] \times [1 \dd \Cols{\Exp{G}{S}}]$.}
      \Output{The symbol $\Exp{G}{S}[i,j]$.}
      $(t,\delta_r,\delta_c,c_r,c_c) \gets (1,i,j,\DirTop,\DirLeft)$\label{lineone}\\
      $(\LevelId_r,\LevelId_c) \gets (\lceil \log A_{\rm rows}[1] \rceil,\lceil \log A_{\rm cols}[1] \rceil)$\\
      \While{$\LevelId_r > 0$ {\bf or} $\LevelId_c > 0$}{\label{fig:2d-log-access-while}
        \tcp{Invariant: $\delta_r \in [1 \dd \Rows{\Exp{G}{N_t}}]$,
                        $\delta_c \in [1 \dd \Cols{\Exp{G}{N_t}}]$,
                        $\delta_r \leq 2^{\LevelId_r+1}$,
                        $\delta_c \leq 2^{\LevelId_c + 1}$, and
                        $\Exp{G}{S}[i,j] = \AccessTwoDim{G}{N_t}{\delta_r}{\delta_c}{c_r}{c_c}$}\label{linethree}
        \If{$(c_r,c_c) = (\DirTop,\DirLeft)$}{
          $(t', \delta'_r, \delta'_c , c'_r , c'_c) \gets
            \texttt{TopLeftMap}(t, \LevelId_r, \LevelId_c, \delta_r, \delta_c)$\label{linefive}\\
        }
        \ElseIf{$(c_r,c_c) = (\DirTop,\DirRight)$}{
          $(t', \delta'_r, \delta'_c , c'_r , c'_c) \gets
            \texttt{TopRightMap}(t, \LevelId_r, \LevelId_c, \delta_r, \delta_c)$\label{lineseven}\\
        }
        \ElseIf{$(c_r,c_c) = (\DirBottom,\DirLeft)$}{
          $(t', \delta'_r, \delta'_c , c'_r , c'_c) \gets
            \texttt{BottomLeftMap}(t, \LevelId_r, \LevelId_c, \delta_r, \delta_c)$\label{linenine}\\
        }
        \Else{
          $(t', \delta'_r, \delta'_c , c'_r , c'_c) \gets
            \texttt{BottomRightMap}(t, \LevelId_r, \LevelId_c, \delta_r, \delta_c)$\label{lineeleven}\\
        }
        \tcp{Invariant: $\delta'_r \in [1 \dd \Rows{\Exp{G}{N_{t'}}}]$,
                        $\delta'_c \in [1 \dd \Cols{\Exp{G}{N_{t'}}}]$, and
                        $\AccessTwoDim{G}{N_{t}}{\delta_r}{\delta_c}{c_r}{c_c} =
                        \AccessTwoDim{G}{N_{t'}}{\delta'_r}{\delta'_c}{c'_r}{c'_c}$}\label{lineelevenandhalf}
        $(t, \delta_r, \delta_c, c_r, c_c) \gets (t', \delta'_r, \delta'_c , c'_r , c'_c)$\\
        \While{$2^{\LevelId_r} > A_{\rm rows}[t]$}{
          $\LevelId_r \gets \LevelId_r - 1$\label{linefourteen}
        }
        \While{$2^{\LevelId_c} > A_{\rm cols}[t]$}{
          $\LevelId_c \gets \LevelId_c - 1$\label{linesixteen}
        }
        \tcp{Either $\LevelId_r$ or $\LevelId_c$ was reduced by at least one in the two loops above}
      }
      \tcp{Invariant: $\delta_r = \delta_c = 1$ and $N_t \in V_{l}$}
      \Return $A_{\rm rhs}[t]$
    \end{algorithm}
  \end{minipage}
  \caption{Random access to the 2D string in $\bigO(|G| \log^2 n)$
    space.}\label{fig:2d-log-access}
\end{figure}

\bibliographystyle{alphaurl}
\bibliography{paper}

\end{document}